%% file: ALENEX27_full-version.tex
\pdfoutput=1
\documentclass[twoside,leqno,twocolumn]{article}
\usepackage[letterpaper]{geometry}   
\usepackage{siamproceedings}

\usepackage[utf8]{inputenc}
\usepackage[T1]{fontenc}
\usepackage{amssymb,amsfonts}
\usepackage{mathtools}
\usepackage{bm}
\usepackage{xspace}
\usepackage{microtype}
\usepackage[noadjust]{cite}

\usepackage{algpseudocodex}
\algrenewcommand\algorithmicrequire{\textbf{Input:}}
\algrenewcommand\algorithmicensure{\textbf{Output:}}
\algrenewcommand\algorithmicthen{}
\algrenewcommand\algorithmicdo{}
\algnewcommand{\IfInline}[2]{\textbf{if} #1 \textbf{then} #2}

\usepackage[dvipsnames]{xcolor}
\definecolor{VUBlue}{rgb}{0,0.466,0.702}
\definecolor{VUOrange}{rgb}{0.8,0.254,0}
\newcommand{\algstep}[1]{%
  \State\vspace{2pt}%
  \textcolor{VUBlue}{{\bfseries$\blacktriangleright$}\ #1}%
  \vspace{1pt}%
}

\usepackage{booktabs}
\usepackage{tabularx}
\usepackage{makecell}
\usepackage{multirow}

\usepackage{pgf}
\usepackage{caption}
\usepackage{subcaption}

\usepackage[group-separator={,}]{siunitx}

\usepackage{enumitem}
\setlist[1]{labelindent=\parindent}
\setlist[enumerate,1]{label=(\arabic*),noitemsep}
\setlist[itemize]{noitemsep}
\setlist[description]{noitemsep}

\usepackage{listings}
\usepackage{dirtytalk}
\usepackage[normalem]{ulem}
\usepackage{cryptocode}
\usepackage{placeins}
\usepackage{comment}
\usepackage{marginnote}
\usepackage[skins]{tcolorbox}
\usepackage{appendix}
\newsiamthm{conjecture}{Conjecture}
\newsiamthm{observation}{Observation}
\newsiamremark{remark}{Remark}

\crefname{hypothesis}{Hypothesis}{Hypotheses}
\crefname{enumi}{condition}{conditions}
\newcommand{\lineref}[1]{Line~\ref{#1}}

\ifdefined\N\else\DeclareMathOperator{\N}{\mathbb{N}}\fi
\ifdefined\Z\else\DeclareMathOperator{\Z}{\mathbb{Z}}\fi
\ifdefined\R\else\DeclareMathOperator{\R}{\mathbb{R}}\fi

\newcommand{\cO}{\mathcal{O}}
\newcommand{\cOstar}{\mathcal{O}^{*}}

\newcommand{\poly}{\mathrm{poly}}

\newcommand{\setsize}[1]{\left|#1\right|}

\newcommand{\occ}{\mathsf{occ}}
\newcommand{\Exact}{\text{Exact}}
\newcommand{\MSlider}{\text{M-SLIDER}}
\newcommand{\SEQSlider}{\text{SEQ-SLIDER}}

\newcommand{\key}{\mathrm{key}}
\newcommand{\rank}{\mathrm{rank}}
\newcommand{\ed}{\mathrm{ed}}

\newif\ifpaper
\definecolor[named]{benGray}{rgb}{0.61,0.61,0.63}
\definecolor[named]{maybeDone}{cmyk}{0.45,0,0.50,0.32}

\newcommand{\fpt}{\ensuremath{\mathsf{FPT}}\xspace}

\newcommand{\letx}{\mathord{\diamond}}
\newcommand{\id}{\rank}
\newcommand{\CMM}{CMM\xspace}
\newcommand{\ldCMM}{\ensuremath{(\ell,d)}-\CMM}

\newcommand{\MaxIP}{MaxIP\xspace}
\newcommand{\SETH}{SETH\xspace}
\def\dd{\mathinner{.\,.}}
\newcommand{\chisup}{f_{\chi^2}}

\newcommand{\defproblem}[3]{
\vspace{2mm}
\noindent\fbox{\begin{minipage}{0.96\columnwidth}
   \textsc{#1}\\ {\bf{Input:}} #2 \\ {\bf{Output:}} #3
\end{minipage}}\vspace{2mm}}

\newcommand{\ie}[1]{(i.e.,~#1)}

\newif\ifhidetodos
\hidetodostrue
\newcommand{\cbox}[2][yellow]{\fcolorbox{#1}{white}{\parbox{\dimexpr\linewidth-2\fboxsep}{\strut #2\strut}}}
\newcommand{\customtodo}[3]{\textcolor{#2}{\(\blacktriangledown\)}\marginnote{\raggedright \textcolor{#2}{\textbf{#1:} #3}}}
\newcommand{\customtododisplay}[3]{\noindent\cbox[#2]{\textcolor{#2}{\textbf{#1:} #3}}}
\newcommand{\customtodoinline}[3]{\textcolor{#2}{\textcolor{#2}{\textbf{#1:} #3}}}
\ifhidetodos
    \renewcommand{\customtodo}[3]{}
    \renewcommand{\customtododisplay}[3]{}
    \renewcommand{\customtodoinline}[3]{}
\fi

\usepackage{soul}

\newsavebox{\ulbox}
\newlength{\ulwidth}
\newcommand{\colunderline}[2]{%
  \sbox{\ulbox}{$#2$}%
  \settowidth{\ulwidth}{\usebox{\ulbox}}%
  \makebox[0pt][l]{\color{#1}\rule[-2.2pt]{\ulwidth}{0.6pt}}%
  \usebox{\ulbox}%
}
\newcommand{\vb}[1]{\textcolor{VUBlue}{#1}}
\newcommand{\vo}[1]{\textcolor{VUOrange}{#1}}

\begin{document}

\title{\Large Mining Correlated Motifs with Wildcards in PPI Networks: Exact and Optimal}
\author{Lucas P. Oliveira\thanks{CWI, Amsterdam, The Netherlands and University of Campinas, Campinas, Brazil
    (\email{lucas.peres@students.ic.unicamp.br}).}
    \and Ben Bals\thanks{CWI, Amsterdam, The Netherlands and Vrije Universiteit Amsterdam, Amsterdam, The Netherlands (\email{bjjb@cwi.nl}).}
    \and Aalt D. J. van Dijk\thanks{University of Amsterdam, Amsterdam, The Netherlands (\email{a.d.j.vandijk@uva.nl}).}
    \and Solon P. Pissis\thanks{The Cyprus Institute, Nicosia, Cyprus
    (\email{s.pissis@cyi.ac.cy}).}}

\date{}

\maketitle

\begin{abstract}
A \emph{protein--protein interaction} (PPI) network is a graph whose nodes are protein sequences and whose edges represent physical interactions between proteins.
\emph{Correlated Motif Mining} (CMM) seeks pairs of sequence motifs whose co-occurrence in interacting proteins reveals binding sites and interaction mechanisms.
Formally, in the $(\ell,d)$-\CMM problem, we are given a graph $G=(V,E)$, with $V\subseteq\Sigma^*$ for an alphabet $\Sigma$, and integers $0\leq d<\ell$, and we seek a pair $\{X,Y\}$ of motifs, each of length $\ell$ and containing $d$ wildcards, that maximizes a support measure $f(G,\{X,Y\})$; e.g., the count $f_E$ of edges $\{u,v\}\in E$ such that $X$ matches a substring of $u$ and $Y$ matches a substring of $v$, or vice versa.
Although existing approaches have produced biologically meaningful results, they are either inefficient or heuristic-based, as the underlying optimization problem is NP-hard.
Here, we revisit $(\ell,d)$-\CMM.
For $d=0$, we give an exact $\mathcal{O}(N^2)$-time and $\mathcal{O}(N)$-space algorithm, where $N$ is the total length of the node strings, and prove that this time is conditionally optimal under the Strong Exponential-Time Hypothesis.
For any $d>0$, we give an exact $\mathcal{O}\!\left(\binom{\ell}{d}^2N^2\right)$-time and $\mathcal{O}(N)$-space algorithm, and observe that the problem is W[1]-hard parameterized by $d$. The latter algorithm improves the state of the art by exponentially large factors in both time and space.
We further demonstrate experimentally that our algorithms are practical on real-world PPI networks and find significantly higher-scoring motif pairs than the state of the art.
Moreover, we show that our algorithms find many top-scoring motif pairs that the state of the art misses entirely. Taken together, our results suggest that exact \CMM with wildcards is more tractable in practice than previously assumed.
\end{abstract}

\section{Introduction.}\label{sec:intro}
Large-scale biological networks capturing protein--protein interactions (PPIs) are now available for a wide range of organisms~\cite{luck2020reference,oughtred2021biogrid,panneerselvam2024intact}.
While such data reveal how proteins function within a PPI network, they offer no direct insight into how those interactions are encoded at the sequence level---in particular, which regions of protein sequences correspond to physical interaction sites. Popular tools such as AlphaFold~\cite{jumper2021alphafold} have revolutionized protein structure prediction, but experimental validation of interaction interfaces at scale remains a major bottleneck: the human proteome alone contains tens of thousands of binary interactions~\cite{luck2020reference}, and determining their structural basis experimentally is infeasible at this scale.
Moreover, even large-scale structural characterization of interacting proteins would not directly answer the question of how interaction specificity is encoded in protein sequence or structure. In other words: \emph{which specific residues determine which partners a protein can bind?}

To answer this question, a variety of bioinformatics approaches have been developed~\cite{idrees2024uncovering,soleymani2022protein,jia2024comprehensive}. In particular, several algorithmic approaches have been proposed for locating binding sites by mining overrepresented pairs of sequence patterns, called \emph{motif pairs}, from sequences of interacting proteins~\cite{DBLP:journals/bioinformatics/LiLW06,tan2006,DBLP:journals/tkde/LiLLW07,leung2007,DBLP:conf/sdm/LiSLW08,DBLP:journals/tcbb/BoyenDNHD11}.
One such approach, known as \emph{Correlated Motif Mining} (CMM), identifies binding sites by searching for a consensus pattern in one group of proteins that interacts with nearly all proteins containing a second consensus pattern.
When such a pair is found, both patterns are likely to represent complementary surface regions that enable physical binding.
\Cref{fig:motifpairs} illustrates this idea: $\{\texttt{G}\letx\letx\texttt{P}\letx\texttt{NY}, \texttt{P}\letx\letx\texttt{P}\}$ and $\{\texttt{R}\letx\texttt{R}\letx\letx\texttt{pS}, \texttt{L}\letx\texttt{R}\letx\letx\texttt{pS}\}$ are two correlated motif pairs, where an undirected edge connects two protein sequences if the first contains motif $\texttt{G}\letx\letx\texttt{P}\letx\texttt{NY}$ and the second contains motif $\texttt{P}\letx\letx\texttt{P}$, and analogously for $\{\texttt{R}\letx\texttt{R}\letx\letx\texttt{pS}, \texttt{L}\letx\texttt{R}\letx\letx\texttt{pS}\}$.
Although it remains unclear what fraction of protein interfaces can be described by correlated motifs of this kind, existing approaches consistently show that such motifs do contain meaningful information about interaction interfaces.

Existing computational approaches fall into two broad categories: (i) interaction-driven approaches~\cite{DBLP:journals/bioinformatics/LiLW06,DBLP:journals/tkde/LiLLW07,DBLP:conf/sdm/LiSLW08}; and (ii) motif-driven approaches~\cite{tan2006,leung2007,DBLP:journals/tcbb/BoyenDNHD11}.
Interaction-driven approaches search for (quasi-)bicliques---that is, subsets of nodes in a PPI network, in which nearly every node on one side is connected to nearly all nodes on the other.
These subgraphs capture an all-versus-all, or most-versus-most, interaction pattern, from which a representative motif pair is then derived from the sequences labeling the nodes.
Motif-driven approaches, by contrast, begin with candidate motif pairs and evaluate their support directly in the PPI network.
While both paradigms have been shown to yield biologically meaningful results, the motif-driven approach has several conceptual advantages over its interaction-driven counterpart.
First, motif pairs are identified directly rather than derived in post-processing.
Second, all proteins containing a given motif are taken into account, rather than an arbitrarily selected subset.
Third, when the interactions between two groups of proteins arise from multiple compatible binding sites---e.g., $\{\texttt{G}\letx\letx\texttt{P}\letx\texttt{NY}, \texttt{P}\letx\letx\texttt{P}\}$ and $\{\texttt{R}\letx\texttt{R}\letx\letx\texttt{pS}, \texttt{L}\letx\texttt{R}\letx\letx\texttt{pS}\}$ in \Cref{fig:motifpairs}---interaction-driven methods are forced to merge them into one motif pair, thereby losing this structural distinction.

\begin{figure}[htbp]
\centering
\begin{tikzpicture}[
  scale=0.875,
  transform shape,
  font=\small,
  arr/.style={draw, rounded corners=1pt, minimum width=3.1cm, minimum height=0.95cm, inner sep=2pt, align=center},
  edge/.style={ultra thick},
  edgeA/.style={ultra thick, draw=VUBlue},
  edgeB/.style={ultra thick, draw=VUOrange}
]

\node[arr] (u1) at (0,3.2) {$\colunderline{VUBlue}{\vb{\texttt{G}}\texttt{RS} \vb{\texttt{P}} \texttt{P} \vb{\texttt{NY}}} \dots \colunderline{VUOrange}{\vo{\texttt{R}}\texttt{N}\vo{\texttt{R}}\texttt{GS}\vo{\texttt{p}\texttt{S}}}$};
\node[arr] (u2) at (0,2.0) {$\dots \dots \colunderline{VUOrange}{\vo{\texttt{R}}\texttt{R}\vo{\texttt{R}}\texttt{pP}\vo{\texttt{pS}}} \dots \dots$};
\node[arr] (u3) at (0,0.8) {$\dots \colunderline{VUBlue}{\vb{\texttt{G}}\texttt{RR}\vb{\texttt{P}}\texttt{S}\vb{\texttt{NY}}} \dots \dots \dots$};

\node[arr] (v1) at (6.0,3.2) {$\dots \colunderline{VUBlue}{\vb{\texttt{P}}\texttt{RS}\vb{\texttt{P}}} \dots \dots \dots$};
\node[arr] (v2) at (6.0,2.0) {$\dots \dots \colunderline{VUOrange}{\vo{\texttt{L}}\texttt{G}\vo{\texttt{R}}\texttt{PS}\vo{\texttt{pS}}} \dots \dots$};
\node[arr] (v3) at (6.0,0.8) {$\dots\colunderline{VUBlue}{\vb{\texttt{P}}\texttt{GP}\vb{\texttt{P}}} \dots \dots$};

\draw[edgeA] (u1.east) -- (v1.west);
\draw[edgeB] (u1.east) -- (v2.west);
\draw[edgeA] (u1.east) -- (v3.west);

\draw[edgeB] (u2.east) -- (v2.west);
\draw[edge] (u2.east) -- (v3.west);

\draw[edgeA] (u3.east) -- (v1.west);
\draw[edgeA] (u3.east) -- (v3.west);

\end{tikzpicture}
\caption{Motif pairs $\textcolor{VUBlue}{\{\texttt{G}\letx\letx\texttt{P}\letx\texttt{NY}, \texttt{P}\letx\letx\texttt{P}\}}$ and $\textcolor{VUOrange}{\{\texttt{R}\letx\texttt{R}\letx\letx\texttt{pS}, \texttt{L}\letx\texttt{R}\letx\letx\texttt{pS}\}}$ as correlated motifs.
The colored edges show their correlation.}
\label{fig:motifpairs}
\end{figure}

\subsection{The problem and the state of the art.}

In this paper, we focus on the motif-driven approach to CMM.
The vanilla version of the problem is defined as follows.

\defproblem{Correlated Motif Mining (\CMM)}
{A graph $G=(V,E)$, with $V\subseteq\Sigma^*$ for an alphabet $\Sigma$, and an integer $\ell>0$.}
{A pair $\{X,Y\}$, with $X,Y\in\Sigma^{\ell}$, that maximizes a support measure $f(G,\{X,Y\})$.}

We next define the \CMM problem \emph{with wildcards}; wildcards are placeholder symbols, denoted by $\letx$, that can match any character from the underlying alphabet.

\defproblem{\CMM with Wildcards ($(\ell,d)$-\CMM)}
{A graph $G=(V,E)$, with $V\subseteq\Sigma^*$ for an alphabet $\Sigma$, and integers $\ell>0$ and $0\le d<\ell$.}
{A pair $\{X,Y\}$, with $X,Y\in(\Sigma\sqcup\{\letx\})^{\ell}$, each containing exactly $d$ wildcards, that maximizes a support measure $f(G,\{X,Y\})$.}

While we formalize the problems in a general form, the motivating application leads us to use the following widely-used support measures~\cite{tan2006,leung2007,DBLP:journals/tcbb/BoyenDNHD11}: the count $f_E$ of edges $\{u,v\}\in E$ such that $X$ matches a substring of $u$ and $Y$ matches a substring of $v$, or vice versa; and the $\chi^2$-score $f_{\chi^2}$ of $G$ with respect to $\{X,Y\}$~\cite{tan2006}.

To the best of our knowledge, three techniques have so far been proposed. Unfortunately, these techniques differ not only in their mining method, but also in
the matching relation and the support measure used.

The first work by Tan et al.~\cite{tan2006} introduced D-MOTIF, an exact algorithm for
finding correlated motif pairs under the Hamming distance.
D-MOTIF considers $f_{\chi^2}$ to quantify support.
The basic idea of the algorithm is to enumerate all
possible motif pairs from every pair $\{u,v\}\in E$,
using the \emph{Hamming distance neighborhood} notion,
and then check if they have enough support.
The authors also suggested a heuristic counterpart called D-STAR.
Unfortunately, neither D-MOTIF nor D-STAR scales to networks with more than 250 proteins.

The second work by Leung et al.~\cite{leung2007}
introduced FindMotif, an exact algorithm for \ldCMM.
FindMotif uses a different support measure that considers
the probability that the observed number of interactions is generated under a
null hypothesis. The algorithm builds a table of $\cO(\binom{\ell}{d}^2|\Sigma|^{2(\ell-d)})$ entries, one per possible motif pair over alphabet $\Sigma$,
and computes the support for each using a \emph{voting} technique.
The authors also suggested a heuristic counterpart called MotifHeuristics,
which was able to handle larger networks than D-STAR.

The third work by Boyen et al.~\cite{DBLP:journals/tcbb/BoyenDNHD11} proved that $(\ell,d)$-\CMM is NP-hard for so-called \emph{compliant biclique-maximal} measures, which include $f_E$ and $f_{\chi^2}$.
They also introduced SLIDER, a generic metaheuristic that incorporates two \emph{steepest-ascent} procedures. The central component of these procedures is a neighborhood function defined by viewing a motif as a window that slides along the protein sequence of one of the proteins.
Unlike more conventional neighborhood functions, this has a clear biological interpretation: if a motif partially overlaps a binding site in a sequence, a few mutations should suffice to convert it into a binding site.
In this sense, both neighborhood functions are designed to identify nearby motifs that may correspond to regions close to one another on actual proteins.
The difference is that one neighborhood function considers as neighbors all motifs that could, in principle, lie near the given motif on any protein, whereas the other restricts attention to motifs that are actually nearby on a single selected protein.
Although SLIDER can be paired with an arbitrary support measure, the authors used $f_{\chi^2}$, providing compelling evidence that it is biologically the most relevant,
and showing that SLIDER outperforms MotifHeuristics.

\subsection{Our contributions.}
Since the motif-driven approach to CMM is computationally hard, in general~\cite{DBLP:journals/tcbb/BoyenDNHD11},
it is widely believed that \emph{exact optimization} is impractical for realistically sized PPI networks.
Consequently, the main aim of the existing works~\cite{tan2006,leung2007,DBLP:journals/tcbb/BoyenDNHD11} has been to design efficient and effective \emph{heuristics} that identify high-quality motif pairs in reasonable time.
This need is particularly acute for modern real-world datasets, which are large, noisy, and in need of methods that scale while still capturing biologically meaningful patterns.

Here, however, we present an \emph{exact} method that remains efficient in practice on modern datasets, thereby showing that the problem can be solved far more effectively than is commonly believed.
Our results thus refute the prevailing view that only heuristic approaches are viable for this task. Our crucial observation is that the hardness proof in~\cite{DBLP:journals/tcbb/BoyenDNHD11} proceeds under the assumption that both $\ell$ and $d$ are unbounded, which is not realistic in the motivating application.
Intuitively, we confine the computational hardness to \emph{parameter} $\ell$.
Indeed, the values considered for $\ell$ and $d$ are typically  small~\cite{tan2006,leung2007,DBLP:journals/tcbb/BoyenDNHD11}.

We make the following theoretical contributions:

\begin{enumerate}
    \item We give an \emph{exact} $\cO(N^2)$-time algorithm for \CMM (the vanilla version), where $N\coloneqq\sum_{v \in V}|v|$. The space usage is $\cO(N)$. See \Cref{sec:exact}.
    \item We show that this algorithm is essentially \emph{optimal}: no algorithm can solve \CMM in time $\cO(N^{2-\varepsilon} \poly(|E|, |V|))$, for any $\varepsilon > 0$, unless the popular Strong Exponential-Time Hypothesis (\SETH)~\cite{DBLP:journals/jcss/ImpagliazzoPZ01} is false. See \Cref{sec:exact}.
    \item We build on the \CMM algorithm to give an \emph{exact} $\cO\!\left(\binom{\ell}{d}^2N^2\right)$-time algorithm for $(\ell,d)$-\CMM (the wildcards version). The space usage is $\cO(N)$; and the algorithm is embarrassingly parallelizable.
    This improves dramatically on the state of the art from $\cOstar\!\left(\binom{\ell}{d}^2\setsize{\Sigma}^{2(\ell-d)}\right)$ time to $\cOstar\!\left(\binom{\ell}{d}^2\right)$ and from
    $\cOstar\!\left(\binom{\ell}{d}^2\setsize{\Sigma}^{2(\ell-d)}\right)$ space to $\cO(N)$.\footnote{The $\cOstar(\cdot)$ notation suppresses polynomial factors.}
     See \Cref{sec:wild}.

    \item  Finally, we observe that $(\ell,d)$-\CMM is W[1]-hard parameterized by $d$. See \Cref{sec:wild}.
\end{enumerate}

We make the following practical contributions:

\begin{enumerate}
    \item We show that the parallelized version of our algorithm for $(\ell,d)$-\CMM processes real-world datasets in a reasonable amount of time.
    \item We show that the runtime and space usage predicted by our theoretical results are accurately reflected in practice.
    \item We show that our exact algorithm finds motif pairs with significantly higher support than the state of the art, which highlights its usefulness.
    \item We show that our exact algorithm finds many top-scoring motif pairs that the state of the art misses.
\end{enumerate}

For the above practical contributions, see \Cref{sec:exp}. We begin with some preliminaries in \Cref{sec:prel}.

\section{Preliminaries.}\label{sec:prel}

We assume throughout that arrays and strings start at position~$1$.

\subsection{Strings.} A \emph{string} is a finite sequence of characters over an alphabet $\Sigma$.
In this work, we consider strings over a totally ordered alphabet.
We denote the length of a string $T$ by $|T|$.
For any index $i \in [1, |T|]$, we denote the $i$-th character of $T$ by $T[i]$.
A string of the form $T[i\dd j] = T[i]T[i+1]\cdots T[j]$,
where $1 \le i \le j \le |T|$, is called a \emph{fragment} of $T$.
If $T$ is known, we can encode $T[i\dd j]$ in $\cO(1)$ space as the pair $(i,j)$.
A fragment of the form $T[1\dd j]$ is called a \emph{prefix} of $T$, and a fragment of the form $T[i\dd |T|]$ is called a \emph{suffix} of $T$.
 A string $P$ is a \emph{substring} of $T$
 if $P=T[i\dd i+|P|-1]$, for some $i$.
 In this case, we say that $P$ \emph{occurs} in $T$.
 A \emph{wildcard} is a special symbol $\letx$
 matching any $a\in\Sigma\sqcup\{\letx\}$.

Given a string $T$ of length $n$, the \emph{suffix array} of $T$, denoted by $\textsf{SA}_T[1 \dd n]$, stores a permutation of the indices $\{1, \dots, n\}$ such that $T[\textsf{SA}_T[i] \dd n]$ is the $i$-th lexicographically smallest suffix of~$T$.
Let $\operatorname{lcp}(X,Y)$ denote the length of the \emph{longest common prefix} (LCP) of the strings $X$ and $Y$.
The \emph{LCP array} of $T$, denoted by $\textsf{LCP}_T[1\dd n]$, is defined as follows: $\textsf{LCP}_T[1] = 0$, and for every $i \in [2, n]$, \(
\textsf{LCP}_T[i] = \operatorname{lcp}\bigl(T[\textsf{SA}_T[i-1]\dd n],\, T[\textsf{SA}_T[i]\dd n]\bigr).
\)

\subsection{Graphs and motifs.}
Let $G=(V,E)$ be an undirected graph, where $V \subseteq \Sigma^*$, for an alphabet~$\Sigma$.
Let $N \coloneqq \sum_{v\in V} |v|$ denote the total length of all $v\in V$.
An \textit{$(\ell,d)$-motif} is a string $X \in (\Sigma \sqcup \{\letx\})^{\ell}$ containing exactly $d$ wildcard symbols~$\letx$.
When $d = 0$, the motif is simply a string over~$\Sigma$ of length~$\ell$.
An $(\ell,d)$-\textit{motif pair} $\{X, Y\}$ is an unordered pair of $(\ell,d)$-motifs.
We may omit $\ell$ and $d$ when they are clear from the context and then simply speak of motifs and motif pairs.

\subsection{Motif-induced subgraphs.}
For a motif~$X$, let $V_X \coloneqq \{\,v \in V \mid \text{$X$ occurs in } v\,\}$ be the set of nodes where~$X$ occurs. If a node $v$ contains multiple occurrences of~$X$, then $v$ is included only once in~$V_X$.

For a motif pair $\{X,Y\}$, the set of edges whose endpoints contain $X$ and~$Y$, respectively, is
\[
E_{X,Y} \coloneqq \{\,\{u,v\} \in E \mid u \in V_X \text{ and } v \in V_Y\,\}.
\]

We set $f_E(G, \{X,Y\})  \coloneqq\setsize{E_{X,Y}}$.
The subgraph \textit{selected} by $\{X,Y\}$ is then $G_{X,Y} \coloneqq (V_X \cup V_Y, E_{X,Y})$.
If $V_X \cap V_Y$ is nonempty, we may have an edge between distinct nodes $v, u \in V_X$ provided that at least one of them is in $V_X \cap V_Y$.
On the other hand, if $v,u \in V_X \setminus V_Y$, they cannot be adjacent by definition of~$E_{X,Y}$.

A motif pair $\{X,Y\}$ is called a $(k_X,k_Y,k_{X,Y})$\textit{-motif pair} for~$G$ if $G_{X,Y}$ satisfies $|V_X| = k_X$, $|V_Y| = k_Y$, and $|V_X \cap V_Y| = k_{X,Y}$.
It is called \textit{complete} if every node of $V_X$ is adjacent to every node of $V_Y$.
A complete bipartite graph is called a \textit{biclique}.

\subsection{Support measures.}
A \textit{support measure} $f$ on~$G$ is a function that maps an $(\ell,d)$-motif pair $\{X,Y\}$ to a nonnegative real number $f(G,\{X,Y\})$.
Different types of support functions have been considered in the literature~\cite{tan2006,leung2007,DBLP:journals/tcbb/BoyenDNHD11}.
We restrict our attention to $f_E$ and $\chisup$, which we define below.
The former is simple and illustrative and will serve to make our theoretical analysis lucid.
The second works well in practice \cite{DBLP:journals/tcbb/BoyenDNHD11} and we use it in our experimental evaluation (we also show that our theoretical results extend to it).

\subsection{The $\chisup$ support measure.}
The $\chi^2$-based support measure of a graph~$G$ with respect to an $(\ell,d)$-motif pair $\{X,Y\}$, denoted by $\chisup$, is defined as
{\small\[
\chisup(G,\{X,Y\}) \coloneqq
\begin{cases}
    \dfrac{(|E_{X,Y}| - \overline{E_{X,Y}})^2}{\overline{E_{X,Y}}} & \text{if $|E_{X,Y}| > \overline{E_{X,Y}}$}, \\[8pt]
    0 & \text{otherwise},
\end{cases}
\]}
where $\overline{E_{X,Y}}$ is the expected number of edges between $V_X$ and~$V_Y$ under a uniform random-graph model.
Concretely,
\[
\overline{E_{X,Y}} \coloneqq \ed(G) \cdot E^{\max}_{k_X,k_Y,k_{X,Y}},
\]
where $\ed(G) \coloneqq |E| / \binom{|V|}{2}$ is the \textit{edge density} of~$G$, and
\[
E^{\max}_{k_X,k_Y,k_{X,Y}} \coloneqq k_X k_Y - \binom{k_{X,Y}}{2} - k_{X,Y}
\]
is the maximum possible number of edges in the subgraph selected by any $(k_X,k_Y,k_{X,Y})$-motif pair.
The parameters $k_X$, $k_Y$, and~$k_{X,Y}$ are as defined above; $|V|$ and $|E|$ are trivially obtained from the input.
Intuitively, this measure captures how much above the average density a subgraph selected by a motif pair is, taking into account its size.

\section{An exact algorithm for \CMM.}\label{sec:exact}
To simplify our presentation, we focus on the $f_E$ support measure for now, but the same algorithmic ideas also apply to the $\chisup$ support measure.
Recall that $f_E(G, \{X,Y\})$ counts the edges $\{u,v\}\in E$ such that $X$ matches a substring of $u$ and $Y$ matches a substring of $v$, or vice versa.

\begin{theorem}\label{thm:basic}
There is an algorithm for \CMM with $f_E$ running in time $\cO(N^2)$, where $N\coloneqq\sum_{v \in V}|v|$.
The space used is $\cO(N)$.
\end{theorem}
\begin{proof}
The algorithm (see \Cref{alg:basic}) iterates over all motif pairs $\{X,Y\}$, such that $X$ and $Y$ are substrings of at least one of the labels, and then counts the corresponding number of edges.
That is, it computes the value $f_E(G, \{X,Y\})$ explicitly.
Therefore, taking the maximum will clearly produce the desired output.

The more interesting part is how to implement some of the operations efficiently enough.
The short version is that the suffix array (SA) will give us all the necessary tools.
First, observe that we can construct the SA of the concatenation of the node strings (with some separation character) in $\cO(N)$ time~\cite{SA} after sorting all characters occurring in
all strings $v\in V$ in $\cO(N\log N)$ time.
The SA (together with the LCP array) gives us a constant-size unique name for each length-$\ell$ substring of any label that can be computed (after preprocessing) in constant time given a location in a node string.
Call this the \emph{rank} of a substring.
Note that the ranks are the same for all occurrences of the same substring.
In addition, these are integers bounded by $N$.
One may think of this as a perfect hash function that we can use to (1) compare substrings in constant time, and (2) implement the dictionary $C$ defined in \lineref{alg:basic:dict} using $\cO(N)$ space.
Using (1) and (2), we implement the loop in \lineref{alg:basic:outer-loop} so that it runs in $\cO(N)$ time.
In particular, after a single linear pass over the node strings, we use these names to build a simple inverse index to access the nodes that contain a given rank, as required in \lineref{alg:basic:index}.
As we access the nodes for each label exactly once and this inverse index has total size $\cO(N)$, that is also the total time to access the nodes.
Similarly, observe that in a linear pass, we can build a list of the unique ranks that occur in each node, making the loop in \lineref{alg:basic:inner-loop} take $\cO(N)$ time in total for each $X$.
The operations in the innermost loop require constant time.
In particular, the operation $\preceq$ can be any total order on the length-$\ell$ substrings.
To get a constant-time comparator, we choose the rank order.

Thus, we are finally left to argue why the loop in \lineref{alg:basic:edge-loop} does not blow up the running time.
Observe that each edge $\{u,v\}$ is only considered for a motif pair $\{X,Y\}$ if $X$ occurs in $u$ and $Y$ occurs in $v$. (Each edge is visited exactly once in Line~\ref{alg:basic:edge-loop}, oriented so that
$u\in V_X$; the second loop adds the $Y$'s that occur only in $u$.)
This can happen at most $(\setsize{u}-\ell+1) (\setsize{v}-\ell+1)$ times. Then
\begin{align*}
    &\sum_{\{u,v\} \in E} (\setsize{u}-\ell+1) (\setsize{v}-\ell+1)\\
    \le &\sum_{(u,v) \in V \times V} \setsize{u} \setsize{v}
    = \left(\sum_{u \in V} \setsize{u}\right) \left(\sum_{v \in V} \setsize{v}\right) = \cO(N^2),
\end{align*}
yields the desired bound. The space used is $\cO(N)$ as SA and $C$ are implemented in $\cO(N)$ space.
\end{proof}

\begin{algorithm}[t]
\caption{Basic algorithm for \CMM with $f_E$. Some operations require care to implement efficiently; see the proof of \Cref{thm:basic}.}
\label{alg:basic}
\begin{algorithmic}[1]
\Require Graph $G=(V,E)$ with $V \subseteq \Sigma^*$, integer $\ell>0$
\Ensure Motif pair $\{X^\star,Y^\star\}$ from $\Sigma^{\ell}$ maximizing $f_E$
\State Let $C:\Sigma^{\ell} \to \Z_{\geq 0}$ be indexed using SA ranks \label{alg:basic:dict}
\State $(X^\star,Y^\star) \gets (\bot, \bot)$
\For{each distinct length-$\ell$ substring $X$ in $V$} \label{alg:basic:outer-loop}
    \State $V_X \gets \{u \in V \mid X \text{ is a substring of } u\}$ \label{alg:basic:index}
    \For{each $\{u,v\}\in E$ with $u\in V_X$ }\label{alg:basic:edge-loop}
        \For{each distinct length-$\ell$ substring $Y$ of $v$} \label{alg:basic:inner-loop}
            \If{$X \preceq Y$}
                \State $C[Y] \gets C[Y] + 1$
            \EndIf
        \EndFor
    \EndFor
    \For{each $\{u,v\}\in E$ with $u,v\in V_X$}
        \For{each distinct length-$\ell$ substring $Y$ of $u$}
            \If{$X \preceq Y$ and $Y$ is \textbf{not} a substring of $v$}
                \State $C[Y] \gets C[Y] + 1$
            \EndIf
        \EndFor
    \EndFor
    \State{Update $\{X^\star, Y^\star\}$ if $\max(C) > f_E(G, \{X^\star, Y^\star\})$}
    \State Reset (the touched entries of) $C$
\EndFor
\State \Return $\{X^\star,Y^\star\}$
\end{algorithmic}
\end{algorithm}

\begin{remark}[Sparse graphs]
    It is easy to see that the running time of \Cref{alg:basic} can also be bounded by $\cO(N \log N + \setsize{E} \max_{v \in V}\setsize{v}^2)$.
    Thus, we expect the algorithm to perform better on sparse graphs where the labels are relatively short. The latter is natural because in many practical instances, the labels have roughly the same length (thus, relatively short compared to $N$).
\end{remark}
\begin{remark}[Enumeration]
    Instead of simply optimizing $f_E$ in \Cref{alg:basic},
    we can build a complete sorted list of all motif pairs and their support value (storing them using ranks).
    Since the values of the support measure are bounded by $|E| \le N^2$, we can do this while maintaining the $\cO(N^2)$ running time using an algorithm like radix sort.
    Unfortunately, this increases the space complexity to $\cO(N^2)$ in the worst case.
    In practice, we are only interested in the top-$k$ motif pairs.
    These we can maintain with an $\cO(\log k)$ multiplicative overhead by keeping a heap with the top-$k$ pairs.
\end{remark}

We defer the changes required for $f_{\chi^2}$
and the \SETH lower bounds for both $f_{E}$ and $f_{\chi^2}$ to Appendix~\ref{app:exact}.

\section{An exact algorithm for \ldCMM.}\label{sec:wild}
\Cref{alg:basic} for \CMM (\Cref{thm:basic}) will serve as a basis for our exact algorithm for \ldCMM.

The basic premise is simple: iterate over all possible combinations of wildcard positions in the two motifs and then run a modified version of \Cref{alg:basic} that ignores the characters at these positions.
While this seems easy, achieving it without losing efficiency is somewhat tricky, especially when targeting $\cO(N)$ space.
Recall that in \Cref{alg:basic}, we heavily relied on the rank of every length-$\ell$ substring in the suffix array (SA) of the concatenated node strings.
In particular, we relied on the fact that this rank is a compact representation of the substrings.
Thus, we were able to compare any two substrings for equality by simply comparing their ranks.
With the $d$ wildcards, we need to treat substrings that are actually different as the same, that is, we have a relaxed notion of equality that ignores the wildcarded positions.
Defining the right notion of rank and implementing access to it efficiently requires a substantial amount of care. We show the following core lemma.

\begin{lemma}\label{lem:wildcard-ranks}
Given a set $V\subseteq\Sigma^*$ of strings of total length $N$ over $\Sigma=\{1,2,\ldots,N\}$, an integer $\ell\in[1,N]$, and a set $H \subseteq [\ell]$ with $d\coloneqq |H| < \ell$, we can
construct an $\cO(N)$-size data structure in $\cO((d+1) N)$ time and $\cO(N)$ working space answering the following queries in $\cO(1)$ time: given $(u,i)$, with $u\in V$ and $i \in [1,\setsize{u} - \ell  +1]$, return a \emph{rank} (a unique identifier) of string $M$ of length $\ell$, such that, for all $j \in [1, \ell]$:
\[
M[j] \coloneqq
\begin{cases}
\letx & \text{if } j \in H, \\
u[i + j - 1] & \text{if } j \notin H.
\end{cases}
\]
In particular, the output rank is an integer from $[0,N)$.
\end{lemma}

\begin{proof}
Consider \Cref{alg:wildcard-ranks}.
We show that this algorithm solves exactly the claimed task.

Let $H=\{h_1<\cdots<h_d\}$. Further let $S_V=v_1\#_1v_2\#_2\dots v_{|V|}\#_{|V|}$ be the concatenation of the input strings with appropriate separators and let $N_V\coloneqq|S_V|-|V|$. (We will ignore suffixes of $S_V$ starting with separators.)
In $\cO(N_V)$ time, we construct the suffix array $\textsf{SA}_{S_V}$~\cite{SA} and the LCP array $\textsf{LCP}_{S_V}$~\cite{DBLP:conf/cpm/KasaiLAAP01}.

Consider the list of the suffixes of $S_V$.
The characters that occur at the positions $h_1, \ldots, h_d$ of every suffix are replaced by wildcards to obtain $N_V$ suffixes each split into $(d + 1)$ fragments.
The above replacements are \emph{implicit}; we never apply them.
Moreover, we are interested only in the length-$\ell$ prefix of each suffix.
Some care is needed: if a suffix is shorter than $\ell$, we ignore it from any further computation; and if a suffix contains a separator $\#_i$ in its length-$\ell$ prefix, we also ignore it. Using $\textsf{SA}_{S_V}$ and starting with the first $N_V$ substrings (prefixes of suffixes) of length exactly $(h_1-1)$, a new rank is assigned to each of them according to the lexicographical ordering of these substrings; this can be easily derived from a single scan of $\textsf{LCP}_{S_V}$ in $\cO(N_V)$ time. We then proceed to $h_2$ analogously, and so on.
To maintain $\cO(N_V)$ space, for every consecutively considered pair of wildcard positions, we use radix sort over the $N_V$ pairs of ranks (corresponding to the $N_V$ suffixes) to obtain a new integer ranking over $[0,N_V)$, and the process continues.
In the end, we have a new lexicographical ordering of the length-$\ell$ prefixes of suffixes of $S_V$ that disregards the characters at the $h_1$-th,$\ldots$,$h_d$-th positions. Moreover, for every two successive length-$\ell$ prefixes in this order, we know whether they have an LCP of length at least $\ell$, which suffices for our purposes.
The whole process thus takes $\cO(N_V\cdot (d+1))$ time.
Since this list is sorted, identical length-$\ell$ prefixes (here identical means with $d$ wildcards) are \emph{grouped} and so each group member gets the same rank  from $[0,N_V)$.

Two length-$\ell$ substrings are identical under $H$ if and only if their constituent fragments between the indices of $H$ are identical. By treating each substring as a tuple of ranks corresponding to these fragments---specifically $(r_1, r_2, \dots, r_{d+1})$---and applying radix sort, we obtain a lexicographical ordering (of the suffixes) based strictly on the subsequence of non-wildcard characters. This correctly groups all occurrences of a specific motif $M$ into a contiguous range in the sorted list.

The ranks (unique identifiers) are stored in the array $\mathsf{R}$ returned by \Cref{alg:wildcard-ranks}, indexed by the starting position of each suffix in $S_V$. Given $u\in V$ and an offset $i$, the length-$\ell$ prefix of the corresponding suffix is exactly the length-$\ell$ substring $u[i\dd i+\ell-1]$; hence, a single lookup in $\mathsf{R}$ returns its identifier in $\cO(1)$ time, which is invariant under replacing the positions in $H$ by $\letx$, as claimed. The $d$ radix-sort rounds and the LCP array scans each cost $\cO(N_V)$ time. In total, the algorithm costs $\cO((d+1) N)$ time and uses $\cO(N)$ working space.
\end{proof}

\begin{remark}
The technique in \Cref{lem:wildcard-ranks}
is inspired by the construction of the \emph{gapped suffix array}~\cite{DBLP:conf/spire/CrochemoreT10}.
A notable difference is the following:
since $\ell$ is fixed, \emph{we avoid the use of LCE queries} (cf.~\cite{DBLP:journals/tcs/LandauV86}),
and compute the relevant LCP values by comparing integers in $\cO(1)$ time (Lines 8--17).
While this makes no difference in theory, it is crucial in practice, as LCE data structures are typically either costly to construct or
costly to query.
\end{remark}

\begin{algorithm}[t]
\caption{Rank assignment to $(\ell,d)$-motifs. See the proof of \Cref{lem:wildcard-ranks}.}
\label{alg:wildcard-ranks}
\begin{algorithmic}[1]
\Require Wildcard positions $H=\{h_1<\dots<h_d\}\subseteq[\ell]$;
string $S_V=v_1\#_1 v_2\#_2\ldots v_{\setsize{V}}\#_{\setsize{V}}$;
suffix array $\mathsf{SA}_{S_V}$; LCP array $\mathsf{LCP}_{S_V}$
\Ensure Array $\mathsf{R}$ such that $\mathsf{R}[i]$ is the rank of suffix
$S_V[i \dd |S_V|]$ by its length-$\ell$ prefix, in lexicographic order
ignoring the positions in $H$
\State Let $\mathcal{S}$ be the positions $i$ such that $S_V[i\dd i+\ell-1]$
       is defined and contains no separator
\State $\mathsf{R}[i] \gets 0$ for each $i \in \mathcal{S}$
\State Add $h_0 \coloneqq 0$ and $h_{d+1} \coloneqq \ell+1$ to $H$ \Comment{Handle intervals before first and after last wildcard}
\For{$t \gets 0$ \textbf{to} $d$}
    \State $I_t = (h_t, h_{t+1})$ \Comment{Non-$\letx$ interval}
    \State \IfInline{$I_t$ is empty}{continue}
    \algstep{Step 1: Compute ranks for $I_t$}
    \State $\rank\gets 0$;\quad $\mathsf{R}'[\mathsf{SA}_{S_V}[1]] \gets 0$
    \For{$j \gets 2$ \textbf{to} $|S_V|$}
      \State \textbf{if} $\mathsf{LCP}_{S_V}[j] < \setsize{I_t}$ \textbf{then} $\rank\gets\rank+1$
      \State $\mathsf{R}'[\mathsf{SA}_{S_V}[j]] \gets \rank$
    \EndFor
    \algstep{Step 2: Refine by (current rank, rank for $I_t$)}
    \State Stably sort $\mathcal{S}$ by $\key(j) = \left(\mathsf{R}[\mathcal{S}[j]],\ \mathsf{R}'[\mathcal{S}[j]+h_t]\right)$
    \State $\rank \gets 0;\quad \widehat{\mathsf R}[\mathcal S[1]]\gets 0$
    \For{$j \gets 2$ \textbf{to} $|\mathcal{S}|$}
      \State \IfInline{$\key(j) \ne \key(j-1)$}{$\rank \gets \rank + 1$}

      \State $\widehat{\mathsf R}[\mathcal S[j]] \gets \rank$
    \EndFor
    \State $\mathsf{R} \gets \widehat{\mathsf R}$
\EndFor
\State \Return $\mathsf{R}$
\end{algorithmic}
\end{algorithm}

\begin{theorem}\label{thm:main}
There is an algorithm for $(\ell,d)$-\CMM with $f_E$ running in time \[\cO\left(\binom{\ell}{d}^2 N^2\right),\]
where $N\coloneqq\sum_{v \in V}|v|$.
The space used is $\cO(N)$.
\end{theorem}

\begin{proof}
We solve $(\ell,d)$-\CMM by a slight adaptation of the algorithm from \Cref{sec:exact}. Fix a pair $(H,H')$ of wildcard combinations, one for $X$ and one for $Y$. In a pre-processing step, we sort the characters occurring in all strings $v\in V$, thus mapping them onto the range $\{1,2,\ldots,N\}$. This is done once in $\cO(N \log N)$ time. We now apply \Cref{lem:wildcard-ranks} that takes $\cO((d+1)N)$ time using $\cO(N)$ space, and obtain, for every length-$\ell$ substring $X$ (resp.\ $Y$) of a node, an identifier $\id(X)$ (resp.\ $\id(Y)$) under $H$ (resp.\ under $H'$) such that: any two length-$\ell$ substrings receive the
same identifier if and only if they agree on all non-wildcard positions.
Namely, we view $\id(X)$ as the equivalent of the rank from \Cref{alg:basic}.

These identifiers satisfy the two properties that the proof of
\Cref{thm:basic} requires of the rank: they let us (1) test equality of two
(wildcarded) length-$\ell$ substrings in $\cO(1)$ time and (2) index the dictionary $C$ in
$\cO(N)$ space. Hence, we run \Cref{alg:basic} verbatim on the fixed pair $(H,H')$,
substituting the wildcard-aware ranks for the plain ranks. By \Cref{thm:basic}, this
computes the maximum $f_E$ over all motif pairs whose wildcards sit at positions
$(H,H')$ in $\cO(N^2)$ time and $\cO(N)$ space.

Since the problem requires maximizing an objective function over unknown (non-fixed)
wildcard positions, the algorithm iterates through all possible pairs of $d$-size subsets $H \subset [\ell]$. This exhaustive search guarantees completeness: if an optimal motif pair
$\{X, Y\}$ exists with a specific wildcard configuration, the algorithm will inevitably
process it during the iteration. Running the above for each of the $\binom{\ell}{d}$ choices of $H$ and $\binom{\ell}{d}$ choices of $H'$ and taking the overall maximum therefore solves
$(\ell,d)$-\CMM in  $\cO\!\left(\binom{\ell}{d}^2 N^2\right)$ time using $\cO(N)$ space.
\end{proof}
This algorithm can be extended to the $\chisup$ support measure by adapting the $\chisup$ version of \Cref{alg:basic} (see \Cref{cor:basic-chisup}).
We defer the W[1]-hardness for $(\ell,d)$-\CMM, parameterized by $d$, to Appendix~\ref{app:wild}.

\subsection{Parallelization.}
\label{sec:parallel}

The algorithm of \Cref{thm:main} is easy to parallelize since the computation can be carried out independently for each \emph{wildcard combination}.
As motif pairs are unordered, roughly half of the $\binom{\ell}{d}^2$ combinations need to be considered, giving $\binom{\ell}{d}\big(\binom{\ell}{d}+1\big)/2$ independent jobs, each running in $\cO(N^2)$ time.
For $H \neq H'$, the motifs $X$ and $Y$ are always distinct, so the tie-breaking test $X \preceq Y$ of \Cref{alg:basic} is applied only in the jobs with $H = H'$.
Since each thread processes a disjoint subset of the motif pairs, the global optimum is the maximum of the per-job optima, so aggregation is trivial. We will later see that this is both easy to implement and efficient in practice.
Moreover, by using \Cref{thm:main}, $\cO(N)$ space is used per thread.

\section{Experimental results.} \label{sec:exp}
In this section, we compare our exact algorithm against SEQ-SLIDER and M-SLIDER, two variants of SLIDER~\cite{DBLP:journals/tcbb/BoyenDNHD11}, the state-of-the-art (heuristic) method.
Each employs a different neighborhood function when sliding over neighboring motifs during steepest ascent. We also empirically evaluate the time and space complexities of our exact algorithm to confirm our theoretical findings.
These goals were split into the following research questions:

\begin{enumerate}[leftmargin=3.1em,label=\bfseries RQ\arabic*:]
    \item How do the scores of SLIDER's solution set compare to the ones of the optimal solution set?
    \item What is the recall of SLIDER's solution set relative to the optimal solution set?
    \item How does our algorithm scale as $|V|$ (and thus $N$) grows?
    \item How does our algorithm scale as $|E|$ grows?
    \item How does our algorithm scale as the number $t$ of threads used increases?
\end{enumerate}

\begin{remark}
We excluded D-MOTIF~\cite{tan2006} and FindMotif~\cite{leung2007}, as their complexities are prohibitive for our datasets and parameters, and for the same scalability reasons, we excluded their heuristic counterparts, D-STAR and MotifHeuristics.
FindMotif and MotifHeuristics also lack publicly available implementations and rely on a different, probabilistic support measure.
MotifHeuristics was re-implemented and evaluated in~\cite{DBLP:journals/tcbb/BoyenDNHD11}, where it was shown to be outperformed by SLIDER.
\end{remark}

\subsection{Datasets.}
We selected 5 networks of different species and with varying sizes among the 10 most queried networks from the STRING database~\cite{mering2003,szklarczyk2023} (version 12).
In order to replicate the setup of~\cite{DBLP:journals/tcbb/BoyenDNHD11} for RQ1 and RQ2 (which concern the \emph{accuracy} of SLIDER), we also included the two networks used in~\cite{DBLP:journals/tcbb/BoyenDNHD11}.
The characteristics of the networks are summarized in \Cref{tab:string-data-preproc}.
These networks were preprocessed as described in \Cref{a:datasets}. All networks in \Cref{tab:string-data-preproc} were used to address RQ1 and RQ2.
For RQ3, RQ4, and RQ5, we also generated smaller artificial networks as described in \Cref{a:datasets} to thoroughly investigate different input parameters.

\begin{table}[t]
\footnotesize
\caption{Top: Parameters of datasets from the STRING database after preprocessing. Bottom: Parameters of datasets from~\cite{DBLP:journals/tcbb/BoyenDNHD11} after preprocessing.}\label{tab:string-data-preproc}
\begin{center}
\begin{tabular}{lccc}
\toprule
Dataset & $|V|$ & $|E|$ & $N$ \\
\midrule
E.\ coli          & \num{4006}  & \num{32244}  & \num{1277734} \\
S.\ cerevisiae I  & \num{5791}  & \num{104188} & \num{2744228} \\
D.\ melanogaster  & \num{10522} & \num{118698} & \num{6043033} \\
A.\ thaliana      & \num{20239} & \num{331525} & \num{8903236} \\
H.\ sapiens I     & \num{16201} & \num{236930} & \num{9605435} \\
\midrule
S.\ cerevisiae II & \num{1620}  & \num{9060}   & \num{862281}  \\
H.\ sapiens II    & \num{10380} & \num{69515}  & \num{5607838} \\
\bottomrule
\end{tabular}
\end{center}
\end{table}

\subsection{Implementations.}
We implemented our exact algorithm (\Cref{thm:main}) in \texttt{C++}, which we make available at \url{https://github.com/lucper/cmm}.
We used \texttt{libsais64} (version 2.4.10) to construct the SA and the LCP array~\cite{libsais} and implemented the parallel version of \Cref{thm:main} (as described in \Cref{sec:parallel}) using OpenMP~\cite{dagum1998-omp}.
For SLIDER, we used the \texttt{Java} implementation provided by the authors~\cite{DBLP:journals/tcbb/BoyenDNHD11}, which is also available in our repository. Since we do not compare the efficiency of the two implementations, using the latter suffices for our purposes.

\subsection{Setup.}
All experiments were conducted on two 32-core Intel Xeon Gold 6530 CPUs (2 threads per core, totaling $P=128$ threads across the two CPUs), operating at a base clock frequency of 2.1\,GHz (up to 4.0\,GHz), and 512\,GB RAM.
We compiled our implementations with \texttt{g++} version \texttt{15.2.1} at the \texttt{-O3} optimization level.
Wall-clock times were measured with \texttt{std::chrono::steady\_clock}.
Peak RAM usage was measured as the maximum resident set size using the \texttt{getrusage} function.

\subsection{Metrics for RQ1 and RQ2.}\label{s:metrics}
\paragraph{Similarity index.}
In order to compare solution sets, we introduce a notion of similarity between $(\ell,d)$-motif pairs.
We count the edges in which two $(\ell,d)$-motif pairs, $\{X,Y\}$ and $\{Z,W\}$, co-occur such that the $(\ell,d)$-motifs in each endpoint of the edge match substrings that are at most $h$ letters apart, where $h$ is a non-negative integer.
(Note that for $h=0$, we count the edges in which there exists a substring in each endpoint that both $(\ell,d)$-motifs for that endpoint match.)

Let $\occ_X(u)$ denote the set of occurrences of an $(\ell,d)$-motif $X$ in string $u$, namely $\occ_X(u)=\{i \mid X \text{ matches } u[i \dd i+\ell-1]\}$.
Let $E_h'((X,Y),(W,Z)) \subseteq E$ be the set of edges $\{u,v\}$ such that there exists:
\begin{enumerate}
    \item $i \in \occ_X(u)$ and $j \in \occ_W(u)$ with $|i-j| \leq h$; and
    \item $i \in \occ_Y(v)$ and $j \in \occ_Z(v)$ with $|i-j| \leq h$.
\end{enumerate}
Further let $E_h(\{X,Y\},\{W,Z\})=E_h'((X,Y),(W,Z)) \cup E_h'((X,Y),(Z,W))$.
Then, we define the \emph{similarity index} between $(\ell,d)$-motif pairs $\{X,Y\}$ and $\{Z,W\}$ as
\[
s_h(\{X,Y\},\{Z,W\})\coloneqq\frac{|E_h(\{X,Y\},\{Z,W\})|}{|E_{X,Y} \cup E_{Z,W}|}.
\]
In particular, we note that $s_0$ corresponds to the similarity index used by~\cite{DBLP:journals/tcbb/BoyenDNHD11} in their experimental evaluation of SLIDER.
In our evaluation, we used the more lenient version $s_{\ell}$ for a motif length $\ell$; that is, we afforded SLIDER's solution set more slack when it was compared to our exact algorithm's solution set.

\paragraph{Reduced solution set.}
Let $S=\{M_1,M_2,\ldots,M_k\}$ be a solution set ranked from the highest-scoring motif pair $M_1$ to the lowest-scoring one, $M_k$.
We construct a \emph{reduced solution set} $S^* \subseteq S$ with the following property: for any distinct motif pairs $M'$ and $M''$ in $S^*$, we require that $s_0(M',M'') < 1$.
Note that our similarity relation $s_h$ is not transitive, so there may exist multiple subsets of $S$ with this property.
Following the approach of~\cite{DBLP:journals/tcbb/BoyenDNHD11}, our set $S^*$ is constructed with the following greedy strategy: for each $i=1, \ldots, k$, we remove every $M_j \in S$ with $j > i$ satisfying $s_0(M_i,M_j)=1$.

\paragraph{Coverage of one solution set by another.}
Given two reduced solution sets $S_A$ and $S_B$, we define the \emph{coverage of $S_B$ by $S_A$}, parameterized by $h$ and $\tau$, as
{\small\[
\mathrm{Cov}_{h,\tau,A}(S_B)\coloneqq\frac{\setsize{\{M \in S_B \mid \exists M' \in S_A,\ s_h(M,M') \geq \tau\}}}{\setsize{S_B}}.
\]}
This resembles the notion of \emph{recall}: we are effectively measuring the fraction of motif pairs from $S_B$ retrieved by $A$ (under similarity $s_h$ at a threshold $\tau$).

\paragraph{Setup.}
In order to compare M-SLIDER and SEQ-SLIDER's solution quality with that of our exact algorithm, we used complementary metrics: first, we compared the $f_{\chi^2}$ score distribution of M-SLIDER and SEQ-SLIDER's reduced solution sets with our exact algorithm's (RQ1); second, we used the similarity index introduced above to measure how similar the reduced solution sets are (RQ2).

We ran M-SLIDER and SEQ-SLIDER over 10 trials, each with a maximum convergence time of 10 minutes (totaling 100 minutes of running time for each dataset), and requested at most $\num{10000}$ highest scoring $(\ell,d)$-motif pairs per trial.
Then, we concatenated them, removed identical $(\ell,d)$-motif pairs \ie{the same pair rediscovered across independent trials}, sorted by the $f_{\chi^2}$ score descending, retained at most $\num{10000}$ top-scoring distinct $(\ell,d)$-motif pairs, and applied the ``de-duplication'' procedure described above.
The resulting reduced solution sets of M-SLIDER and SEQ-SLIDER are denoted by $S_{\MSlider}$ and $S_{\SEQSlider}$, respectively.
We ran our exact algorithm once, retrieved the top $\num{10000}$ highest scoring $(\ell,d)$-motif pairs, and applied the ``de-duplication'' procedure; the resulting reduced solution set is denoted by $S_{\Exact}$.
For RQ1 and RQ2, we ran this setup for parameter pairs $(\ell=8,d=5)$ and $(\ell=8,d=3)$, following the setup of~\cite{DBLP:journals/tcbb/BoyenDNHD11}.

\subsection{Results for RQ1.}\label{s:rq1}
We compared the $f_{\chi^2}$ scores of $S_{\MSlider}$ and $S_{\SEQSlider}$ with $S_{\Exact}$ as follows.
Since any solution set is ranked from highest to lowest $f_{\chi^2}$ score and the highest scoring $(\ell,d)$-motif pair $M^*$ of $S_{\Exact}$ has the best possible $f_{\chi^2}$ score, we normalized the $f_{\chi^2}$ score of each $(\ell,d)$-motif pair from all (reduced) solution sets by dividing each by $f_{\chi^2}(M^*)$.
Since we are comparing the top-scoring $(\ell,d)$-motif pairs across $S_{\MSlider}$, $S_{\SEQSlider}$, and $S_{\Exact}$, the resulting solution sets $S_{\MSlider}'$, $S_{\SEQSlider}'$, and $S_{\Exact}'$ after normalization were truncated to have size $\min\{\setsize{S_{\MSlider}},\setsize{S_{\SEQSlider}},\setsize{S_{\Exact}}\}$.
\Cref{fig:rq1-score-l8d5} shows the resulting (normalized) $f_{\chi^2}$ score distributions of $S_{\MSlider}'$ and $S_{\SEQSlider}'$ compared to $S_{\Exact}'$ for the parameter pair $(\ell=8,d=5)$; for the parameter pair $(\ell=8,d=3)$, see \Cref{fig:rq1-score-l8d3}.
Firstly, we note that SEQ-SLIDER consistently outperforms M-SLIDER in all datasets, which confirms the evaluation in~\cite{DBLP:journals/tcbb/BoyenDNHD11}.
Secondly, the $f_{\chi^2}$ score of $S_{\SEQSlider}'$ is comparable to $S_{\Exact}'$ only for the E.\ coli, S.\ cerevisiae (I and II), and D.\ melanogaster datasets across parameter pairs $(\ell=8,d=3)$ and $(\ell=8,d=5)$.
For the (larger) remaining datasets, A.\ thaliana and H.\ sapiens (I and II), the $f_{\chi^2}$ score of $S_{\SEQSlider}'$ is far from $S_{\Exact}'$.
For instance, for the parameter pair $(\ell=8,d=5)$, $S_{\SEQSlider}$ has average $f_{\chi^2}$ score $\num{14963}$, $\num{7313}$, and $\num{1094}$ for the A.\ thaliana, H.\ sapiens I, and H.\ sapiens II datasets, respectively, whereas $S_{\Exact}$ has average $f_{\chi^2}$ score $\num{26777}$, $\num{10439}$, and $\num{1514}$, respectively.
We note that these gaps are larger for the parameter pair $(\ell=8,d=3)$ (see \Cref{fig:rq1-score-l8d3}).

These results establish that our exact algorithm mines motif pairs whose $f_{\chi^2}$ scores are at least as high as SEQ-SLIDER's on all datasets and substantially higher on larger ones.
Two possible explanations are: (1) our algorithm captures the same interactions with more precise motif pairs, raising the score; or (2) our algorithm finds interactions that SLIDER misses altogether.
\Cref{s:rq2} is designed to show that case (2) is a significant contributor to the observed difference in scores.

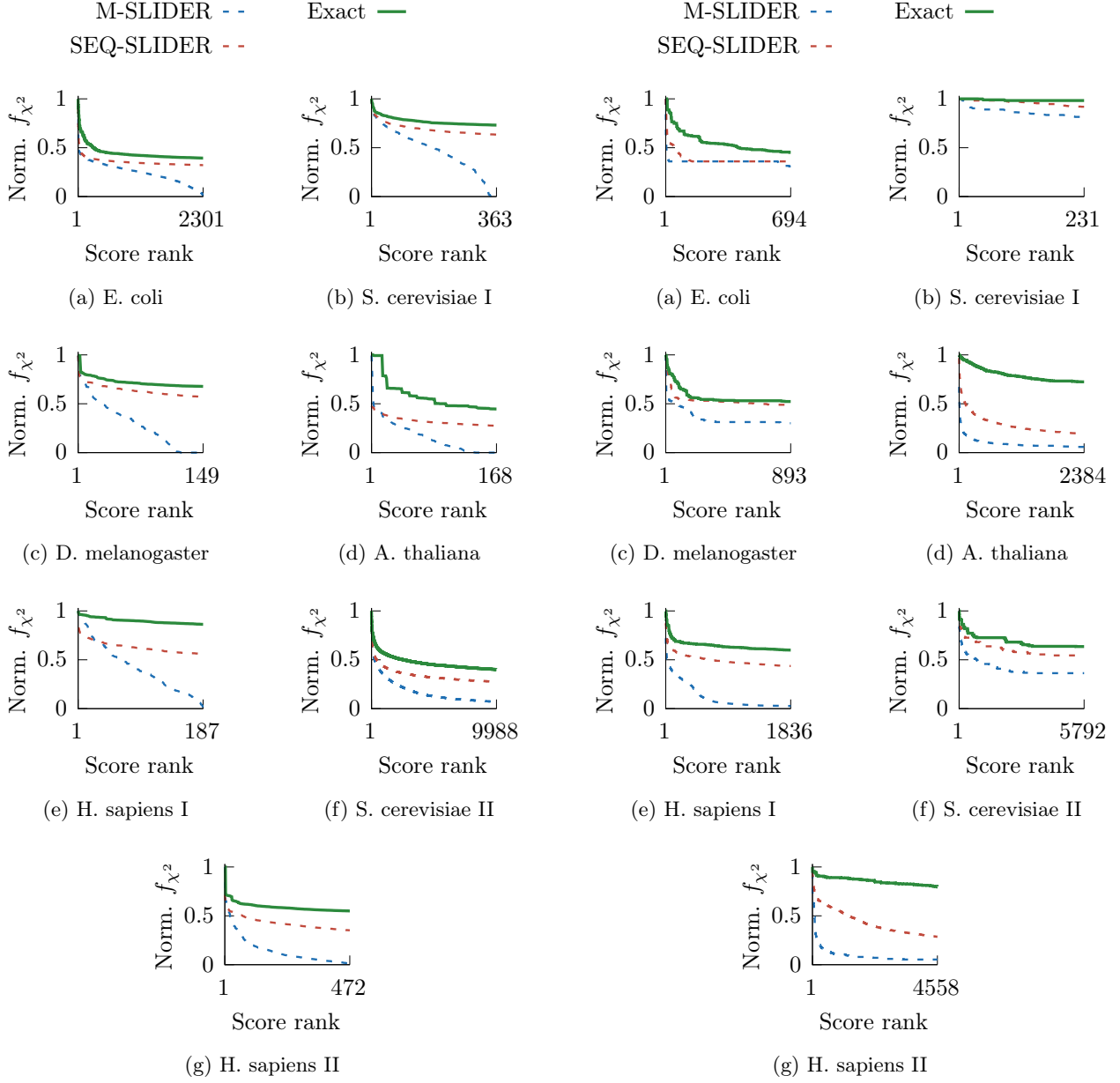
\begin{figure}[t!]
    \centering
    \input{curve_legend.tex}

    \vspace{0.4em}

    \begin{subfigure}[t]{0.48\columnwidth}
        \centering\input{curve.ecoli.l8d5.dedup.tex}
        \caption{E.\ coli}
    \end{subfigure}\hfill%
    \begin{subfigure}[t]{0.48\columnwidth}
        \centering\input{curve.yeast.l8d5.dedup.tex}
        \caption{S.\ cerevisiae I}
    \end{subfigure}

    \vspace{0.6em}

    \begin{subfigure}[t]{0.48\columnwidth}
        \centering\input{curve.fly.l8d5.dedup.tex}
        \caption{D.\ melanogaster}
    \end{subfigure}\hfill%
    \begin{subfigure}[t]{0.48\columnwidth}
        \centering\input{curve.plant.l8d5.dedup.tex}
        \caption{A.\ thaliana}
    \end{subfigure}

    \vspace{0.6em}

    \begin{subfigure}[t]{0.48\columnwidth}
        \centering\input{curve.human.l8d5.dedup.tex}
        \caption{H.\ sapiens I}
    \end{subfigure}\hfill%
    \begin{subfigure}[t]{0.48\columnwidth}
        \centering\input{curve.yeast_aj.l8d5.dedup.tex}
        \caption{S.\ cerevisiae II}
    \end{subfigure}

    \vspace{0.6em}

    \begin{subfigure}[t]{0.48\columnwidth}
        \centering\input{curve.human_aj.l8d5.dedup.tex}
        \caption{H.\ sapiens II}
    \end{subfigure}\hfill%

    \caption{Distribution of (normalized) $f_{\chi^2}$ scores of $S_{\MSlider}'$ and $S_{\SEQSlider}'$ compared to $S_{\Exact}'$ across datasets of \Cref{tab:string-data-preproc} for parameter pair $(\ell=8,d=5)$.}
    \label{fig:rq1-score-l8d5}
\end{figure}

\begin{figure}[t!]
    \centering
    \input{curve_legend.tex}

    \vspace{0.4em}

    \begin{subfigure}[t]{0.48\columnwidth}
        \centering\input{curve.ecoli.l8d3.dedup.tex}
        \caption{E.\ coli}
    \end{subfigure}\hfill%
    \begin{subfigure}[t]{0.48\columnwidth}
        \centering\input{curve.yeast.l8d3.dedup.tex}
        \caption{S.\ cerevisiae I}
    \end{subfigure}

    \vspace{0.6em}

    \begin{subfigure}[t]{0.48\columnwidth}
        \centering\input{curve.fly.l8d3.dedup.tex}
        \caption{D.\ melanogaster}
    \end{subfigure}\hfill%
    \begin{subfigure}[t]{0.48\columnwidth}
        \centering\input{curve.plant.l8d3.dedup.tex}
        \caption{A.\ thaliana}
    \end{subfigure}

    \vspace{0.6em}

    \begin{subfigure}[t]{0.48\columnwidth}
        \centering\input{curve.human.l8d3.dedup.tex}
        \caption{H.\ sapiens I}
    \end{subfigure}\hfill%
    \begin{subfigure}[t]{0.48\columnwidth}
        \centering\input{curve.yeast_aj.l8d3.dedup.tex}
        \caption{S.\ cerevisiae II}
    \end{subfigure}

    \vspace{0.6em}

    \begin{subfigure}[t]{0.48\columnwidth}
        \centering\input{curve.human_aj.l8d3.dedup.tex}
        \caption{H.\ sapiens II}
    \end{subfigure}\hfill%

    \caption{Distribution of (normalized) $f_{\chi^2}$ scores of $S_{\MSlider}'$ and $S_{\SEQSlider}'$ compared to $S_{\Exact}'$ across datasets of \Cref{tab:string-data-preproc} for parameter pair $(\ell=8,d=3)$.}
    \label{fig:rq1-score-l8d3}
\end{figure}

\subsection{Results for RQ2.}\label{s:rq2}
We measured the fraction (coverage) of $(\ell,d)$-motif pairs in $S_{\Exact}$ that have similarity $s_{\ell} \geq 70\%$ to at least one $(\ell,d)$-motif pair in $S_{\MSlider}$ and $S_{\SEQSlider}$.
Here, we summarize the results for SEQ-SLIDER only since it outperformed M-SLIDER in all but one instance (H.\ sapiens I with $(\ell=8,d=3)$); refer to \Cref{tab:rq2} in the appendix.
The coverage of $S_{\Exact}$ by $S_{\SEQSlider}$ ranges from $8\%$ to $52\%$ across both parameter pairs for the A.\ thaliana and H.\ sapiens (I and II) datasets, meaning that SEQ-SLIDER misses at least roughly half of the top-scoring motif pairs, and up to $92\%$ of them (M-SLIDER fares similarly).
For the remaining datasets --- E.\ coli, S.\ cerevisiae (I and II) and D.\ melanogaster --- the coverage of $S_{\Exact}$ by $S_{\SEQSlider}$ ranges from $45\%$ to $100\%$.
In particular, it was $100\%$ only for S.\ cerevisiae I and D.\ melanogaster with parameter pair $(\ell=8,d=3)$.
Note that we used the more lenient version $s_{\ell}$ of the similarity index to derive these results.

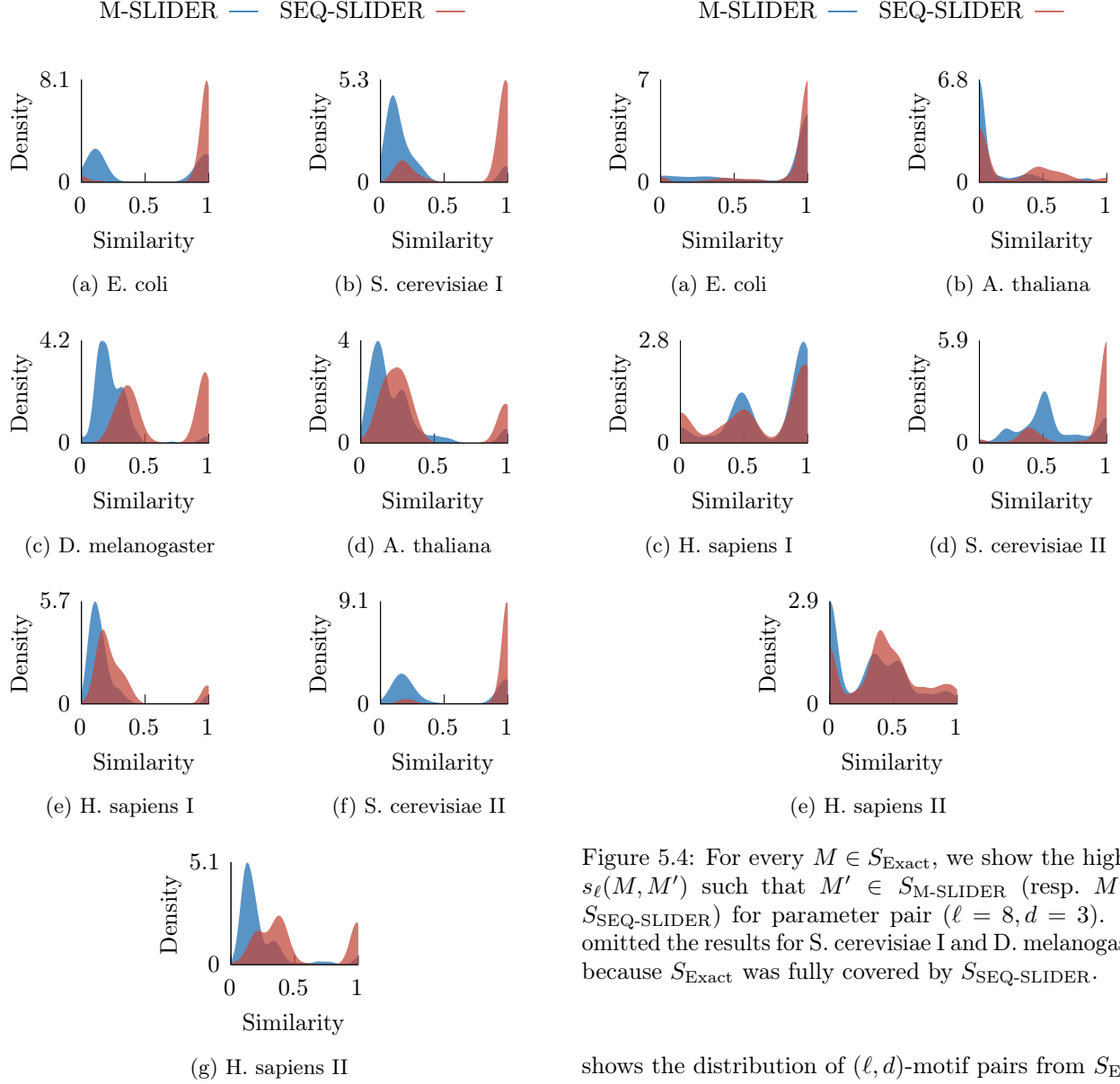
\begin{figure}[t!]
    \centering
    \input{density_legend.tex}

    \vspace{0.4em}

    \begin{subfigure}[t]{0.48\columnwidth}
        \centering\input{density.ecoli.l8d5.h8.dedup.tex}
        \caption{E.\ coli}
    \end{subfigure}\hfill%
    \begin{subfigure}[t]{0.48\columnwidth}
        \centering\input{density.yeast.l8d5.h8.dedup.tex}
        \caption{S.\ cerevisiae I}
    \end{subfigure}

    \vspace{0.6em}

    \begin{subfigure}[t]{0.48\columnwidth}
        \centering\input{density.fly.l8d5.h8.dedup.tex}
        \caption{D.\ melanogaster}
    \end{subfigure}\hfill%
    \begin{subfigure}[t]{0.48\columnwidth}
        \centering\input{density.plant.l8d5.h8.dedup.tex}
        \caption{A.\ thaliana}
    \end{subfigure}

    \vspace{0.6em}

    \begin{subfigure}[t]{0.48\columnwidth}
        \centering\input{density.human.l8d5.h8.dedup.tex}
        \caption{H.\ sapiens I}
    \end{subfigure}\hfill%
    \begin{subfigure}[t]{0.48\columnwidth}
        \centering\input{density.yeast_aj.l8d5.h8.dedup.tex}
        \caption{S.\ cerevisiae II}
    \end{subfigure}

    \vspace{0.6em}

    \begin{subfigure}[t]{0.48\columnwidth}
        \centering\input{density.human_aj.l8d5.h8.dedup.tex}
        \caption{H.\ sapiens II}
    \end{subfigure}\hfill%

    \caption{For every $M \in S_{\Exact}$, we show the highest $s_{\ell}(M,M')$ such that $M' \in S_{\MSlider}$  (resp.\ $M' \in S_{\SEQSlider}$) for parameter pair $(\ell=8,d=5)$.}
    \label{fig:rq2-density-l8d5}
\end{figure}

\begin{figure}[t!]
    \centering
    \input{density_legend.tex}

    \vspace{0.4em}

    \begin{subfigure}[t]{0.48\columnwidth}
        \centering\input{density.ecoli.l8d3.h8.dedup.tex}
        \caption{E.\ coli}
    \end{subfigure}\hfill%
    \begin{subfigure}[t]{0.48\columnwidth}
        \centering\input{density.plant.l8d3.h8.dedup.tex}
        \caption{A.\ thaliana}
    \end{subfigure}

    \vspace{0.6em}

    \begin{subfigure}[t]{0.48\columnwidth}
        \centering\input{density.human.l8d3.h8.dedup.tex}
        \caption{H.\ sapiens I}
    \end{subfigure}\hfill%
    \begin{subfigure}[t]{0.48\columnwidth}
        \centering\input{density.yeast_aj.l8d3.h8.dedup.tex}
        \caption{S.\ cerevisiae II}
    \end{subfigure}

    \vspace{0.6em}

    \begin{subfigure}[t]{0.48\columnwidth}
        \centering\input{density.human_aj.l8d3.h8.dedup.tex}
        \caption{H.\ sapiens II}
    \end{subfigure}\hfill%

    \caption{For every $M \in S_{\Exact}$, we show the highest $s_{\ell}(M,M')$ such that $M' \in S_{\MSlider}$ (resp.\ $M' \in S_{\SEQSlider}$) for parameter pair $(\ell=8,d=3)$. We omitted the results for S.\ cerevisiae I and D.\ melanogaster because $S_{\Exact}$ was fully covered by $S_{\SEQSlider}$.}
    \label{fig:rq2-density-l8d3}
\end{figure}

Moreover, for each $(\ell,d)$-motif pair $M \in S_{\Exact}$, we also computed the highest similarity $s_{\ell}(M,M')$, where $M' \in S_{\MSlider}$ (resp.\ $M' \in S_{\SEQSlider}$).
\Cref{fig:rq2-density-l8d5} shows the distribution of $(\ell,d)$-motif pairs from $S_{\Exact}$ for the parameter pair $(\ell=8,d=5)$ as a density plot.
We observe that $S_{\Exact}$ has a substantial amount of $(\ell,d)$-motif pairs with similarity $s_{\ell}$ below $50\%$ to their closest counterparts in $S_{\MSlider}$ and $S_{\SEQSlider}$ for the D.\ melanogaster, A.\ thaliana, and H.\ sapiens (I and II) datasets.
On the other hand, for the E.\ coli and S.\ cerevisiae (I and II) datasets, most $(\ell,d)$-motif pairs in $S_{\Exact}$ have similarity above $70\%$ to their closest counterparts in $S_{\MSlider}$ and $S_{\SEQSlider}$.

We arrive at the same conclusion for the parameter pair $(\ell=8,d=3)$, although SEQ-SLIDER had a better overall performance compared to $(\ell=8,d=5)$; see \Cref{fig:rq2-density-l8d3}.
In particular, we note that every motif pair in $S_{\Exact}$ has similarity $s_{\ell} \geq 70\%$ with some motif pair in $S_{\MSlider}$ or $S_{\SEQSlider}$ on the S.\ cerevisiae I and D.\ melanogaster datasets.

Together, these results establish that our exact algorithm mines a substantial amount of $(\ell,d)$-motif pairs that are missed by SEQ-SLIDER and M-SLIDER for the \emph{largest} datasets from the STRING database (A.\ thaliana and H.\ sapiens I) and from~\cite{DBLP:journals/tcbb/BoyenDNHD11} (H.\ sapiens II).

\begin{figure}[t]
  \centering
  \includegraphics[width=0.9\linewidth]{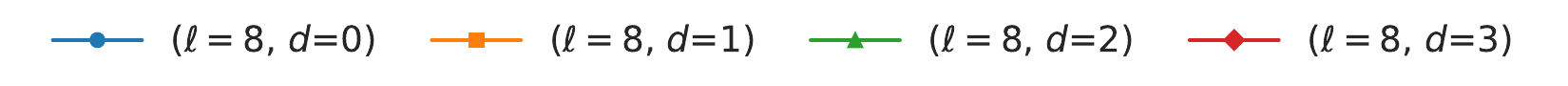}\\[2pt]
  \begin{subfigure}[t]{0.49\linewidth}
    \centering
    \includegraphics[width=\linewidth]{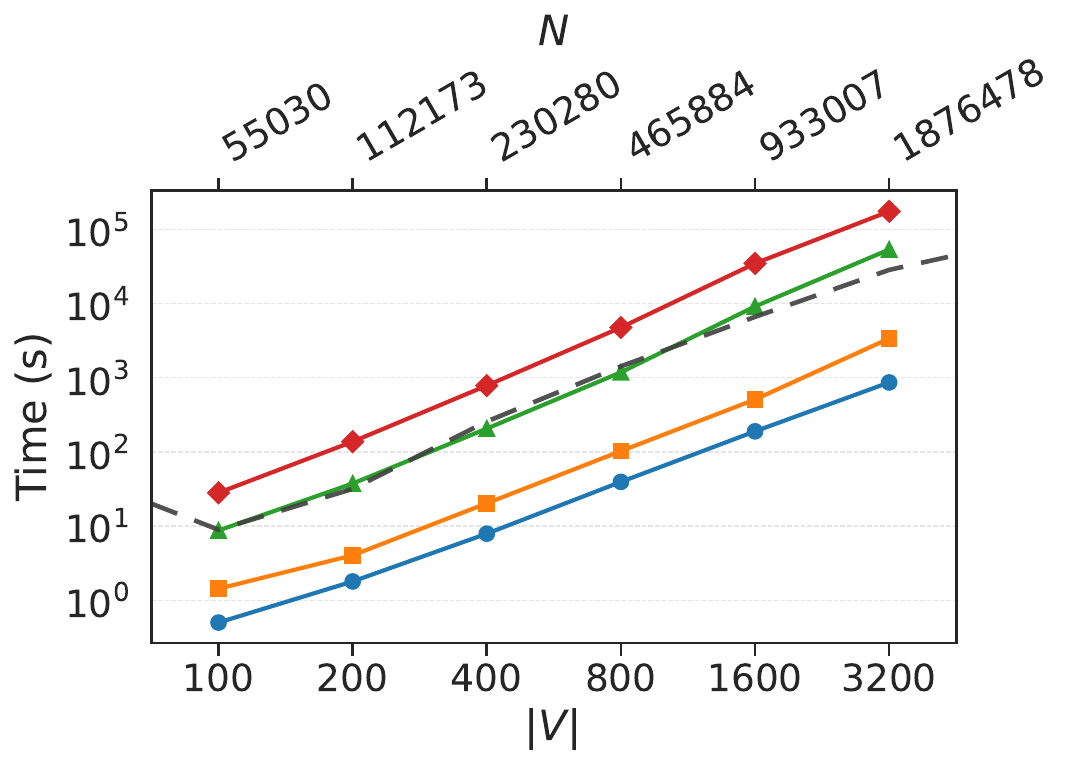}
  \end{subfigure}
  \hfill
  \begin{subfigure}[t]{0.49\linewidth}
    \centering
    \includegraphics[width=\linewidth]{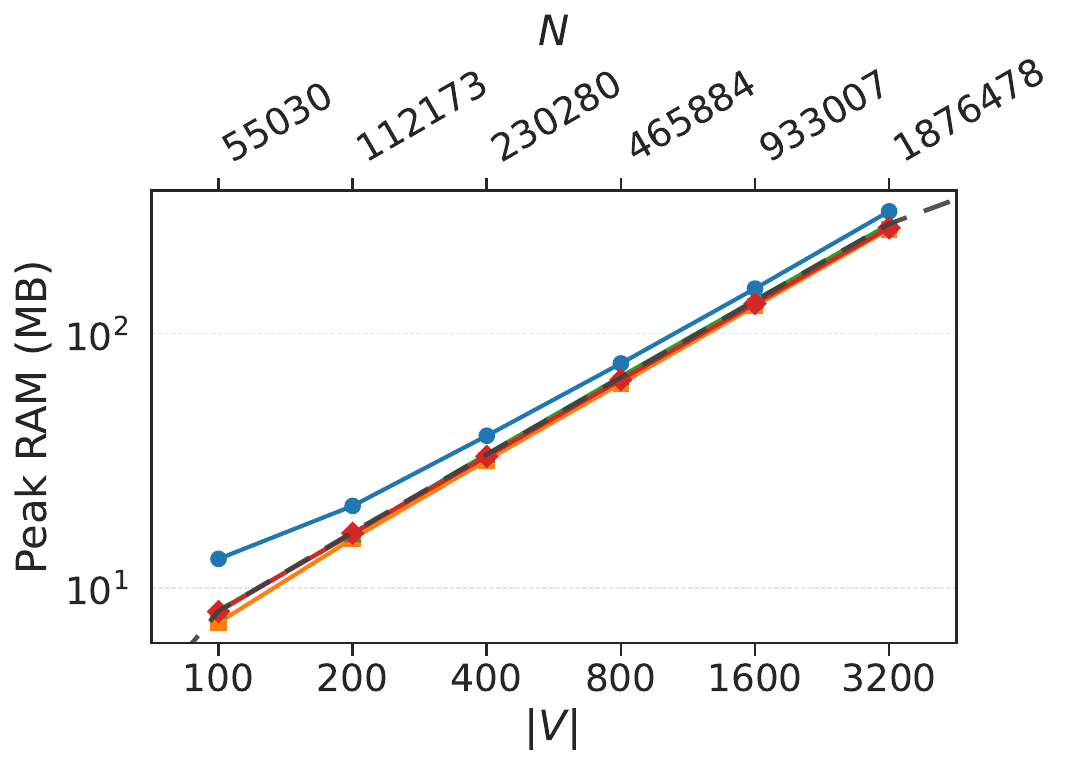}
  \end{subfigure}

  \vspace{0.25cm}

  \includegraphics[width=0.9\linewidth]{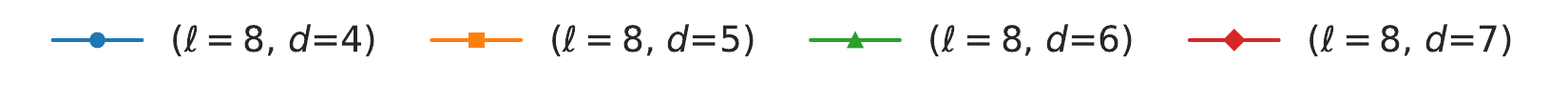}\\[2pt]
  \begin{subfigure}[t]{0.49\linewidth}
    \centering
    \includegraphics[width=\linewidth]{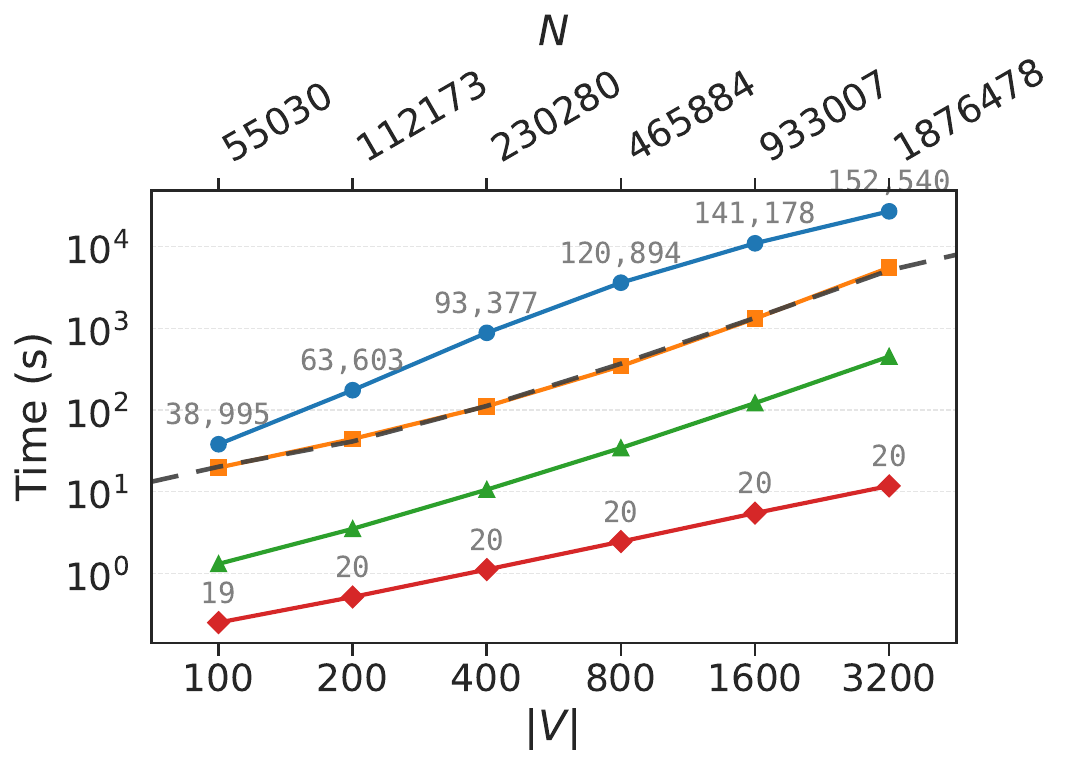}
  \end{subfigure}
  \hfill
  \begin{subfigure}[t]{0.49\linewidth}
    \centering
    \includegraphics[width=\linewidth]{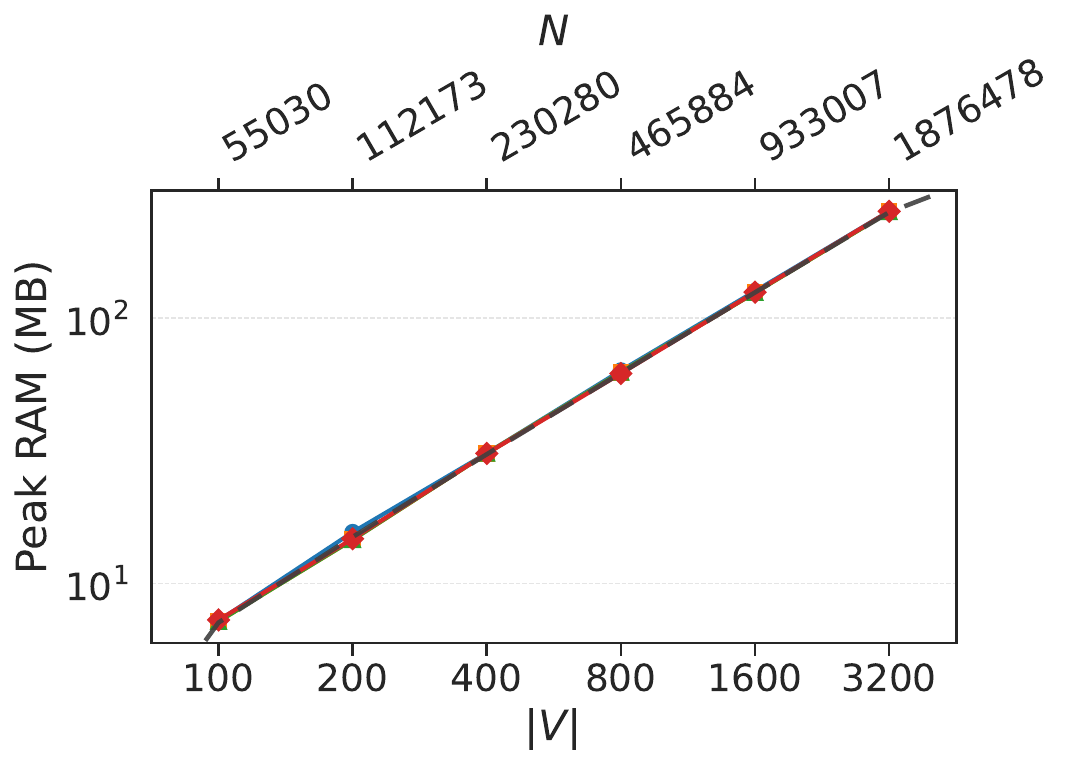}
  \end{subfigure}
  \caption{Wall-clock time and peak RAM usage per thread of exact algorithm for $\ell=8$, $d \in [0,  7]$, $\ed(G)=5\%$, and varying $|V|$ (and $N$). We use log-log axes with the dashed lines indicating a polynomial of degree 2 (resp.\ degree 1) fitting the data for time (resp.\ peak RAM). Top: $d \in [0 , 3]$ with dashed line on $d=2$. Bottom: $d \in [4 , 7]$ with dashed line on $d=5$, and lines for $d=4$ and $d=7$ labeled with $M_V$.}
  \label{fig:rq3_ell8}
\end{figure}

\begin{figure}[t]
  \centering
  \includegraphics[width=0.9\linewidth]{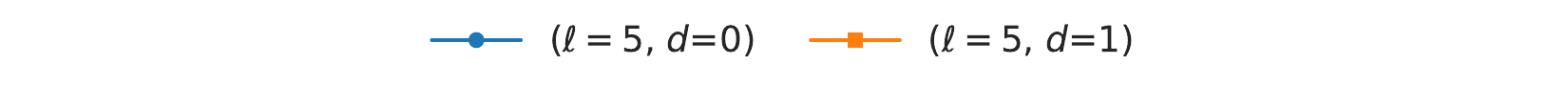}\\[2pt]
  \begin{subfigure}[t]{0.49\linewidth}
    \centering
    \includegraphics[width=\linewidth]{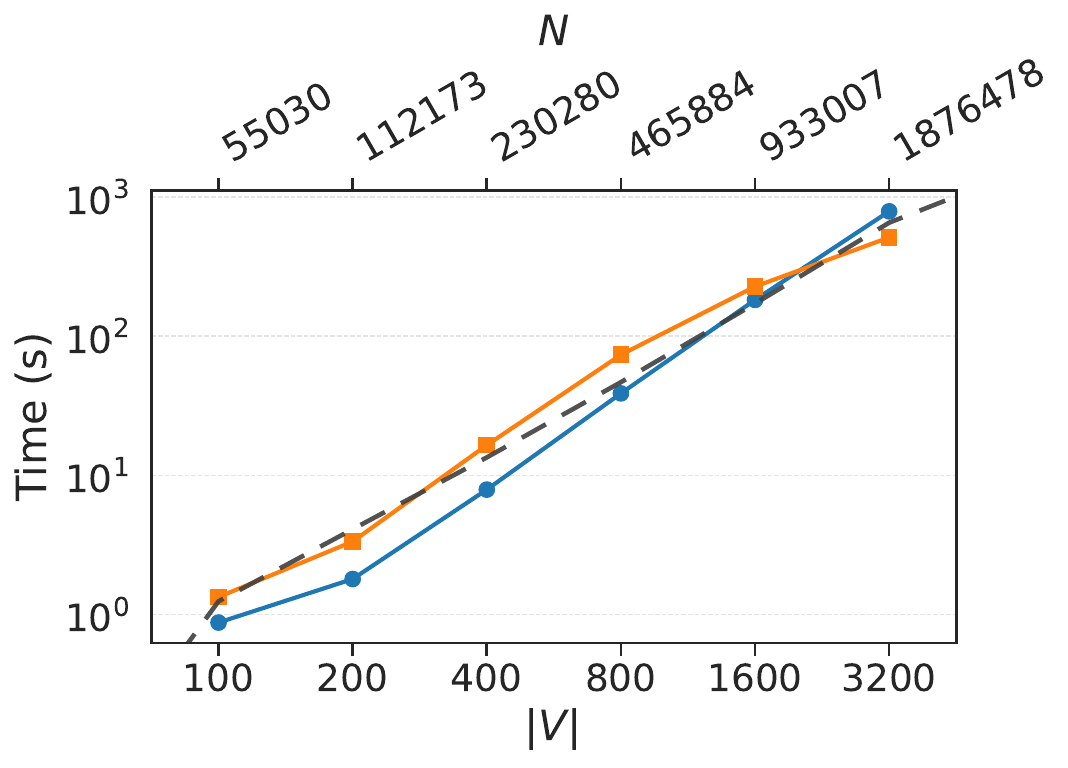}
  \end{subfigure}
  \hfill
  \begin{subfigure}[t]{0.49\linewidth}
    \centering
    \includegraphics[width=\linewidth]{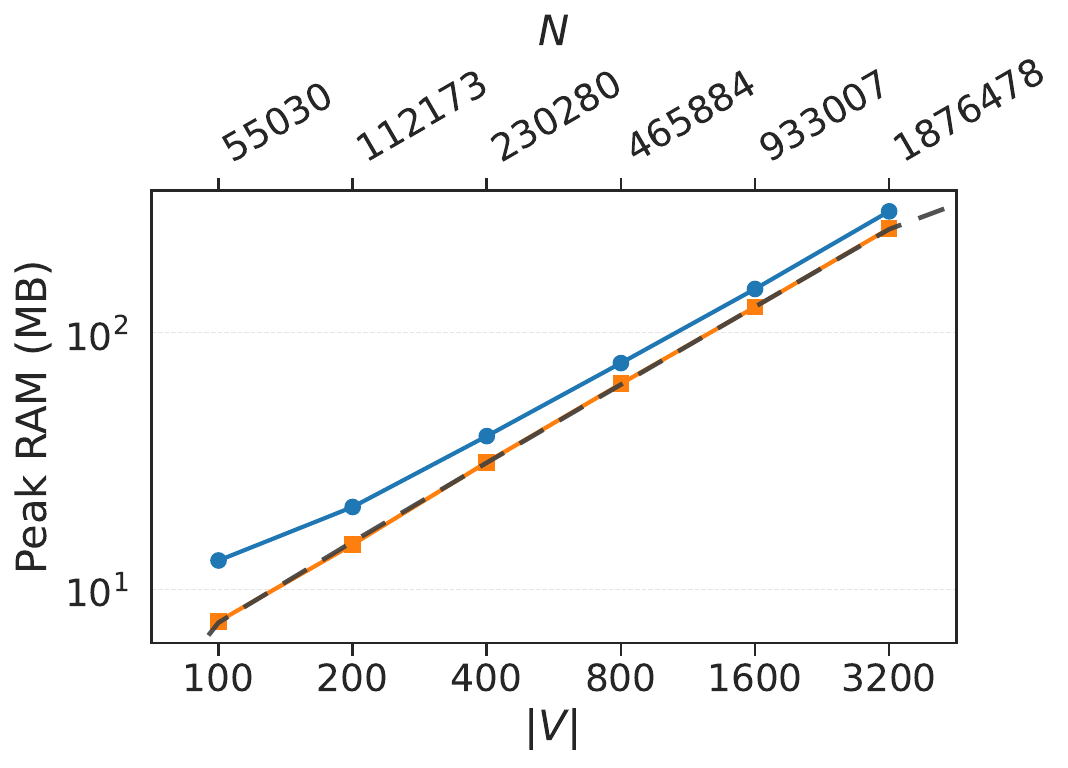}
  \end{subfigure}

  \vspace{0.25cm}
  \includegraphics[width=0.9\linewidth]{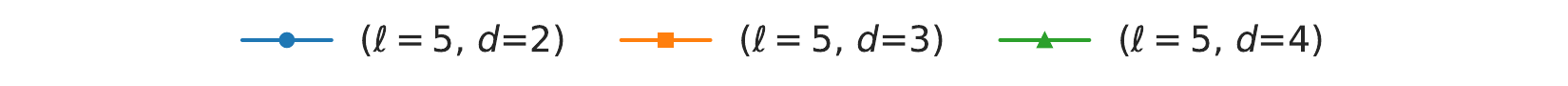}\\[2pt]
  \begin{subfigure}[t]{0.49\linewidth}
    \centering
    \includegraphics[width=\linewidth]{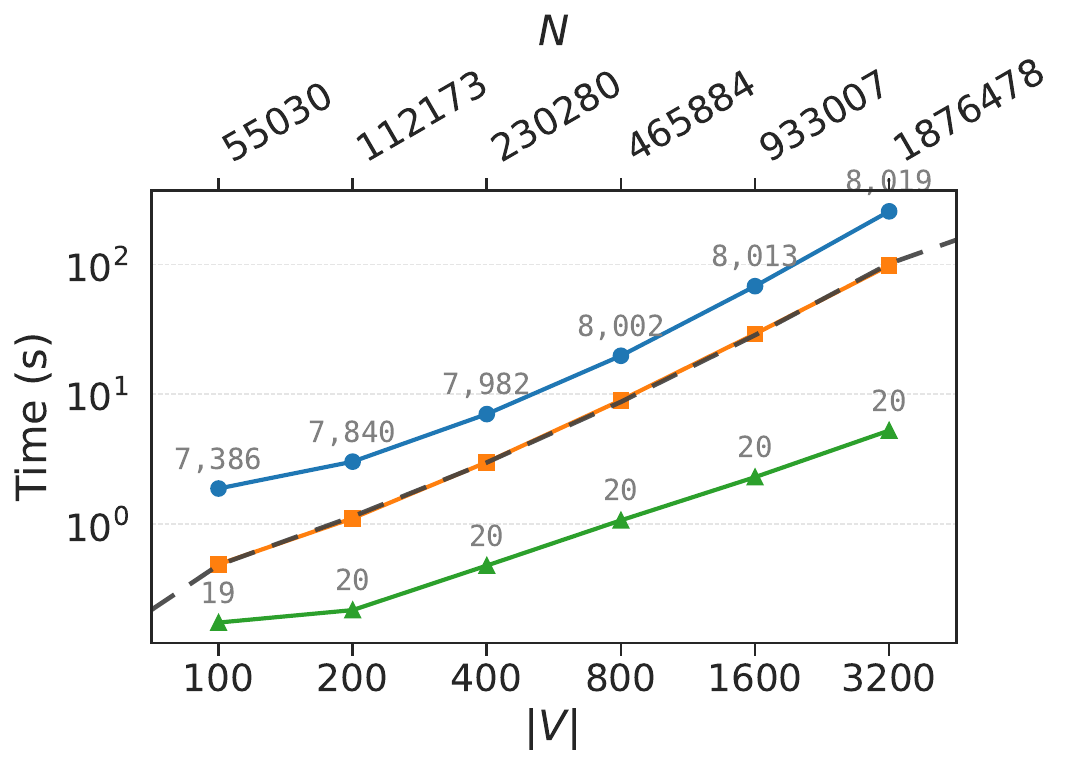}
  \end{subfigure}
  \hfill
  \begin{subfigure}[t]{0.49\linewidth}
    \centering
    \includegraphics[width=\linewidth]{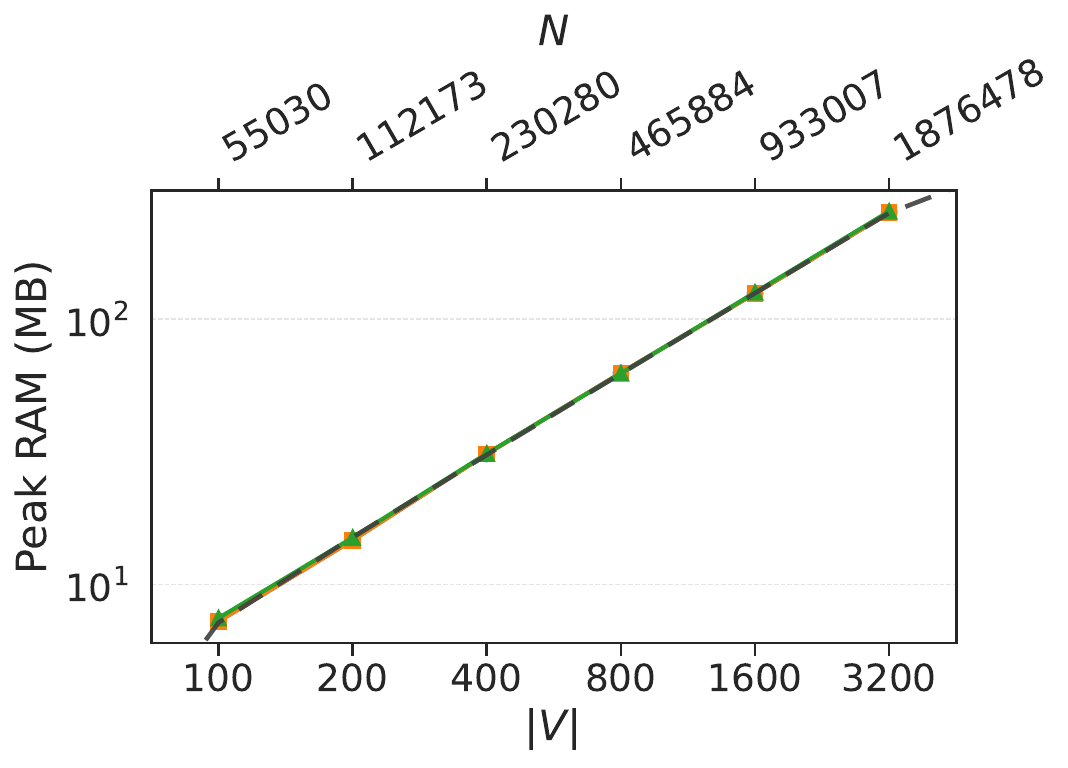}
  \end{subfigure}
  \caption{Wall-clock time and peak RAM usage per thread of exact algorithm for $\ell=5$, $d \in [0,  4]$, $\ed(G)=5\%$, and varying $|V|$ (and $N$). We use log-log axes with the dashed lines indicating a polynomial of degree 2 (resp.\ degree 1) fitting the data for time (resp.\ peak RAM). Top: $d \in [0 , 1]$ with dashed line on $d=1$. Bottom: $d \in [2 , 4]$ with dashed line on $d=3$, and lines for $d=2$ and $d=4$ labeled with $M_V$.}
  \label{fig:rq3_ell5}
\end{figure}

\subsection{Results for RQ3.}
We fixed the edge density $\ed(G)$ for all instances to $5\%$, and varied $|V|$ and $N$.
We used $t\coloneqq\min\{\binom{\ell}{d}\big(\binom{\ell}{d}+1\big)/2,P\}$ threads,
where by $P = 128$ we denote the number of threads available in our system.
For $\ell=8$, we varied $d$ from $0$ to $7$; the results are shown in \Cref{fig:rq3_ell8}.
We also performed this experiment for $\ell=5$ with $d$ varying from $0$ to $4$.
The results are shown in \Cref{fig:rq3_ell5}; they yield the same conclusion, but with smaller running times and peak RAM usage. We observe a \emph{quadratic} behavior in $N$ for the running times shown in \Cref{fig:rq3_ell8} (turned into straight lines with log-log axes for readability), for fixed values of $\ell$ and $d$.
The running times further highlight their dependence on the maximum number $M_V$ of distinct $(\ell,d)$-motifs in $V$ across all $\binom{\ell}{d}$ combinations of wildcard positions. For fixed $\ell$ and $N$, observe that $M_V$ decreases as $d$ grows, since more wildcard positions yield fewer distinct $(\ell,d)$-motifs in $V$.
Indeed, in practice, $M_V$ is significantly smaller than $N$ for larger $d$. (We label exemplary lines in \Cref{fig:rq3_ell8} with $M_V$.)
This reflects in faster running times for \emph{a fixed} wildcard combination.
However, when considering all wildcard combinations, the running times are non-decreasing for $d \in [0,\lfloor \ell/2 \rfloor)$ and non-increasing for $d \in [\lfloor \ell/2 \rfloor,\ell)$ because the term $\binom{\ell}{d}$ behaves the same way on these intervals.
The peak RAM usage shown in \Cref{fig:rq3_ell8} also confirms the \emph{linear} behavior in $N$ \emph{per thread}.

We also report the running time and peak RAM usage per thread of our exact algorithm on the datasets of \Cref{tab:string-data-preproc}; see \Cref{tab:string-data-preproc-l8} in the appendix.
We remark that for the dataset with the largest $N$, H.\ sapiens I, the exact algorithm took roughly 50h with peak RAM \num{1.29}GB per thread for parameter pair $(\ell=8,d=3)$ and only 2h with peak RAM \num{1.28}GB per thread for parameter pair $(\ell=8,d=5)$. We used $t\coloneqq P = 128$ threads for both.

\subsection{Results for RQ4.}
We fixed $\setsize{V}=200$ and varied the edge density $\ed(G)$.
The rest of the setup is analogous to RQ3: for $\ell=8$, we varied $d$ from $0$ to $7$ (\Cref{fig:rq4_ell8}); for $\ell=5$, we varied $d$ from $0$ to $4$ (\Cref{fig:rq4_ell5}); and we used $t\coloneqq\min\{\binom{\ell}{d}\big(\binom{\ell}{d}+1\big)/2,P\}$ threads. Here, our theoretical prediction for the time bound is $\cO(\binom{\ell}{d}^2|E|\max_{v \in V}|v|^2)$, which is confirmed by the running times shown in \Cref{fig:rq4_ell8}.
Indeed, by varying $\ed(G)$ and fixing the remaining parameters, the running time grows linearly with input size as $\ed(G)$ grows.
Moreover, the peak RAM usage for each thread remains constant in this regime, which is particularly appealing for sparse graphs.

\begin{figure}[t]
  \centering
  \includegraphics[width=0.9\linewidth]{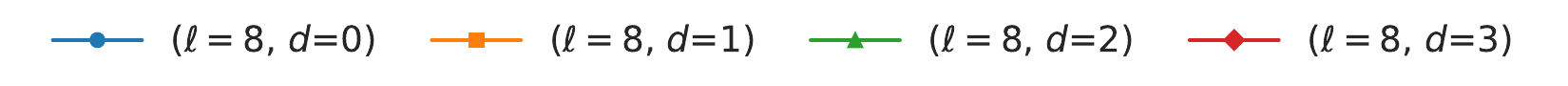}\\[2pt]
  \begin{subfigure}[t]{0.49\linewidth}
    \centering
    \includegraphics[width=\linewidth]{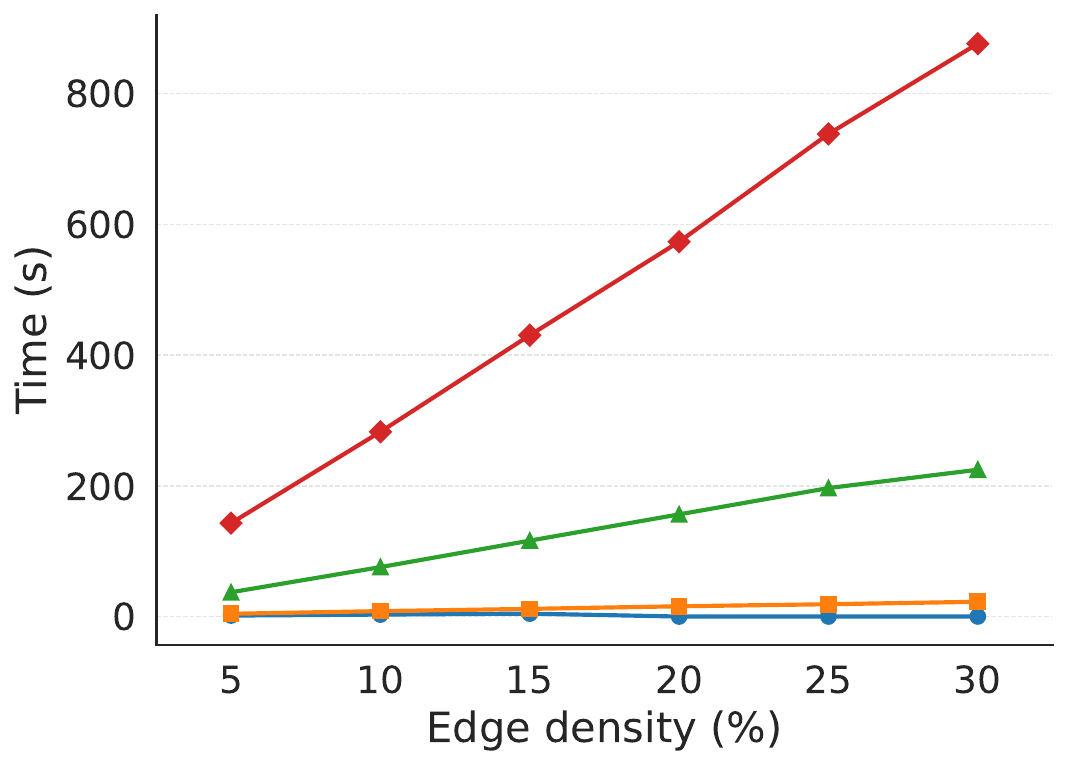}
  \end{subfigure}
  \hfill
  \begin{subfigure}[t]{0.49\linewidth}
    \centering
    \includegraphics[width=\linewidth]{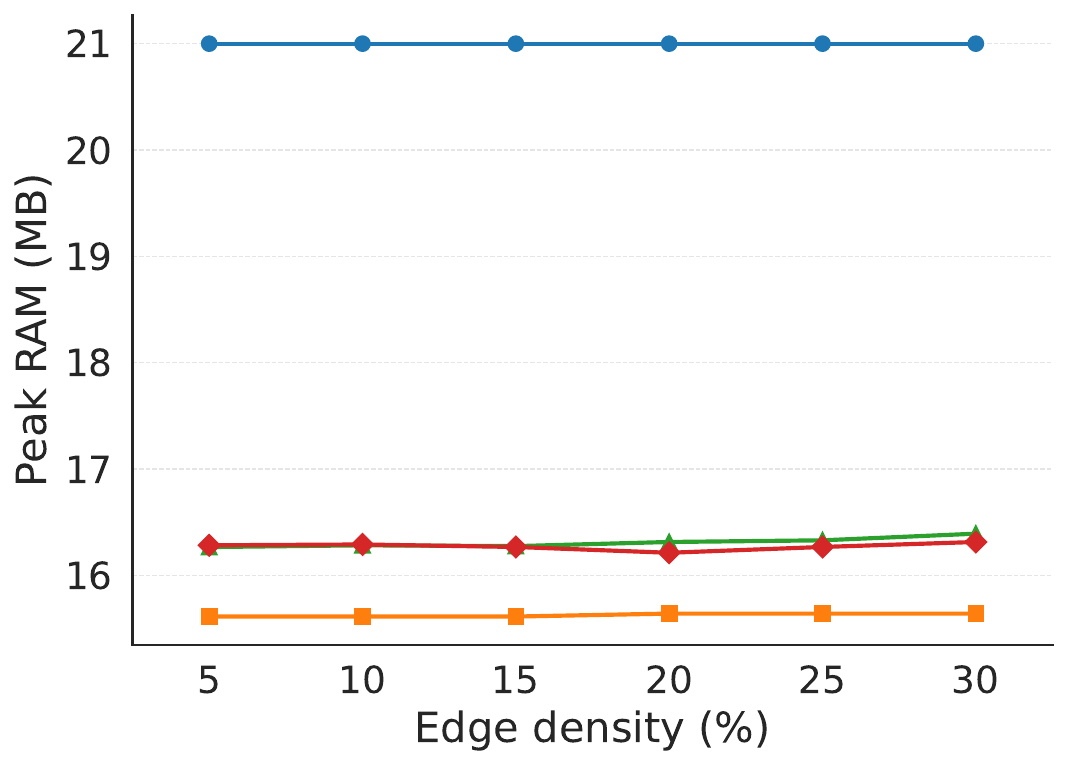}
  \end{subfigure}

  \vspace{0.25cm}
  \includegraphics[width=0.9\linewidth]{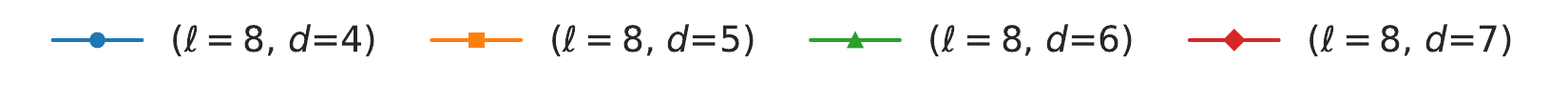}\\[2pt]
  \begin{subfigure}[t]{0.49\linewidth}
    \centering
    \includegraphics[width=\linewidth]{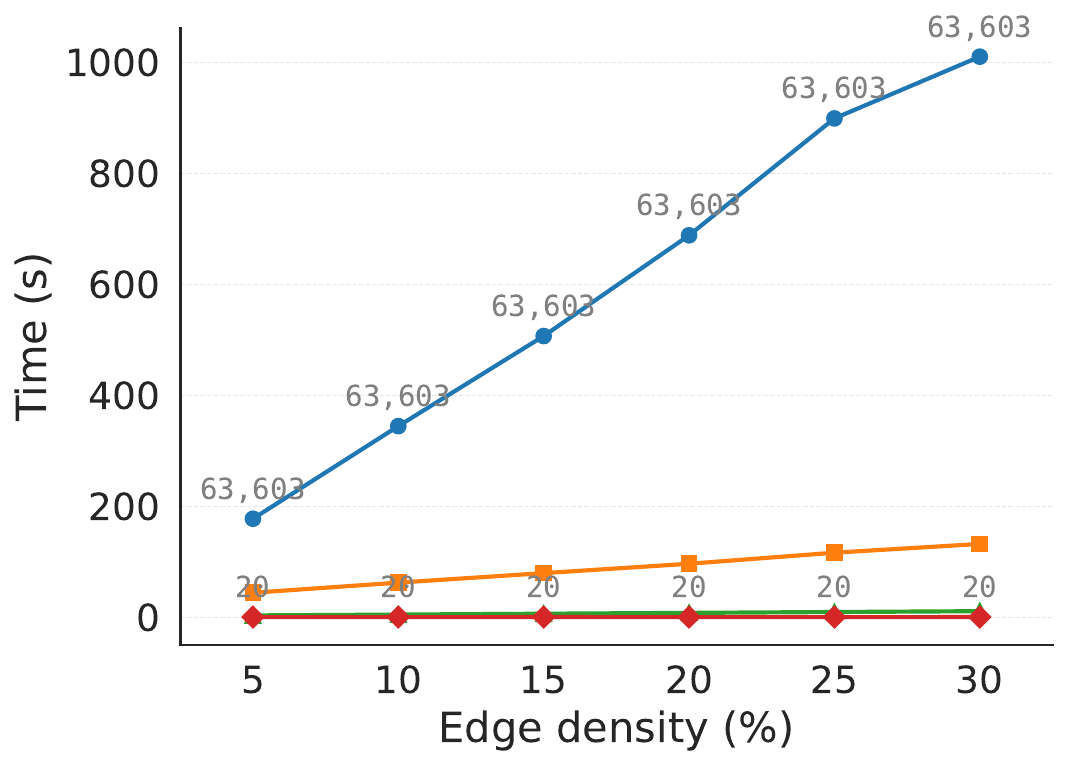}
  \end{subfigure}
  \hfill
  \begin{subfigure}[t]{0.49\linewidth}
    \centering
    \includegraphics[width=\linewidth]{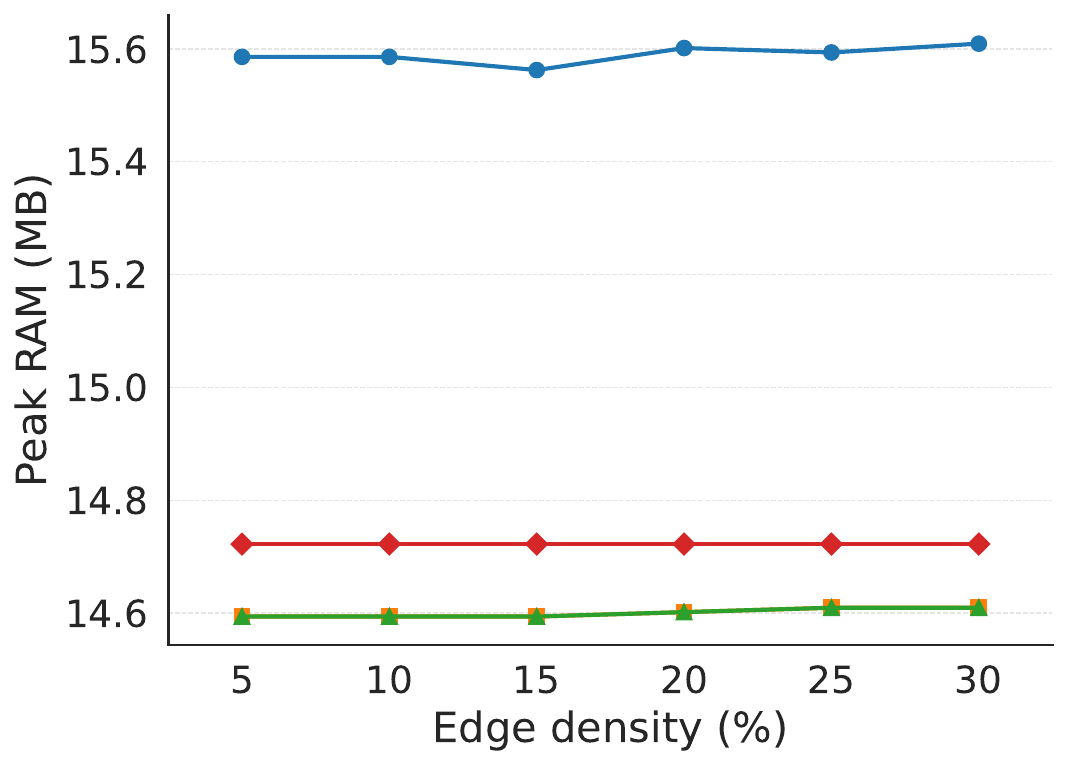}
  \end{subfigure}
  \caption{Wall-clock time and peak RAM usage per thread of exact algorithm for $\setsize{V}=200$, $\ell=8$, $d \in [0 , 7]$, and varying $\ed(G)$. Top: $d \in [0 , 3]$. Bottom: $d \in [4 , 7]$ and lines for $d=4$ and $d=7$ labeled with $M_V$.}
  \label{fig:rq4_ell8}
\end{figure}

\begin{figure}[!h]
  \centering
  \includegraphics[width=0.9\linewidth]{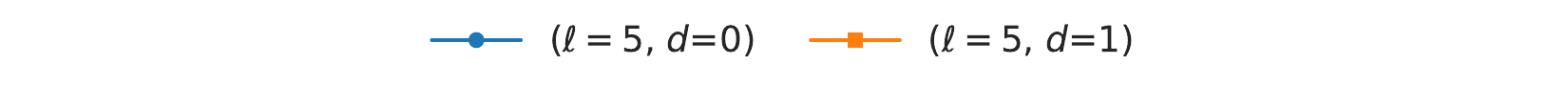}\\[2pt]
  \begin{subfigure}[t]{0.49\linewidth}
    \centering
    \includegraphics[width=\linewidth]{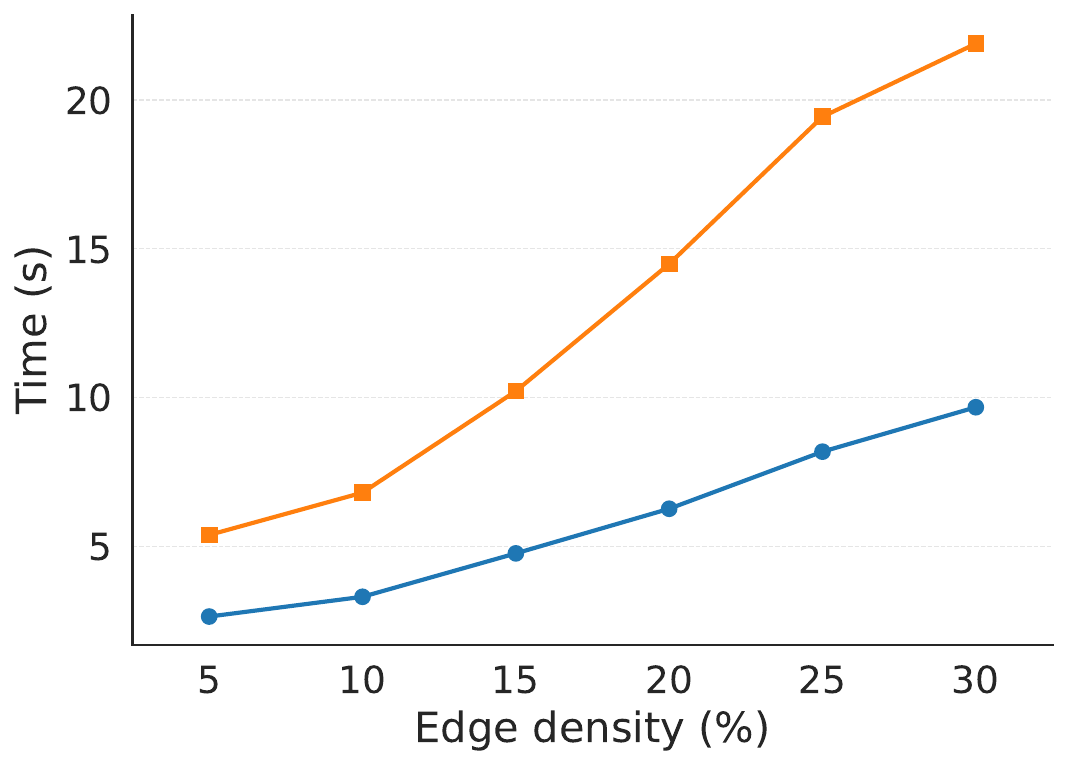}
  \end{subfigure}
  \hfill
  \begin{subfigure}[t]{0.49\linewidth}
    \centering
    \includegraphics[width=\linewidth]{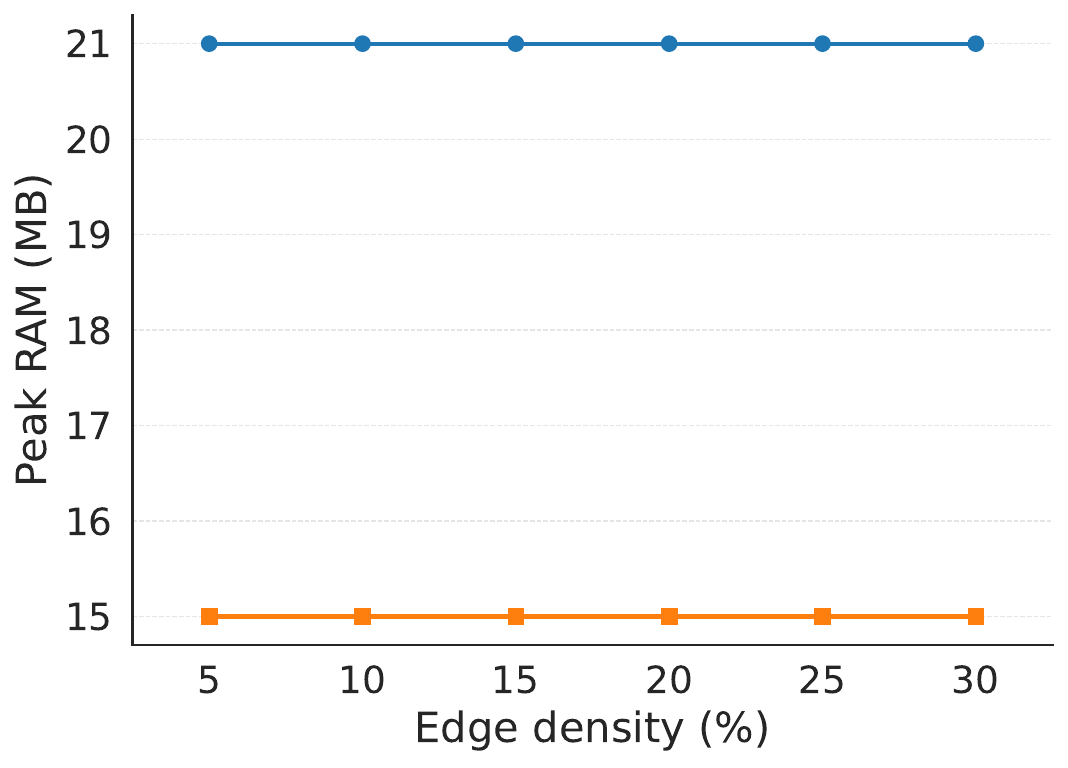}
  \end{subfigure}

  \vspace{0.25cm}
  \includegraphics[width=0.9\linewidth]{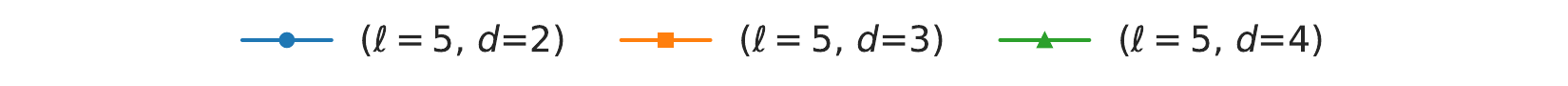}\\[2pt]
  \begin{subfigure}[t]{0.49\linewidth}
    \centering
    \includegraphics[width=\linewidth]{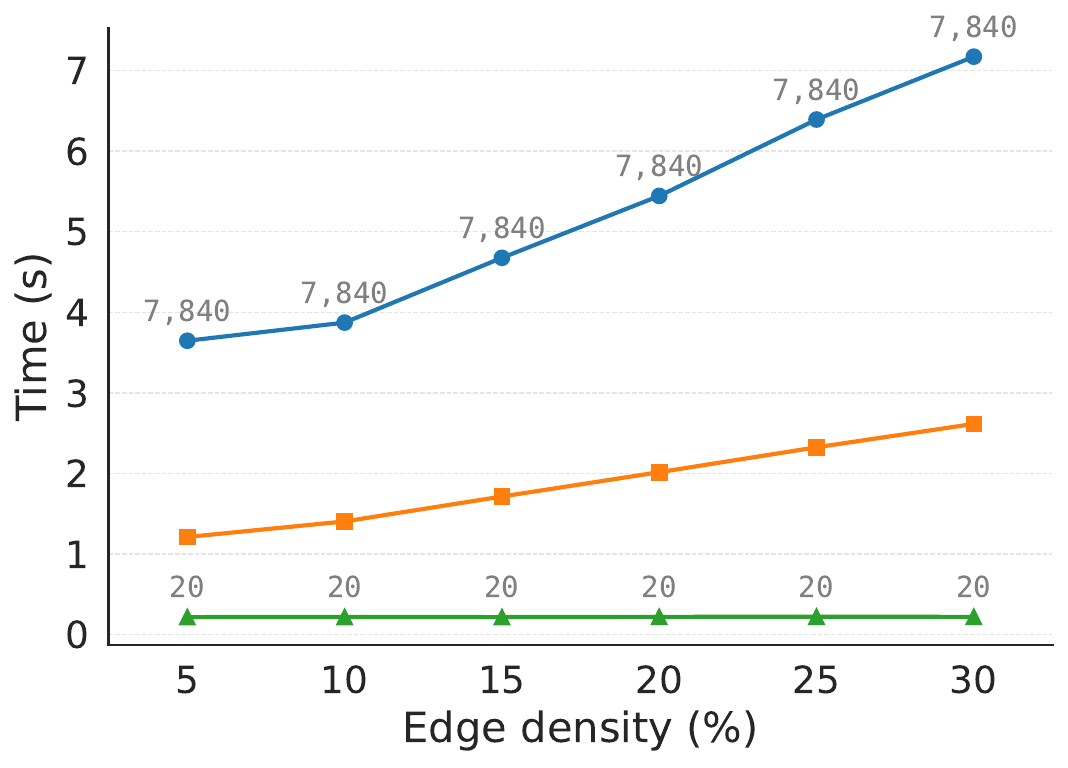}
  \end{subfigure}
  \hfill
  \begin{subfigure}[t]{0.49\linewidth}
    \centering
    \includegraphics[width=\linewidth]{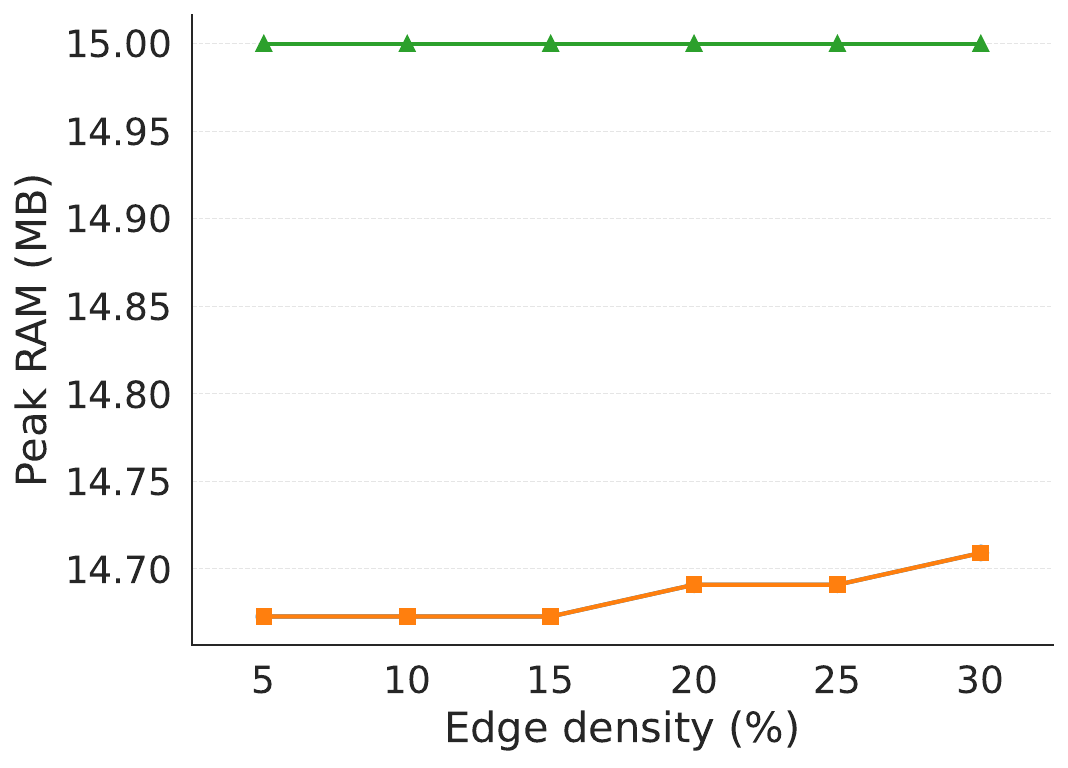}
  \end{subfigure}
  \caption{Wall-clock time and peak RAM usage per thread of exact algorithm for $\setsize{V}=200$, $\ell=5$, $d \in [0 , 4]$, and varying $\ed(G)$. Top: $d \in [0 , 1]$ and line for $(\ell=5,d=1)$ labeled with $M_V$. Bottom: $d \in [2 , 4]$, and lines for $d=2$ and $d=4$ labeled with $M_V$.}
  \label{fig:rq4_ell5}
\end{figure}

\subsection{Results for RQ5.}
\begin{figure}[t]
  \centering
  \begin{subfigure}[t]{0.49\linewidth}
    \centering
    \includegraphics[width=\linewidth]{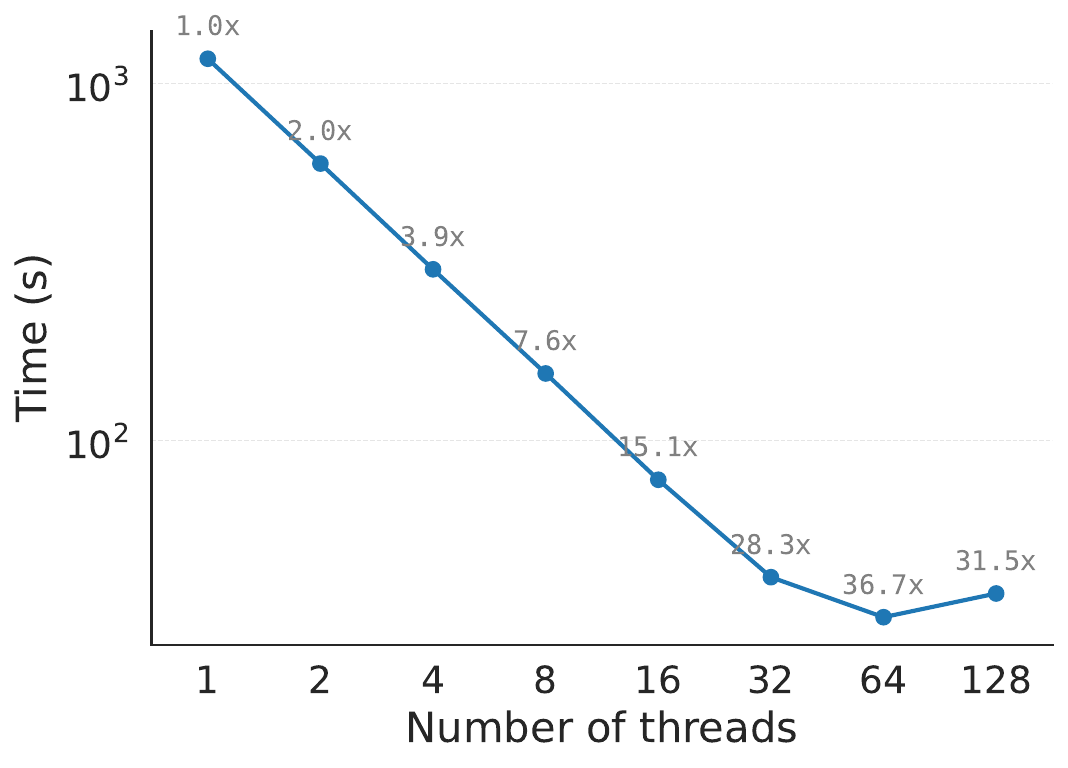}
  \end{subfigure}
  \hfill
  \begin{subfigure}[t]{0.49\linewidth}
    \centering
    \includegraphics[width=\linewidth]{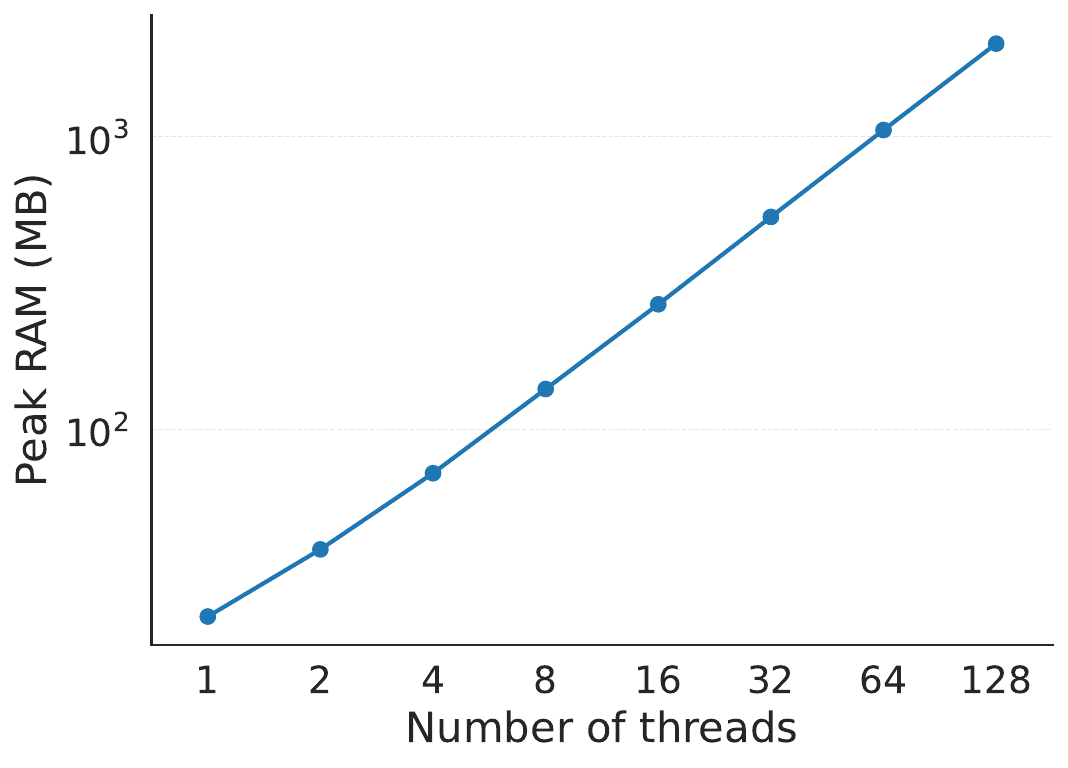}
  \end{subfigure}
  \caption{Wall-clock time and total peak RAM usage of exact algorithm for $\setsize{V}=200$, $\ed(G)=5\%$, $\ell=8$, and $d=2$ with varying number $t$ of threads. We use log-log axes and label the line with the speed-up obtained with the number of threads specified on the $x$-axis.}
  \label{fig:rq5}
\end{figure}
We fixed $\setsize{V}=200$, $\ed(G)=5\%$, $\ell=8$, $d=2$, and varied the number $t$ of threads from $1$ to $P=128$, the number available in our system, in order to assess the effectiveness of our parallelization.
\Cref{fig:rq5} shows near-linear speed-up up to 32 threads ($28.3\times$), after which returns diminish: the speed-up peaks at $36.7\times$ with 64 threads and does not improve further.
We conjecture that these diminishing returns are due to load imbalance: some wildcard combinations are more costly than others, and at 128 threads, an expensive combination dominates the wall-clock time.
Moreover, since the algorithm parallelizes over $\binom{\ell}{d}\big(\binom{\ell}{d}+1\big)/2$ wildcard combinations, at most that many threads can be productively used; since we test at most 128, this limit is never reached.
Finally, the space complexity becomes $\cO(tN)$ because each thread requires $\cO(N)$ working space, and we have $t$ parallel threads.

\section*{Reproducibility and source code.}
The source code and datasets utilized in this study are currently available via \url{https://github.com/lucper/cmm}. All experimental results presented in this paper are fully reproducible using these provided resources.

\section*{AI use disclosure.} All theoretical concepts and core proofs originated entirely from the authors' own work. The authors employed AI tools to polish language and style, support verification of proofs, produce LaTeX code for \Cref{fig:motifpairs}, \Cref{fig:maxip-construction}, and~\Cref{fig:chisup-construction}, and produce code for plotting data. Every AI-generated output was carefully reviewed, verified, and revised by the authors, who assume full responsibility for the paper's content.

\section*{Acknowledgments.}
Lucas P.~Oliveira is funded by grant~\#2025/08034-9, São Paulo Research Foundation (FAPESP).

\clearpage
\FloatBarrier
\bibliographystyle{siamplain}
\bibliography{references}

\appendix

\clearpage

\section{Materials omitted from~\Cref{sec:exact}.} \label{app:exact}

\subsection{The $f_{\chi^2}$ support measure.}
While the function $f_E$ lends itself well to our theoretical analysis, it is not ideal for the biological domain.
In particular, in a PPI network, maximizing the $f_E$ function finds a motif pair $\{X,Y\}$ that potentially explains many interactions \ie{there are many edges where $X$ appears on one side and $Y$ on the other}.
To make it more likely we find an actual interaction, we also want the property that there are few \emph{non}-edges where $X$ appears on one side and $Y$ on the other.
In other words, we want the presence of $X$ and $Y$ to make an interaction likely.
This is not captured by $f_E$.
In graph-theoretic terms, we want to ensure that the subgraph selected by $\{X,Y\}$ is \emph{dense}.
Here, this means denser than the average over the whole interaction network.
This is exactly what the $\chisup$ support function is designed to capture.
Boyen et al.~``experimentally establish the superiority of the Chi-square-based support measure over other support measures'' \cite{DBLP:journals/tcbb/BoyenDNHD11}.

Fortunately, the hard part in computing $\chisup$ turns out to be computing $f_E$.
In the following corollary, we show that the additional information needed to derive $\chisup(G,\{X,Y\})$ from $f_E(G, \{X,Y\}) = \setsize{E_{X,Y}}$ is easy to compute, lifting \Cref{alg:basic} to support $\chisup$.

\begin{corollary}[to \Cref{thm:basic}]
There is an algorithm for \CMM with $f_{\chi^2}$ running in time $\cO(N^2)$, where $N=\sum_{v \in V}|v|$.
The space used is $\cO(N)$.
\label{cor:basic-chisup}
\end{corollary}
\begin{proof}
We sketch how to adapt \Cref{alg:basic}.
Given a graph $G=(V,E)$, where $V \subseteq \Sigma^*$, and a positive integer $\ell$, our goal is to find a pair $\{X,Y\}$ of motifs maximizing $\chisup(G,\{X,Y\})$.
Computing this quantity reduces to computing the following parameters: the number $k_X$ (resp.\ $k_Y$) of nodes in which $X$ (resp.\ $Y$) occurs, the number $k_{X,Y}$ of nodes in which both $X$ and $Y$ occur, and the number $|E_{X,Y}|$ of edges in which $X$ occurs in an endpoint and $Y$ occurs in the other.

The number $|E_{X,Y}|$ of edges can be easily obtained in $\cO(N^2)$ time using \Cref{thm:basic}.
For $k_X$ (resp.\ $k_Y$), iterate over the nodes and, for each distinct length-$\ell$ substring $X$ of the current node, increment a counter indexed by $\rank(X)$.
By handling duplicates locally at each node, this counting can be done in $\cO(N)$ time.
This is a simplification of Algorithm~\ref{alg:basic}.
Finally, for $k_{X,Y}$, construct a graph $G'$ from $G$ as follows. Let $V(G')=V$ and $E(G')=\{\{v,v\} \mid v \in V\}$. Applying Algorithm~\ref{alg:basic} to $G'$ yields the desired number $k_{X,Y}$ for every pair of length-$\ell$ substrings $\{X,Y\}$. (\Cref{alg:basic} and its analysis apply to graphs with self-loops).
\end{proof}

\subsection{Why we cannot do better.}
In this section, we use tools from fine-grained complexity to establish conditional hardness.
\SETH~\cite{DBLP:journals/jcss/ImpagliazzoPZ01} is a standard assumption about the complexity of deciding the satisfiability of conjunctive normal form formulas.
We show that a polynomially faster algorithm for \CMM (in $N$) yields a faster algorithm for the following problem, which in turn is known to refute \SETH~\cite{DBLP:journals/tcs/Williams05,DBLP:conf/coco/Chen18}.

\defproblem{Maximum Inner Product (\MaxIP)}
{Two sets $A,B\subseteq\{0,1\}^m$ of $n$ vectors each.}
{The value
$\max_{a\in A,\ b\in B}\langle a,b\rangle$.}

Recall that $\langle \cdot, \cdot \rangle$ denotes the inner product.
The hardness conjecture for the \MaxIP problem states that for every constant $\varepsilon>0$, there exists a constant $c(\varepsilon)>0$
such that \MaxIP with dimension $m=c(\varepsilon)\log n$ cannot be solved in time
$\cO(n^{2-\varepsilon})$.
Refuting this conjecture refutes \SETH.
With these tools, we are ready to prove our conditional hardness result.

\newcommand{\mipX}{\textcolor{VUBlue}{\mathtt{011}}}
\newcommand{\mipY}{\textcolor{VUOrange}{\mathtt{101}}}
\tikzset{
  miplab/.style={draw, fill=white, rounded corners=1pt, minimum width=2.6cm, minimum height=0.6cm, inner sep=2pt},
  miphit/.style={ultra thick, draw=VUBlue},
  mipmiss/.style={thick, draw=black!50},
  mipnote/.style={font=\footnotesize\itshape, text=black!60}
}
\newcommand{\mipcore}{%
  \node[miplab, label=left:$u_1$] (u1) at (0,1.6) {$\mathtt{001}\#_1\,\mipX\#_1$};
  \node[miplab, label=left:$u_2$] (u2) at (0,0.8) {$\mathtt{010}\#_2\,\mipX\#_2$};
  \node[miplab, label=left:$u_3$] (u3) at (0,0.0) {$\mathtt{001}\#_3\,\mathtt{010}\#_3$};
  \node[miplab, label=right:$w_1$] (w1) at (4.2,1.6) {$\mipY\#_1$};
  \node[miplab, label=right:$w_2$] (w2) at (4.2,0.8) {$\mipY\#_2\,\mathtt{111}\#_2$};
  \node[miplab, label=right:$w_3$] (w3) at (4.2,0.0) {$\mathtt{110}\#_3\,\mathtt{111}\#_3$};
  \draw[miphit]  (u1) -- (w1);
  \draw[miphit]  (u2) -- (w2);
  \draw[mipmiss] (u3) -- (w3);
}

\begin{figure}[t]
\centering
\begin{tikzpicture}[font=\small]
\mipcore
\end{tikzpicture}
\caption{The construction in the proof of \Cref{thm:maxip-lb} for $n=m=3$, $a_1=(1,0,1)$, $a_2=(0,1,1)$, $a_3=(1,1,0)$, $b_1=(1,1,0)$, $b_2=(0,0,1)$, and $b_3=(0,1,1)$.
Here, $\ell=3$, and $\mathtt{enc}(1)=\mathtt{01}$, $\mathtt{enc}(2)=\mathtt{10}$, $\mathtt{enc}(3)=\mathtt{11}$.
The motif pair $\{\textcolor{VUBlue}{\mathtt{011}},\textcolor{VUOrange}{\mathtt{101}}\}$, which encodes $a_3$ and $b_1$, occurs in the endpoints of the two highlighted edges, so $f_E(G,\{\mathtt{011},\mathtt{101}\})=2=\langle a_3,b_1\rangle$.}
\label{fig:maxip-construction}
\end{figure}

\begin{theorem}
\label{thm:maxip-lb}
There is no algorithm solving \CMM with $f_E$ in time $\cO(N^{2-\varepsilon} \poly(\setsize{E}, \setsize{V}))$, for any $\varepsilon > 0$, unless \SETH fails.
\end{theorem}
\begin{proof}
We reduce the decision version of \MaxIP to the decision version of \CMM; since the decision version is no harder than \CMM itself, the lower bound carries over.
Let $A=\{a_1,\ldots,a_n\}$ and $B=\{b_1,\ldots,b_n\}$ be sets of $m$-dimensional binary vectors, where $m$ and $n$ are positive integers.
Let $\mathtt{enc} : \mathbb{Z}_{\geq 1} \to \{0,1\}^+$ denote the usual  binary encoding of any positive integer and $a[k]$ (resp.\ $b[k]$) the $k$-th coordinate of vector $a$ (resp.\ $b$).
Our goal is to determine whether there are vectors $a \in A$ and $b \in B$ such that $\sum_{i=1}^{m}a[i]b[i] \geq p$, where $p$ is a non-negative integer.

We construct an instance $\langle G=(V,E),\ell \rangle$ of \CMM with $V = U \cup W$ as follows.
Let $U=\{u_1,\ldots,u_m\}$ and $W=\{w_1,\ldots,w_m\}$ be sets of strings over the alphabet $\Sigma=\{0,1\} \cup \{\#_j \mid 1 \leq j \leq m\}$, and $E=\{\{u_j,w_j\} \mid 1 \leq j \leq m\}$.
If some $u_j$ or $w_j$ is the empty string, we omit that node together with the edge $\{u_j,w_j\}$; such a coordinate $j$ contributes to no inner product, so this does not affect the argument below.
For each $1 \leq j \leq m$, let
\[
    u_j = \bigodot_{\substack{1 \leq i \leq n \\ a_i[j]=1}} 0 \cdot \mathtt{enc}(i)\#_j \quad\text{and}\quad w_j = \bigodot_{\substack{1 \leq i \leq n \\ b_i[j]=1}} 1 \cdot \mathtt{enc}(i)\#_j,
\]
where $\odot$ denotes string concatenation.
That is, the node corresponding to each index is labeled by an encoding of the vectors that have that index set to 1.
\Cref{fig:maxip-construction} illustrates the construction.
Notice that our encoding requires exactly $\lceil \log (n + 1) \rceil$ bits (after padding with zeros); hence, we set $\ell\coloneqq\lceil \log (n + 1) \rceil + 1$.

By construction, if there are vectors $a_i \in A$ and $b_j \in B$ such that $\sum_{k=1}^{m}a_i[k]b_j[k] \geq p$, then the strings $0 \cdot \mathtt{enc}(i)$ and $1 \cdot \mathtt{enc}(j)$ are substrings of endpoints of at least $p$ edges of $G$.
In particular, notice that, for any $j \in [m]$, we can choose a pair $(X_j,Y_j)$ of length-$\ell$ binary strings such that $X_j$ and $Y_j$ are substrings of nodes $u_j$ and $w_j$, respectively, and none of them contain the separator symbol $\#_j$ (unless the $j$-th bit is unset in all vectors in at least one of $A$ and $B$, in which case we can simply ignore it); hence, they encode vectors whose inner product equals at least 1.

Conversely, suppose that strings $X$ and $Y$ are length-$\ell$ substrings of endpoints of at least $p$ edges in an instance of \CMM obtained by our construction.
If either string contains a separator symbol $\#_k$, then it occurs only in $u_k$ and $w_k$, so $\{u_k,w_k\}$ is the only edge whose endpoints can contain both strings, and the support value of $\{X,Y\}$ is at most $1$; in that case we may replace $\{X,Y\}$ by a pair $(X_k,Y_k)$ as described in the previous paragraph.
Thus, we assume that $X$ and $Y$ are length-$\ell$ binary strings.
Since all edges connect one node from $U$ with one from $W$ and by the choice of the labels, we may also assume that $X$ starts with a $0$ (and thus occurs only in $U$) and $Y$ starts with a $1$ (and thus occurs only in $W$).
Since $\ell=\lceil \log (n+1) \rceil + 1$, the strings $X$ and $Y$ without their leading bit are the binary encodings of integers $i,j \in \{1,\ldots,n\}$; let $a_{i} \in A$ and $b_{j} \in B$ be the corresponding vectors.
By construction, $X$ occurs in $u_k$ if and only if $a_i[k]=1$, and $Y$ occurs in $w_k$ if and only if $b_j[k]=1$.
Hence, the edges whose endpoints contain $X$ and $Y$ are exactly the edges $\{u_k,w_k\}$ with $a_{i}[k]=b_{j}[k]=1$, and therefore $\sum_{k=1}^m a_{i}[k]\,b_{j}[k] = |E_{X,Y}| \geq p$.

It is easy to calculate that for this instance, $N \in \cO(m \cdot n \log n), \setsize{V} \in \cO(m)$, and $\setsize{E} \in \cO(m)$, thus giving the claimed lower bound under \SETH.
\end{proof}

As for the lower bound with $f_{\chi^2}$, some care is needed to carry the result over from $f_E$, since $\chisup(G,\{X,Y\})$ compares $|E_{X,Y}|$ with the expected number $\overline{E_{X,Y}}$ of edges, which depends on $k_X$, $k_Y$, and $k_{X,Y}$.
Call a length-$\ell$ string \emph{plain} if it is a string over $\{0,1\}$, and \emph{special} otherwise.
In the construction of \Cref{thm:maxip-lb}, the plain motifs are the encodings $0\cdot\mathtt{enc}(i)$ and $1\cdot\mathtt{enc}(i)$ of the vectors, whereas every special motif contains a separator symbol.
This construction has two shortcomings with respect to $\chisup$.
First, $X=0\cdot\mathtt{enc}(i)$ occurs in as many nodes as $a_i$ has ones (and analogously for $1\cdot\mathtt{enc}(i)$ and $b_i$), so the expected number of edges varies between pairs of plain motifs, and $\chisup$ does not rank them by their inner product alone.
Second, $\chisup$ favors motif pairs that occur in very few nodes, since their expected number of edges is tiny and a single edge already yields a large score; special motifs are of this kind.
We neutralize both effects by adding isolated nodes to the construction in the proof of \Cref{thm:maxip-lb}.

Both shortcomings, as well as our remedy, concern how $\chisup$ relates to the number $|E_{X,Y}|$ of edges and to the expected number $\overline{E_{X,Y}}$ of edges; by definition, $\chisup(G,\{X,Y\})$ depends only on these two quantities, so it is convenient to view it as a function of them.
We need two of its properties: if the expected number of edges is fixed and below~$1$, then more edges yield a higher score, which lets pairs of plain motifs be ranked by their inner product; and if there is a single edge, then a larger expected number of edges yields a lower score, which lets pairs of plain motifs beat pairs containing a special motif.
Formally, for an integer $t\ge0$ and a real number $x>0$, let
\[
  s(t,x)\coloneqq
  \begin{cases}
    \dfrac{(t-x)^2}{x} & \text{if $t>x$},\\[6pt]
    0 & \text{otherwise},
  \end{cases}
\]
so that $\chisup(G,\{X,Y\})=s(|E_{X,Y}|,\overline{E_{X,Y}})$ whenever $\overline{E_{X,Y}}>0$.
The two properties are then as follows.

\begin{observation}\label{obs:chi-score}
Let $0<x<1$. Then
\begin{enumerate}[label=(\alph*)]
  \item $s(0,x)=0$, and $s(t,x)$ is strictly increasing in~$t$ for integers $t\ge1$;
  \item $s(1,y)<s(1,x)$ for every real number $y>x$.
\end{enumerate}
\end{observation}
\begin{proof}

For (a), we have $s(0,x)=0$ as $0<x$, and $s(t,x)=(t-x)^2/x$ for $t\ge1>x$, which is strictly increasing in $t\ge x$. For (b), if $y\ge1$, then $s(1,y)=0<s(1,x)$ because $x<1$; otherwise, we have $s(1,z)=g(z)$ for $z\in\{x,y\}$, where $g(z)\coloneqq(1-z)^2/z=1/z-2+z$.
Since $0<x<y<1$, we have $xy<1$, and thus $s(1,x)-s(1,y)=(1/x-1/y)+(x-y)=(y-x)\bigl(1/(xy)-1\bigr)>0$.
\end{proof}

With \Cref{obs:chi-score} at hand, we are ready to prove the lower bound with $f_{\chi^2}$.

\begin{figure}[t]
\centering
\begin{tikzpicture}[font=\small]
\mipcore
\node[mipnote] at (2.1,-0.6) {complement nodes};
\node[miplab, label=left:$u'_1$] at (0,-1.2) {$\mathtt{010}\#'_1$};
\node[miplab, label=left:$u'_2$] at (0,-2.0) {$\mathtt{001}\#'_2$};
\node[miplab, label=left:$u'_3$] at (0,-2.8) {$\mipX\#'_3$};
\node[miplab, label=right:$w'_1$] at (4.2,-1.2) {$\mathtt{110}\#'_1\,\mathtt{111}\#'_1$};
\node[miplab, label=right:$w'_2$] at (4.2,-2.0) {$\mathtt{110}\#'_2$};
\node[miplab, label=right:$w'_3$] at (4.2,-2.8) {$\mipY\#'_3$};
\node[mipnote] at (2.1,-3.6) {padding nodes};
\node[miplab, font=\footnotesize, label={[font=\small]left:$z^U_1$}] at (0,-4.2) {$\mathtt{01}\#_1\$^U_1\,\mathtt{1}\#_1\mathtt{0}\cdots$};
\node[inner sep=0pt] at (0,-4.76) {$\vdots$};
\node[miplab, font=\footnotesize, label={[font=\small]left:$z^U_9$}] at (0,-5.45) {$\mathtt{01}\#_1\$^U_9\,\mathtt{1}\#_1\mathtt{0}\cdots$};
\node[miplab, font=\footnotesize, label={[font=\small]right:$z^W_1$}] at (4.2,-4.2) {$\mathtt{01}\#_1\$^W_1\,\mathtt{1}\#_2\mathtt{1}\cdots$};
\node[inner sep=0pt] at (4.2,-4.76) {$\vdots$};
\node[miplab, font=\footnotesize, label={[font=\small]right:$z^W_9$}] at (4.2,-5.45) {$\mathtt{01}\#_1\$^W_9\,\mathtt{1}\#_2\mathtt{1}\cdots$};
\end{tikzpicture}
\caption{The construction in the proof of \Cref{cor:chisup-lb} for the vectors of \Cref{fig:maxip-construction}.
The complement nodes and the $m^2=9$ padding nodes on each side are isolated.
Each of $\textcolor{VUBlue}{\mathtt{011}}$ and $\textcolor{VUOrange}{\mathtt{101}}$ now occurs in exactly $m=3$ nodes, and each padding node $z^U_r$ (resp.\ $z^W_r$) contains every special length-$3$ substring of $u_1,u_2,u_3$ (resp.\ $w_1,w_2,w_3$), such as $\mathtt{1}\#_1\mathtt{0}$ (resp.\ $\mathtt{1}\#_2\mathtt{1}$).}
\label{fig:chisup-construction}
\end{figure}

\begin{corollary}\label{cor:chisup-lb}
There is no algorithm solving \CMM with $f_{\chi^2}$ in time $\cO(N^{2-\varepsilon} \poly(\setsize{E}, \setsize{V}))$, for any $\varepsilon > 0$, unless \SETH fails.
\end{corollary}

\begin{proof}
We reduce \MaxIP to \CMM with $\chisup$.
Let $A$, $B$, $m$, $n$, $\mathtt{enc}$, $\ell=\lceil\log(n+1)\rceil+1$, and the strings $u_j$, $w_j$ ($1\le j\le m$) be as in the proof of \Cref{thm:maxip-lb}.
If all vectors in $A$ or all vectors in $B$ are zero, the maximum inner product is $0$, so we assume that both sets contain a non-zero vector.

Every block $0\cdot\mathtt{enc}(i)\#_j$ of $u_j$ has length $\ell+1$, so consecutive occurrences of $\#_j$ in $u_j$ are $\ell+1$ positions apart.
Hence, a length-$\ell$ substring of $u_j$ contains at most one occurrence of $\#_j$, and if it contains none, it is the first $\ell$ symbols of a block, that is, $0\cdot\mathtt{enc}(i)$ for some $i$ with $a_i[j]=1$.
The same holds for $w_j$ with $1\cdot\mathtt{enc}(i)$ and $b_i[j]=1$.

We construct $\langle G=(V,E),\ell\rangle$ as follows; all nodes other than the $u_j$ and $w_j$ are isolated.
For each $1\le j\le m$, we add the \emph{complement nodes}
\[
    u'_j = \bigodot_{\substack{1 \leq i \leq n \\ a_i[j]=0}} 0 \cdot \mathtt{enc}(i)\#'_j \quad\text{and}\quad w'_j = \bigodot_{\substack{1 \leq i \leq n \\ b_i[j]=0}} 1 \cdot \mathtt{enc}(i)\#'_j,
\]
where $\#'_1,\ldots,\#'_m$ are fresh symbols; thus $u'_j$ (resp.\ $w'_j$) contains the encoding of $i$ if and only if $u_j$ (resp.\ $w_j$) does not.
Let $\mathcal{Z}_U$ (resp.\ $\mathcal{Z}_W$) be the set of all special length-$\ell$ substrings of $u_1,\ldots,u_m$ (resp.\ $w_1,\ldots,w_m$); by our assumption, both are non-empty.
For each $1\le r\le m^2$, we add the \emph{padding nodes}
\[
    z^U_r = \bigodot_{Z\in\mathcal{Z}_U} Z\,\$^U_r \quad\text{and}\quad z^W_r = \bigodot_{Z\in\mathcal{Z}_W} Z\,\$^W_r,
\]
where the $\$^U_r$ and $\$^W_r$ are fresh, pairwise distinct symbols.
Since each block $Z\,\$^U_r$ has length $\ell+1$, the argument above shows that every length-$\ell$ substring of $z^U_r$ contains $\$^U_r$ or belongs to $\mathcal{Z}_U$; in particular, $z^U_r$ contains no plain substring.
The same holds for $z^W_r$.
Finally, if some $u_j$, $w_j$, $u'_j$, or $w'_j$ is the empty string, we omit that node, together with the edge $\{u_j,w_j\}$ if $u_j$ or $w_j$ is omitted.
Thus $E$ consists of the edges $\{u_j,w_j\}$ with both endpoints non-empty, so $|E|\le m$.
For each $j$, at least one of $u_j,u'_j$ and at least one of $w_j,w'_j$ is non-empty, and all padding nodes are non-empty; hence $|V|\ge 2m^2+2m$.
\Cref{fig:chisup-construction} illustrates the construction.

Recall that, for a motif pair $\{X,Y\}$, $V_X$ (resp.\ $V_Y$) is the set of nodes in which $X$ (resp.\ $Y$) occurs, $k_X=|V_X|$, $k_Y=|V_Y|$, $k_{X,Y}=|V_X\cap V_Y|$, and $\overline{E_{X,Y}}=\ed(G)\cdot E^{\max}_{k_X,k_Y,k_{X,Y}}$, where $E^{\max}_{k_X,k_Y,k_{X,Y}}$ is the number of unordered pairs $\{x,y\}$ of distinct nodes with $x\in V_X$ and $y\in V_Y$.
Every motif pair $\{X,Y\}$ with $|E_{X,Y}|=0$ has score~$0$, so it remains to consider pairs with $|E_{X,Y}|\ge1$.
To compare their scores, let $\mu\coloneqq\ed(G)\cdot m^2$; we show that $\overline{E_{X,Y}}=\mu$ whenever $X$ and $Y$ are plain motifs occurring in $G$ such that $X$ starts with~$0$ and $Y$ starts with~$1$, and $\overline{E_{X,Y}}>\mu$ whenever $X$ or $Y$ is special.
Since $|E|\le m$ and $|V|\ge 2m^2+2m$, note that
\[
  0 < \mu = \frac{|E|\cdot m^2}{\binom{|V|}{2}} \le \frac{m^3}{\binom{2m^2+2m}{2}} < 1,
\]
where $\mu>0$ as $E\neq\emptyset$ (otherwise, no pair remains to be considered).
We first consider pairs of two plain motifs and then pairs in which at least one motif is special.

Every plain substring of a node is the first $\ell$ symbols of a block, hence of the form $0\cdot\mathtt{enc}(i)$ or $1\cdot\mathtt{enc}(i)$ with $1\le i\le n$.
The motif $0\cdot\mathtt{enc}(i)$ occurs in $u_j$ if $a_i[j]=1$, in $u'_j$ if $a_i[j]=0$, and in no other node, as the padding nodes contain no plain substring; analogously for $1\cdot\mathtt{enc}(i)$ with $w_j$, $w'_j$, and $b_i$.
We call a motif pair $\{X,Y\}$ a \emph{plain pair} if $X$ and $Y$ are plain motifs occurring in $G$ such that $X$ starts with $0$ and $Y$ starts with $1$.
By the above, $X=0\cdot\mathtt{enc}(i)$ and $Y=1\cdot\mathtt{enc}(i')$ for some $i,i'\in\{1,\ldots,n\}$ (possibly $i=i'$).
For a plain pair, we have $k_X=k_Y=m$ and $k_{X,Y}=0$, so $\overline{E_{X,Y}}=\mu$.
Moreover, $E_{X,Y}$ consists of the edges $\{u_j,w_j\}$ with $a_i[j]=b_{i'}[j]=1$, so $|E_{X,Y}|=\langle a_i,b_{i'}\rangle$.
Consequently, $\chisup(G,\{X,Y\})=s(\langle a_i,b_{i'}\rangle,\mu)$, which, by \Cref{obs:chi-score}(a), is $0$ if $\langle a_i,b_{i'}\rangle=0$ and strictly increasing in $\langle a_i,b_{i'}\rangle$ otherwise.
Every other pair $\{X,Y\}$ of plain motifs has $|E_{X,Y}|=0$: either $X$ and $Y$ occur in no node, or both start with the same symbol, say $0$, and thus occur only among the $u_j$ and their complement nodes, so the endpoint $w_j$ of each edge $\{u_j,w_j\}$ contains neither motif; the case for symbol $1$ is symmetric.

Now let $\{X,Y\}$, possibly with $X=Y$, be a motif pair with $|E_{X,Y}|\ge1$ such that $X$ or $Y$ is special.
Since every edge has the form $\{u_j,w_j\}$, we may assume by symmetry that $X$ occurs in $u_j$ and $Y$ occurs in $w_j$ for some edge $\{u_j,w_j\}\in E_{X,Y}$.
A special substring of $u_j$ (resp.\ $w_j$) contains $\#_j$, so it occurs in $u_j$ (resp.\ $w_j$) and in every padding node $z^U_r$ (resp.\ $z^W_r$), and in no node other than $u_j$, $w_j$, and the padding nodes.
Thus, if $X$ is special, then $\{u_j,w_j\}$ is the only edge with an endpoint in $V_X$, and likewise for $Y$; hence $|E_{X,Y}|=1$.
Moreover, $V_X$ and $V_Y$ contain disjoint subsets of sizes at least $m^2+1$ and $m$, or vice versa: $\{u_j,z^U_1,\ldots,z^U_{m^2}\}\subseteq V_X$ if $X$ is special, $\{w_j,z^W_1,\ldots,z^W_{m^2}\}\subseteq V_Y$ if $Y$ is special, and, by the case of plain motifs, $V_X$ consists of $m$ nodes among the $u_k$ and their complement nodes if $X$ is plain, and $V_Y$ consists of $m$ nodes among the $w_k$ and their complement nodes if $Y$ is plain.
Hence, $E^{\max}_{k_X,k_Y,k_{X,Y}}\ge(m^2+1)m>m^2$ and thus $\overline{E_{X,Y}}>\mu$.
Therefore, by \Cref{obs:chi-score}(b), $\chisup(G,\{X,Y\})=s(1,\overline{E_{X,Y}})<s(1,\mu)$, which is the score of a plain pair with inner product~$1$.

We now combine both cases.
Let $p^\star\coloneqq\max_{a\in A,\,b\in B}\langle a,b\rangle$, and suppose first that $p^\star\ge1$.
A plain pair $\{0\cdot\mathtt{enc}(i),1\cdot\mathtt{enc}(i')\}$ with $\langle a_i,b_{i'}\rangle=p^\star$ scores $s(p^\star,\mu)\ge s(1,\mu)>0$, whereas, by \Cref{obs:chi-score}(a), every plain pair with a smaller inner product scores less, every pair containing a special motif scores less than $s(1,\mu)$, and every remaining pair scores~$0$.
Hence, every motif pair maximizing $\chisup(G,\cdot)$ is a plain pair whose vectors have inner product~$p^\star$.
To solve \MaxIP, we compute a maximizing pair $\{X^\star,Y^\star\}$ and output $\langle a_i,b_{i'}\rangle$ if it is a plain pair $\{0\cdot\mathtt{enc}(i),1\cdot\mathtt{enc}(i')\}$, and $0$ otherwise.
If $p^\star\ge1$, the output is $p^\star$ by the above.
If $p^\star=0$, the output is $0$ either way: if the returned pair is a plain pair, its vectors have inner product~$0$, and otherwise we output~$0$ by definition.

Each of $u_j,w_j,u'_j,w'_j$ has length at most $n(\ell+1)\in\cO(n\log n)$, so $|\mathcal{Z}_U|,|\mathcal{Z}_W|\le mn(\ell+1)$ and each padding node has length at most $mn(\ell+1)^2$.
Hence, we have $N\in\cO(m^3n\log^2 n)$, $|V|\in\cO(m^2)$, and $|\Sigma|\le 2+2m+2m^2$.
The instance can be built in $\cO(n\,\mathrm{polylog}\,n)$ time, e.g., by computing $\mathcal{Z}_U$ and $\mathcal{Z}_W$ through sorting the $\cO(mn\ell)$ length-$\ell$ substrings of the $u_j$ and $w_j$.
Together with the construction, the $\cO(nm)$-time check for zero vectors, and the $\cO(m)$-time decoding, an $\cO(N^{2-\varepsilon}\poly(|E|,|V|))$-time algorithm for \CMM with $\chisup$ would thus solve \MaxIP with $m=c(\varepsilon/2)\log n$ in $\cO(n^{2-\varepsilon}\,\mathrm{polylog}\,n)\subseteq\cO(n^{2-\varepsilon/2})$ time, contradicting the hardness conjecture for \MaxIP and hence \SETH.
\end{proof}

\section{Materials omitted from~\Cref{sec:wild}.} \label{app:wild}

\subsection{Why we cannot hope for a running time polynomial in $\ell$.}

The result of \Cref{sec:wild} can be couched in parameterized complexity terms to say that $(\ell, d)$-\CMM is fixed-parameter tractable (in \fpt) by $\ell$ (since $\binom{\ell}{d} \le 2^\ell$).
Since $\ell$ is a larger parameter than $d$ (except in trivial instances), it is natural to ask if \ldCMM is in \fpt by $d$, too.
That is, whether there is an algorithm with a running time of the form $\cO(f(d) \cdot \poly(N))$, where $f$ is any computable function.
Note that we do not need to allow a separate polynomial dependence on $\ell$ as for all non-trivial instances $\ell \le N$.

\begin{observation}
$(\ell,d)$-\CMM is W[1]-hard by $d$ for both $f_E$ and $\chisup$.
\end{observation}

\begin{proof}
The result follows by checking that the reduction from~\cite{DBLP:journals/tcbb/BoyenDNHD11} is a parameterized reduction from balanced biclique, parameterized by solution size, to $(\ell,d)$-\CMM parameterized by $d$.

An outline of the reduction strategy is as follows. The reduction takes a graph as input and produces a node-labeled graph such that a motif consisting of $0$s and $d$ wildcards selects exactly the $d$ nodes at its wildcard positions, so every set of $d$ nodes can be selected (a motif containing a $1$ selects at most one node). A motif pair selecting two disjoint sets of $d$ nodes each maximizes the support function if and only if the selected subgraph is a biclique. Thus finding a motif pair with $d$ wildcards and the best possible support value is equivalent to finding a biclique $K_{d,d}$, a problem that is known to be W[1]-hard by solution size \cite{DBLP:conf/soda/Lin15}.
\end{proof}

In fact, the authors of~\cite{DBLP:journals/tcbb/BoyenDNHD11} show NP-hardness for an even larger class of support functions (which they call \emph{compliant biclique-maximal} support functions). We omit that discussion, because it is quite technical, and refer the interested reader to~\cite{DBLP:journals/tcbb/BoyenDNHD11}.

\section{Materials omitted from~\Cref{sec:exp}.} \label{app:exp}

\subsection{Datasets.}\label{a:datasets}
\paragraph{Real networks.}
We selected the datasets from the experimental evaluation of SLIDER~\cite{DBLP:journals/tcbb/BoyenDNHD11} (S.\ cerevisiae II and H.\ sapiens II) and from the STRING database~\cite{mering2003,szklarczyk2023} (E.\ coli, S.\ cerevisiae I, D.\ melanogaster, A.\ thaliana, and H.\ sapiens I).
The datasets from~\cite{DBLP:journals/tcbb/BoyenDNHD11} represent physical protein--protein interactions obtained using affinity purification and mass spectrometry~\cite{huttlin2021_human,collins2007_yeast}.
The datasets from the STRING database are obtained by aggregating different databases, so they include other types of interaction in addition to the physical ones~\cite{mering2003}.
For this reason, the edges are weighted with a confidence score ranging from $15\%$ to $99.9\%$.
Following their general guideline, we chose a conservative threshold and kept only the edges with a high confidence score, namely $70\%$ or above (STRING's high confidence threshold).
\Cref{tab:string-data} shows the input parameters for the selected datasets before edge filtering.
\begin{table}[tbhp]
\footnotesize
\caption{Top: Parameters of datasets from STRING database. Bottom: Parameters of datasets from~\cite{DBLP:journals/tcbb/BoyenDNHD11}.}
\label{tab:string-data}
\begin{center}
\begin{tabular}{lccc}
\toprule
Dataset & $|V|$ & $|E|$ & $N$ \\
\midrule
E.\ coli          & \num{4140}  & \num{492380}  & \num{1312010} \\
S.\ cerevisiae I  & \num{6538}  & \num{1412421} & \num{2906015} \\
D.\ melanogaster  & \num{13577} & \num{2516980} & \num{7423206} \\
A.\ thaliana      & \num{25739} & \num{6278747} & \num{10843854} \\
H.\ sapiens I     & \num{19622} & \num{6857702} & \num{11283623} \\
\midrule
S.\ cerevisiae II & \num{1620}  & \num{9060}    & \num{862281} \\
H.\ sapiens II    & \num{10381} & \num{69542}   & \num{5608861} \\
\bottomrule
\end{tabular}
\end{center}
\end{table}

We selected the datasets from the STRING database because they are publicly available and are relatively big for modern standards in this application domain, making them well suited for our experiments.

For all datasets in \Cref{tab:string-data}, we removed nodes containing sequences with non-amino-acid letters (it is customary to denote an unknown amino acid by the capital letter \texttt{X}) along with their incident edges and isolated nodes.
The resulting datasets of this procedure are shown in \Cref{tab:string-data-preproc}, and these were the datasets used in our experiments.
Note that all networks in \Cref{tab:string-data} and \Cref{tab:string-data-preproc} are sparse.

\paragraph{Artificial networks.}
For RQ3, RQ4, and RQ5, we generated smaller artificial networks by sampling sequences from the H.\ sapiens I network, which contains the largest sequences on average.
We then added edges uniformly at random until an edge-density threshold was satisfied and ensured that the constructed instances had no isolated nodes.

\subsection{Results for RQ2.}
\Cref{tab:rq2} shows the coverage of $S_{\Exact}$ by $S_{\SEQSlider}$ and $S_{\MSlider}$ across all datasets.
We omitted the density plots for S.\ cerevisiae I and D.\ melanogaster in \Cref{fig:rq2-density-l8d3} because $S_{\Exact}$ was fully covered by $S_{\SEQSlider}$ for these datasets (which can be verified in \Cref{tab:rq2}).

\begin{table*}[tbhp]
\centering
\footnotesize
\caption{Coverage of $S_{\Exact}$ by $S_{\SEQSlider}$ and $S_{\MSlider}$ for $h=8$ and threshold $\tau=70\%$. Values for $S_{\Exact}$ are omitted as they are (trivially) $100\%$.}\label{tab:rq2}
\begin{tabular}{ll ccc ccc}
\toprule
& & \multicolumn{3}{c}{$(\ell=8,d=3)$} & \multicolumn{3}{c}{$(\ell=8,d=5)$} \\
\cmidrule(lr){3-5}\cmidrule(lr){6-8}
Dataset & Method & \makecell{Reduced\\solution size} & \makecell{Full\\solution size} & Coverage & \makecell{Reduced\\solution size} & \makecell{Full\\solution size} & Coverage \\
\midrule
\multirow{3}{*}{E.\ coli}          & Exact      & \num{6200}  & \num{10000} & ---       & \num{8521}  & \num{10000} & ---       \\
                                   & M-SLIDER   & \num{9916}  & \num{10000} & $75.01\%$ & \num{2301}  & \num{2302} & $45.33\%$ \\
                                   & SEQ-SLIDER & \num{694}   & \num{10000} & $\mathbf{85.61\%}$ & \num{9959}  & \num{10000} & $\mathbf{93.70\%}$ \\
\midrule
\multirow{3}{*}{S.\ cerevisiae I}  & Exact      & \num{231}   & \num{10000} & ---       & \num{10000} & \num{10000} & ---       \\
                                   & M-SLIDER   & \num{9980}  & \num{10000} & $100.00\%$   & \num{363}   & \num{363} & $9.79\%$  \\
                                   & SEQ-SLIDER & \num{1337}  & \num{10000} & $\mathbf{100.00\%}$   & \num{10000} & \num{10000} & $\mathbf{77.20\%}$ \\
\midrule
\multirow{3}{*}{D.\ melanogaster}  & Exact      & \num{893}   & \num{10000} & ---       & \num{10000} & \num{10000} & ---       \\
                                   & M-SLIDER   & \num{9890}  & \num{10000} & $85.55\%$ & \num{149}   & \num{149} & $3.19\%$  \\
                                   & SEQ-SLIDER & \num{1311}  & \num{10000} & $\mathbf{100.00\%}$   & \num{10000} & \num{10000} & $\mathbf{45.17\%}$ \\
\midrule
\multirow{3}{*}{A.\ thaliana}      & Exact      & \num{3460}  & \num{10000} & ---       & \num{10000} & \num{10000} & ---       \\
                                   & M-SLIDER   & \num{10000} & \num{10000} & $4.65\%$  & \num{168}   & \num{168} & $5.76\%$  \\
                                   & SEQ-SLIDER & \num{2384}  & \num{10000} & $\mathbf{8.12\%}$  & \num{10000} & \num{10000} & $\mathbf{21.38\%}$ \\
\midrule
\multirow{3}{*}{H.\ sapiens I}     & Exact      & \num{1836}  & \num{10000} & ---       & \num{9992}  & \num{10000} & ---       \\
                                   & M-SLIDER   & \num{9996}  & \num{10000} & $\mathbf{58.27\%}$ & \num{187}   & \num{187} & $3.71\%$  \\
                                   & SEQ-SLIDER & \num{3057}  & \num{10000} & $52.34\%$ & \num{10000} & \num{10000} & $\mathbf{11.67\%}$ \\
\midrule
\multirow{3}{*}{S.\ cerevisiae II} & Exact      & \num{5989}  & \num{10000} & ---       & \num{9996}  & \num{10000} & ---       \\
                                   & M-SLIDER   & \num{9909}  & \num{10000} & $28.63\%$ & \num{9988}  & \num{10000} & $38.56\%$ \\
                                   & SEQ-SLIDER & \num{5792}  & \num{10000} & $\mathbf{78.56\%}$ & \num{9993}  & \num{10000} & $\mathbf{91.79\%}$ \\
\midrule
\multirow{3}{*}{H.\ sapiens II}    & Exact      & \num{6197}  & \num{10000} & ---       & \num{10000} & \num{10000} & ---       \\
                                   & M-SLIDER   & \num{10000} & \num{10000} & $9.95\%$  & \num{472}   & \num{472} & $4.97\%$  \\
                                   & SEQ-SLIDER & \num{4558}  & \num{10000} & $\mathbf{17.50\%}$ & \num{10000} & \num{10000} & $\mathbf{28.97\%}$ \\
\bottomrule
\end{tabular}
\end{table*}

\subsection{Results for RQ3.}
\Cref{tab:string-data-preproc-l8} shows the running time and peak RAM usage per thread of our exact algorithm on the datasets of \Cref{tab:string-data-preproc}.

\begin{table*}[t]
\footnotesize
\caption{Wall-clock time and peak RAM usage per thread for parameter pairs $(\ell=8,d=3)$ and $(\ell=8,d=5)$ run on $t=128$ threads.}\label{tab:string-data-preproc-l8}
\begin{center}
\begin{tabular}{lcccccccc}
\toprule
& & & & \multicolumn{2}{c}{$(\ell=8,d=3)$} & \multicolumn{2}{c}{$(\ell=8,d=5)$} \\
\cmidrule(lr){5-6} \cmidrule(lr){7-8}
Dataset & $|V|$ & $|E|$ & $N$ & Total time & \makecell{Peak RAM\\per thread} & Total time & \makecell{Peak RAM\\per thread} \\
\midrule
E.\ coli          & \num{4006}  & \num{32244}  & \num{1277734} & 01h:01m:40s & \num{183.8}MB  & 00h:05m:42s & \num{177.3}MB \\
S.\ cerevisiae I  & \num{5791}  & \num{104188} & \num{2744228} & 10h:28m:04s & \num{391.2}MB  & 00h:29m:11s & \num{381.3}MB \\
D.\ melanogaster  & \num{10522} & \num{118698} & \num{6043033} & 09h:03m:58s & \num{782.0}MB    & 00h:44m:58s & \num{766.7}MB \\
A.\ thaliana      & \num{20239} & \num{331525} & \num{8903236} & 35h:15m:09s & \num{1200.0}MB   & 01h:31m:07s & \num{1183.4}MB \\
H.\ sapiens I     & \num{16201} & \num{236930} & \num{9605435} & 49h:54m:31s & \num{1294.5}MB & 01h:55m:56s & \num{1277.3}MB \\
\midrule
S.\ cerevisiae II & \num{1620}  & \num{9060}   & \num{862281}  & 00h:32m:17s & \num{125.7}MB  & 00h:03m:11s & \num{120.5}MB \\
H.\ sapiens II    & \num{10380} & \num{69515}  & \num{5607838} & 11h:20m:41s & \num{761.6}MB  & 00h:27m:42s & \num{746.9}MB \\
\bottomrule
\end{tabular}
\end{center}
\end{table*}

\end{document}

%% file: curve_legend.tex
\begingroup
  \makeatletter
  \providecommand\color[2][]{%
    \GenericError{(gnuplot) \space\space\space\@spaces}{%
      Package color not loaded in conjunction with
      terminal option `colourtext'%
    }{See the gnuplot documentation for explanation.%
    }{Either use 'blacktext' in gnuplot or load the package
      color.sty in LaTeX.}%
    \renewcommand\color[2][]{}%
  }%
  \providecommand\includegraphics[2][]{%
    \GenericError{(gnuplot) \space\space\space\@spaces}{%
      Package graphicx or graphics not loaded%
    }{See the gnuplot documentation for explanation.%
    }{The gnuplot epslatex terminal needs graphicx.sty or graphics.sty.}%
    \renewcommand\includegraphics[2][]{}%
  }%
  \providecommand\rotatebox[2]{#2}%
  \@ifundefined{ifGPcolor}{%
    \newif\ifGPcolor
    \GPcolortrue
  }{}%
  \@ifundefined{ifGPblacktext}{%
    \newif\ifGPblacktext
    \GPblacktexttrue
  }{}%
  \let\gplgaddtomacro\g@addto@macro
  \gdef\gplbacktext{}%
  \gdef\gplfronttext{}%
  \makeatother
  \ifGPblacktext
    \def\colorrgb#1{}%
    \def\colorgray#1{}%
  \else
    \ifGPcolor
      \def\colorrgb#1{\color[rgb]{#1}}%
      \def\colorgray#1{\color[gray]{#1}}%
      \expandafter\def\csname LTw\endcsname{\color{white}}%
      \expandafter\def\csname LTb\endcsname{\color{black}}%
      \expandafter\def\csname LTa\endcsname{\color{black}}%
      \expandafter\def\csname LT0\endcsname{\color[rgb]{1,0,0}}%
      \expandafter\def\csname LT1\endcsname{\color[rgb]{0,1,0}}%
      \expandafter\def\csname LT2\endcsname{\color[rgb]{0,0,1}}%
      \expandafter\def\csname LT3\endcsname{\color[rgb]{1,0,1}}%
      \expandafter\def\csname LT4\endcsname{\color[rgb]{0,1,1}}%
      \expandafter\def\csname LT5\endcsname{\color[rgb]{1,1,0}}%
      \expandafter\def\csname LT6\endcsname{\color[rgb]{0,0,0}}%
      \expandafter\def\csname LT7\endcsname{\color[rgb]{1,0.3,0}}%
      \expandafter\def\csname LT8\endcsname{\color[rgb]{0.5,0.5,0.5}}%
    \else
      \def\colorrgb#1{\color{black}}%
      \def\colorgray#1{\color[gray]{#1}}%
      \expandafter\def\csname LTw\endcsname{\color{white}}%
      \expandafter\def\csname LTb\endcsname{\color{black}}%
      \expandafter\def\csname LTa\endcsname{\color{black}}%
      \expandafter\def\csname LT0\endcsname{\color{black}}%
      \expandafter\def\csname LT1\endcsname{\color{black}}%
      \expandafter\def\csname LT2\endcsname{\color{black}}%
      \expandafter\def\csname LT3\endcsname{\color{black}}%
      \expandafter\def\csname LT4\endcsname{\color{black}}%
      \expandafter\def\csname LT5\endcsname{\color{black}}%
      \expandafter\def\csname LT6\endcsname{\color{black}}%
      \expandafter\def\csname LT7\endcsname{\color{black}}%
      \expandafter\def\csname LT8\endcsname{\color{black}}%
    \fi
  \fi
    \setlength{\unitlength}{0.0500bp}%
    \ifx\gptboxheight\undefined%
      \newlength{\gptboxheight}%
      \newlength{\gptboxwidth}%
      \newsavebox{\gptboxtext}%
    \fi%
    \setlength{\fboxrule}{0.5pt}%
    \setlength{\fboxsep}{1pt}%
    \definecolor{tbcol}{rgb}{1,1,1}%
\begin{picture}(3960.00,560.00)%
    \gplgaddtomacro\gplbacktext{%
    }%
    \gplgaddtomacro\gplfronttext{%
      \csname LTb\endcsname
      \put(1541,401){\makebox(0,0)[r]{\strut{}M-SLIDER}}%
      \csname LTb\endcsname
      \put(1541,107){\makebox(0,0)[r]{\strut{}SEQ-SLIDER}}%
      \csname LTb\endcsname
      \put(2855,401){\makebox(0,0)[r]{\strut{}Exact}}%
    }%
    \gplbacktext
    \put(0,0){\includegraphics[width={198.00bp},height={28.00bp}]{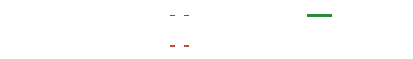}}%
    \gplfronttext
  \end{picture}%
\endgroup

%% file: curve.ecoli.l8d5.dedup.tex
\begingroup
  \makeatletter
  \providecommand\color[2][]{%
    \GenericError{(gnuplot) \space\space\space\@spaces}{%
      Package color not loaded in conjunction with
      terminal option `colourtext'%
    }{See the gnuplot documentation for explanation.%
    }{Either use 'blacktext' in gnuplot or load the package
      color.sty in LaTeX.}%
    \renewcommand\color[2][]{}%
  }%
  \providecommand\includegraphics[2][]{%
    \GenericError{(gnuplot) \space\space\space\@spaces}{%
      Package graphicx or graphics not loaded%
    }{See the gnuplot documentation for explanation.%
    }{The gnuplot epslatex terminal needs graphicx.sty or graphics.sty.}%
    \renewcommand\includegraphics[2][]{}%
  }%
  \providecommand\rotatebox[2]{#2}%
  \@ifundefined{ifGPcolor}{%
    \newif\ifGPcolor
    \GPcolortrue
  }{}%
  \@ifundefined{ifGPblacktext}{%
    \newif\ifGPblacktext
    \GPblacktexttrue
  }{}%
  \let\gplgaddtomacro\g@addto@macro
  \gdef\gplbacktext{}%
  \gdef\gplfronttext{}%
  \makeatother
  \ifGPblacktext
    \def\colorrgb#1{}%
    \def\colorgray#1{}%
  \else
    \ifGPcolor
      \def\colorrgb#1{\color[rgb]{#1}}%
      \def\colorgray#1{\color[gray]{#1}}%
      \expandafter\def\csname LTw\endcsname{\color{white}}%
      \expandafter\def\csname LTb\endcsname{\color{black}}%
      \expandafter\def\csname LTa\endcsname{\color{black}}%
      \expandafter\def\csname LT0\endcsname{\color[rgb]{1,0,0}}%
      \expandafter\def\csname LT1\endcsname{\color[rgb]{0,1,0}}%
      \expandafter\def\csname LT2\endcsname{\color[rgb]{0,0,1}}%
      \expandafter\def\csname LT3\endcsname{\color[rgb]{1,0,1}}%
      \expandafter\def\csname LT4\endcsname{\color[rgb]{0,1,1}}%
      \expandafter\def\csname LT5\endcsname{\color[rgb]{1,1,0}}%
      \expandafter\def\csname LT6\endcsname{\color[rgb]{0,0,0}}%
      \expandafter\def\csname LT7\endcsname{\color[rgb]{1,0.3,0}}%
      \expandafter\def\csname LT8\endcsname{\color[rgb]{0.5,0.5,0.5}}%
    \else
      \def\colorrgb#1{\color{black}}%
      \def\colorgray#1{\color[gray]{#1}}%
      \expandafter\def\csname LTw\endcsname{\color{white}}%
      \expandafter\def\csname LTb\endcsname{\color{black}}%
      \expandafter\def\csname LTa\endcsname{\color{black}}%
      \expandafter\def\csname LT0\endcsname{\color{black}}%
      \expandafter\def\csname LT1\endcsname{\color{black}}%
      \expandafter\def\csname LT2\endcsname{\color{black}}%
      \expandafter\def\csname LT3\endcsname{\color{black}}%
      \expandafter\def\csname LT4\endcsname{\color{black}}%
      \expandafter\def\csname LT5\endcsname{\color{black}}%
      \expandafter\def\csname LT6\endcsname{\color{black}}%
      \expandafter\def\csname LT7\endcsname{\color{black}}%
      \expandafter\def\csname LT8\endcsname{\color{black}}%
    \fi
  \fi
    \setlength{\unitlength}{0.0500bp}%
    \ifx\gptboxheight\undefined%
      \newlength{\gptboxheight}%
      \newlength{\gptboxwidth}%
      \newsavebox{\gptboxtext}%
    \fi%
    \setlength{\fboxrule}{0.5pt}%
    \setlength{\fboxsep}{1pt}%
    \definecolor{tbcol}{rgb}{1,1,1}%
\begin{picture}(1980.00,1700.00)%
    \gplgaddtomacro\gplbacktext{%
      \csname LTb\endcsname
      \put(566,627){\makebox(0,0)[r]{\strut{}$0$}}%
      \csname LTb\endcsname
      \put(566,1043){\makebox(0,0)[r]{\strut{}$0.5$}}%
      \csname LTb\endcsname
      \put(566,1458){\makebox(0,0)[r]{\strut{}$1$}}%
      \csname LTb\endcsname
      \put(649,431){\makebox(0,0){\strut{}$1$}}%
      \csname LTb\endcsname
      \put(1712,431){\makebox(0,0){\strut{}$2301$}}%
    }%
    \gplgaddtomacro\gplfronttext{%
      \csname LTb\endcsname
      \put(159,1055){\rotatebox{-270.00}{\makebox(0,0){\strut{}\normalsize Norm. $f_{\chi^2}$}}}%
      \csname LTb\endcsname
      \put(1181,137){\makebox(0,0){\strut{}\normalsize Score rank}}%
    }%
    \gplbacktext
    \put(0,0){\includegraphics[width={99.00bp},height={85.00bp}]{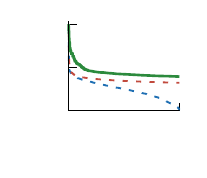}}%
    \gplfronttext
  \end{picture}%
\endgroup

%% file: curve.yeast.l8d5.dedup.tex
\begingroup
  \makeatletter
  \providecommand\color[2][]{%
    \GenericError{(gnuplot) \space\space\space\@spaces}{%
      Package color not loaded in conjunction with
      terminal option `colourtext'%
    }{See the gnuplot documentation for explanation.%
    }{Either use 'blacktext' in gnuplot or load the package
      color.sty in LaTeX.}%
    \renewcommand\color[2][]{}%
  }%
  \providecommand\includegraphics[2][]{%
    \GenericError{(gnuplot) \space\space\space\@spaces}{%
      Package graphicx or graphics not loaded%
    }{See the gnuplot documentation for explanation.%
    }{The gnuplot epslatex terminal needs graphicx.sty or graphics.sty.}%
    \renewcommand\includegraphics[2][]{}%
  }%
  \providecommand\rotatebox[2]{#2}%
  \@ifundefined{ifGPcolor}{%
    \newif\ifGPcolor
    \GPcolortrue
  }{}%
  \@ifundefined{ifGPblacktext}{%
    \newif\ifGPblacktext
    \GPblacktexttrue
  }{}%
  \let\gplgaddtomacro\g@addto@macro
  \gdef\gplbacktext{}%
  \gdef\gplfronttext{}%
  \makeatother
  \ifGPblacktext
    \def\colorrgb#1{}%
    \def\colorgray#1{}%
  \else
    \ifGPcolor
      \def\colorrgb#1{\color[rgb]{#1}}%
      \def\colorgray#1{\color[gray]{#1}}%
      \expandafter\def\csname LTw\endcsname{\color{white}}%
      \expandafter\def\csname LTb\endcsname{\color{black}}%
      \expandafter\def\csname LTa\endcsname{\color{black}}%
      \expandafter\def\csname LT0\endcsname{\color[rgb]{1,0,0}}%
      \expandafter\def\csname LT1\endcsname{\color[rgb]{0,1,0}}%
      \expandafter\def\csname LT2\endcsname{\color[rgb]{0,0,1}}%
      \expandafter\def\csname LT3\endcsname{\color[rgb]{1,0,1}}%
      \expandafter\def\csname LT4\endcsname{\color[rgb]{0,1,1}}%
      \expandafter\def\csname LT5\endcsname{\color[rgb]{1,1,0}}%
      \expandafter\def\csname LT6\endcsname{\color[rgb]{0,0,0}}%
      \expandafter\def\csname LT7\endcsname{\color[rgb]{1,0.3,0}}%
      \expandafter\def\csname LT8\endcsname{\color[rgb]{0.5,0.5,0.5}}%
    \else
      \def\colorrgb#1{\color{black}}%
      \def\colorgray#1{\color[gray]{#1}}%
      \expandafter\def\csname LTw\endcsname{\color{white}}%
      \expandafter\def\csname LTb\endcsname{\color{black}}%
      \expandafter\def\csname LTa\endcsname{\color{black}}%
      \expandafter\def\csname LT0\endcsname{\color{black}}%
      \expandafter\def\csname LT1\endcsname{\color{black}}%
      \expandafter\def\csname LT2\endcsname{\color{black}}%
      \expandafter\def\csname LT3\endcsname{\color{black}}%
      \expandafter\def\csname LT4\endcsname{\color{black}}%
      \expandafter\def\csname LT5\endcsname{\color{black}}%
      \expandafter\def\csname LT6\endcsname{\color{black}}%
      \expandafter\def\csname LT7\endcsname{\color{black}}%
      \expandafter\def\csname LT8\endcsname{\color{black}}%
    \fi
  \fi
    \setlength{\unitlength}{0.0500bp}%
    \ifx\gptboxheight\undefined%
      \newlength{\gptboxheight}%
      \newlength{\gptboxwidth}%
      \newsavebox{\gptboxtext}%
    \fi%
    \setlength{\fboxrule}{0.5pt}%
    \setlength{\fboxsep}{1pt}%
    \definecolor{tbcol}{rgb}{1,1,1}%
\begin{picture}(1980.00,1700.00)%
    \gplgaddtomacro\gplbacktext{%
      \csname LTb\endcsname
      \put(566,627){\makebox(0,0)[r]{\strut{}$0$}}%
      \csname LTb\endcsname
      \put(566,1043){\makebox(0,0)[r]{\strut{}$0.5$}}%
      \csname LTb\endcsname
      \put(566,1458){\makebox(0,0)[r]{\strut{}$1$}}%
      \csname LTb\endcsname
      \put(649,431){\makebox(0,0){\strut{}$1$}}%
      \csname LTb\endcsname
      \put(1712,431){\makebox(0,0){\strut{}$363$}}%
    }%
    \gplgaddtomacro\gplfronttext{%
      \csname LTb\endcsname
      \put(159,1055){\rotatebox{-270.00}{\makebox(0,0){\strut{}\normalsize Norm. $f_{\chi^2}$}}}%
      \csname LTb\endcsname
      \put(1181,137){\makebox(0,0){\strut{}\normalsize Score rank}}%
    }%
    \gplbacktext
    \put(0,0){\includegraphics[width={99.00bp},height={85.00bp}]{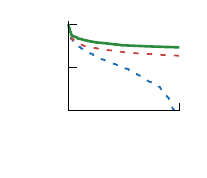}}%
    \gplfronttext
  \end{picture}%
\endgroup

%% file: curve.fly.l8d5.dedup.tex
\begingroup
  \makeatletter
  \providecommand\color[2][]{%
    \GenericError{(gnuplot) \space\space\space\@spaces}{%
      Package color not loaded in conjunction with
      terminal option `colourtext'%
    }{See the gnuplot documentation for explanation.%
    }{Either use 'blacktext' in gnuplot or load the package
      color.sty in LaTeX.}%
    \renewcommand\color[2][]{}%
  }%
  \providecommand\includegraphics[2][]{%
    \GenericError{(gnuplot) \space\space\space\@spaces}{%
      Package graphicx or graphics not loaded%
    }{See the gnuplot documentation for explanation.%
    }{The gnuplot epslatex terminal needs graphicx.sty or graphics.sty.}%
    \renewcommand\includegraphics[2][]{}%
  }%
  \providecommand\rotatebox[2]{#2}%
  \@ifundefined{ifGPcolor}{%
    \newif\ifGPcolor
    \GPcolortrue
  }{}%
  \@ifundefined{ifGPblacktext}{%
    \newif\ifGPblacktext
    \GPblacktexttrue
  }{}%
  \let\gplgaddtomacro\g@addto@macro
  \gdef\gplbacktext{}%
  \gdef\gplfronttext{}%
  \makeatother
  \ifGPblacktext
    \def\colorrgb#1{}%
    \def\colorgray#1{}%
  \else
    \ifGPcolor
      \def\colorrgb#1{\color[rgb]{#1}}%
      \def\colorgray#1{\color[gray]{#1}}%
      \expandafter\def\csname LTw\endcsname{\color{white}}%
      \expandafter\def\csname LTb\endcsname{\color{black}}%
      \expandafter\def\csname LTa\endcsname{\color{black}}%
      \expandafter\def\csname LT0\endcsname{\color[rgb]{1,0,0}}%
      \expandafter\def\csname LT1\endcsname{\color[rgb]{0,1,0}}%
      \expandafter\def\csname LT2\endcsname{\color[rgb]{0,0,1}}%
      \expandafter\def\csname LT3\endcsname{\color[rgb]{1,0,1}}%
      \expandafter\def\csname LT4\endcsname{\color[rgb]{0,1,1}}%
      \expandafter\def\csname LT5\endcsname{\color[rgb]{1,1,0}}%
      \expandafter\def\csname LT6\endcsname{\color[rgb]{0,0,0}}%
      \expandafter\def\csname LT7\endcsname{\color[rgb]{1,0.3,0}}%
      \expandafter\def\csname LT8\endcsname{\color[rgb]{0.5,0.5,0.5}}%
    \else
      \def\colorrgb#1{\color{black}}%
      \def\colorgray#1{\color[gray]{#1}}%
      \expandafter\def\csname LTw\endcsname{\color{white}}%
      \expandafter\def\csname LTb\endcsname{\color{black}}%
      \expandafter\def\csname LTa\endcsname{\color{black}}%
      \expandafter\def\csname LT0\endcsname{\color{black}}%
      \expandafter\def\csname LT1\endcsname{\color{black}}%
      \expandafter\def\csname LT2\endcsname{\color{black}}%
      \expandafter\def\csname LT3\endcsname{\color{black}}%
      \expandafter\def\csname LT4\endcsname{\color{black}}%
      \expandafter\def\csname LT5\endcsname{\color{black}}%
      \expandafter\def\csname LT6\endcsname{\color{black}}%
      \expandafter\def\csname LT7\endcsname{\color{black}}%
      \expandafter\def\csname LT8\endcsname{\color{black}}%
    \fi
  \fi
    \setlength{\unitlength}{0.0500bp}%
    \ifx\gptboxheight\undefined%
      \newlength{\gptboxheight}%
      \newlength{\gptboxwidth}%
      \newsavebox{\gptboxtext}%
    \fi%
    \setlength{\fboxrule}{0.5pt}%
    \setlength{\fboxsep}{1pt}%
    \definecolor{tbcol}{rgb}{1,1,1}%
\begin{picture}(1980.00,1700.00)%
    \gplgaddtomacro\gplbacktext{%
      \csname LTb\endcsname
      \put(566,627){\makebox(0,0)[r]{\strut{}$0$}}%
      \csname LTb\endcsname
      \put(566,1043){\makebox(0,0)[r]{\strut{}$0.5$}}%
      \csname LTb\endcsname
      \put(566,1458){\makebox(0,0)[r]{\strut{}$1$}}%
      \csname LTb\endcsname
      \put(649,431){\makebox(0,0){\strut{}$1$}}%
      \csname LTb\endcsname
      \put(1712,431){\makebox(0,0){\strut{}$149$}}%
    }%
    \gplgaddtomacro\gplfronttext{%
      \csname LTb\endcsname
      \put(159,1055){\rotatebox{-270.00}{\makebox(0,0){\strut{}\normalsize Norm. $f_{\chi^2}$}}}%
      \csname LTb\endcsname
      \put(1181,137){\makebox(0,0){\strut{}\normalsize Score rank}}%
    }%
    \gplbacktext
    \put(0,0){\includegraphics[width={99.00bp},height={85.00bp}]{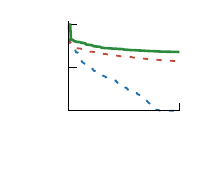}}%
    \gplfronttext
  \end{picture}%
\endgroup

%% file: curve.plant.l8d5.dedup.tex
\begingroup
  \makeatletter
  \providecommand\color[2][]{%
    \GenericError{(gnuplot) \space\space\space\@spaces}{%
      Package color not loaded in conjunction with
      terminal option `colourtext'%
    }{See the gnuplot documentation for explanation.%
    }{Either use 'blacktext' in gnuplot or load the package
      color.sty in LaTeX.}%
    \renewcommand\color[2][]{}%
  }%
  \providecommand\includegraphics[2][]{%
    \GenericError{(gnuplot) \space\space\space\@spaces}{%
      Package graphicx or graphics not loaded%
    }{See the gnuplot documentation for explanation.%
    }{The gnuplot epslatex terminal needs graphicx.sty or graphics.sty.}%
    \renewcommand\includegraphics[2][]{}%
  }%
  \providecommand\rotatebox[2]{#2}%
  \@ifundefined{ifGPcolor}{%
    \newif\ifGPcolor
    \GPcolortrue
  }{}%
  \@ifundefined{ifGPblacktext}{%
    \newif\ifGPblacktext
    \GPblacktexttrue
  }{}%
  \let\gplgaddtomacro\g@addto@macro
  \gdef\gplbacktext{}%
  \gdef\gplfronttext{}%
  \makeatother
  \ifGPblacktext
    \def\colorrgb#1{}%
    \def\colorgray#1{}%
  \else
    \ifGPcolor
      \def\colorrgb#1{\color[rgb]{#1}}%
      \def\colorgray#1{\color[gray]{#1}}%
      \expandafter\def\csname LTw\endcsname{\color{white}}%
      \expandafter\def\csname LTb\endcsname{\color{black}}%
      \expandafter\def\csname LTa\endcsname{\color{black}}%
      \expandafter\def\csname LT0\endcsname{\color[rgb]{1,0,0}}%
      \expandafter\def\csname LT1\endcsname{\color[rgb]{0,1,0}}%
      \expandafter\def\csname LT2\endcsname{\color[rgb]{0,0,1}}%
      \expandafter\def\csname LT3\endcsname{\color[rgb]{1,0,1}}%
      \expandafter\def\csname LT4\endcsname{\color[rgb]{0,1,1}}%
      \expandafter\def\csname LT5\endcsname{\color[rgb]{1,1,0}}%
      \expandafter\def\csname LT6\endcsname{\color[rgb]{0,0,0}}%
      \expandafter\def\csname LT7\endcsname{\color[rgb]{1,0.3,0}}%
      \expandafter\def\csname LT8\endcsname{\color[rgb]{0.5,0.5,0.5}}%
    \else
      \def\colorrgb#1{\color{black}}%
      \def\colorgray#1{\color[gray]{#1}}%
      \expandafter\def\csname LTw\endcsname{\color{white}}%
      \expandafter\def\csname LTb\endcsname{\color{black}}%
      \expandafter\def\csname LTa\endcsname{\color{black}}%
      \expandafter\def\csname LT0\endcsname{\color{black}}%
      \expandafter\def\csname LT1\endcsname{\color{black}}%
      \expandafter\def\csname LT2\endcsname{\color{black}}%
      \expandafter\def\csname LT3\endcsname{\color{black}}%
      \expandafter\def\csname LT4\endcsname{\color{black}}%
      \expandafter\def\csname LT5\endcsname{\color{black}}%
      \expandafter\def\csname LT6\endcsname{\color{black}}%
      \expandafter\def\csname LT7\endcsname{\color{black}}%
      \expandafter\def\csname LT8\endcsname{\color{black}}%
    \fi
  \fi
    \setlength{\unitlength}{0.0500bp}%
    \ifx\gptboxheight\undefined%
      \newlength{\gptboxheight}%
      \newlength{\gptboxwidth}%
      \newsavebox{\gptboxtext}%
    \fi%
    \setlength{\fboxrule}{0.5pt}%
    \setlength{\fboxsep}{1pt}%
    \definecolor{tbcol}{rgb}{1,1,1}%
\begin{picture}(1980.00,1700.00)%
    \gplgaddtomacro\gplbacktext{%
      \csname LTb\endcsname
      \put(566,627){\makebox(0,0)[r]{\strut{}$0$}}%
      \csname LTb\endcsname
      \put(566,1043){\makebox(0,0)[r]{\strut{}$0.5$}}%
      \csname LTb\endcsname
      \put(566,1458){\makebox(0,0)[r]{\strut{}$1$}}%
      \csname LTb\endcsname
      \put(649,431){\makebox(0,0){\strut{}$1$}}%
      \csname LTb\endcsname
      \put(1712,431){\makebox(0,0){\strut{}$168$}}%
    }%
    \gplgaddtomacro\gplfronttext{%
      \csname LTb\endcsname
      \put(159,1055){\rotatebox{-270.00}{\makebox(0,0){\strut{}\normalsize Norm. $f_{\chi^2}$}}}%
      \csname LTb\endcsname
      \put(1181,137){\makebox(0,0){\strut{}\normalsize Score rank}}%
    }%
    \gplbacktext
    \put(0,0){\includegraphics[width={99.00bp},height={85.00bp}]{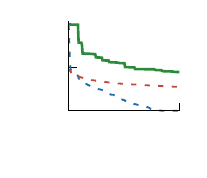}}%
    \gplfronttext
  \end{picture}%
\endgroup

%% file: curve.human.l8d5.dedup.tex
\begingroup
  \makeatletter
  \providecommand\color[2][]{%
    \GenericError{(gnuplot) \space\space\space\@spaces}{%
      Package color not loaded in conjunction with
      terminal option `colourtext'%
    }{See the gnuplot documentation for explanation.%
    }{Either use 'blacktext' in gnuplot or load the package
      color.sty in LaTeX.}%
    \renewcommand\color[2][]{}%
  }%
  \providecommand\includegraphics[2][]{%
    \GenericError{(gnuplot) \space\space\space\@spaces}{%
      Package graphicx or graphics not loaded%
    }{See the gnuplot documentation for explanation.%
    }{The gnuplot epslatex terminal needs graphicx.sty or graphics.sty.}%
    \renewcommand\includegraphics[2][]{}%
  }%
  \providecommand\rotatebox[2]{#2}%
  \@ifundefined{ifGPcolor}{%
    \newif\ifGPcolor
    \GPcolortrue
  }{}%
  \@ifundefined{ifGPblacktext}{%
    \newif\ifGPblacktext
    \GPblacktexttrue
  }{}%
  \let\gplgaddtomacro\g@addto@macro
  \gdef\gplbacktext{}%
  \gdef\gplfronttext{}%
  \makeatother
  \ifGPblacktext
    \def\colorrgb#1{}%
    \def\colorgray#1{}%
  \else
    \ifGPcolor
      \def\colorrgb#1{\color[rgb]{#1}}%
      \def\colorgray#1{\color[gray]{#1}}%
      \expandafter\def\csname LTw\endcsname{\color{white}}%
      \expandafter\def\csname LTb\endcsname{\color{black}}%
      \expandafter\def\csname LTa\endcsname{\color{black}}%
      \expandafter\def\csname LT0\endcsname{\color[rgb]{1,0,0}}%
      \expandafter\def\csname LT1\endcsname{\color[rgb]{0,1,0}}%
      \expandafter\def\csname LT2\endcsname{\color[rgb]{0,0,1}}%
      \expandafter\def\csname LT3\endcsname{\color[rgb]{1,0,1}}%
      \expandafter\def\csname LT4\endcsname{\color[rgb]{0,1,1}}%
      \expandafter\def\csname LT5\endcsname{\color[rgb]{1,1,0}}%
      \expandafter\def\csname LT6\endcsname{\color[rgb]{0,0,0}}%
      \expandafter\def\csname LT7\endcsname{\color[rgb]{1,0.3,0}}%
      \expandafter\def\csname LT8\endcsname{\color[rgb]{0.5,0.5,0.5}}%
    \else
      \def\colorrgb#1{\color{black}}%
      \def\colorgray#1{\color[gray]{#1}}%
      \expandafter\def\csname LTw\endcsname{\color{white}}%
      \expandafter\def\csname LTb\endcsname{\color{black}}%
      \expandafter\def\csname LTa\endcsname{\color{black}}%
      \expandafter\def\csname LT0\endcsname{\color{black}}%
      \expandafter\def\csname LT1\endcsname{\color{black}}%
      \expandafter\def\csname LT2\endcsname{\color{black}}%
      \expandafter\def\csname LT3\endcsname{\color{black}}%
      \expandafter\def\csname LT4\endcsname{\color{black}}%
      \expandafter\def\csname LT5\endcsname{\color{black}}%
      \expandafter\def\csname LT6\endcsname{\color{black}}%
      \expandafter\def\csname LT7\endcsname{\color{black}}%
      \expandafter\def\csname LT8\endcsname{\color{black}}%
    \fi
  \fi
    \setlength{\unitlength}{0.0500bp}%
    \ifx\gptboxheight\undefined%
      \newlength{\gptboxheight}%
      \newlength{\gptboxwidth}%
      \newsavebox{\gptboxtext}%
    \fi%
    \setlength{\fboxrule}{0.5pt}%
    \setlength{\fboxsep}{1pt}%
    \definecolor{tbcol}{rgb}{1,1,1}%
\begin{picture}(1980.00,1700.00)%
    \gplgaddtomacro\gplbacktext{%
      \csname LTb\endcsname
      \put(566,627){\makebox(0,0)[r]{\strut{}$0$}}%
      \csname LTb\endcsname
      \put(566,1043){\makebox(0,0)[r]{\strut{}$0.5$}}%
      \csname LTb\endcsname
      \put(566,1458){\makebox(0,0)[r]{\strut{}$1$}}%
      \csname LTb\endcsname
      \put(649,431){\makebox(0,0){\strut{}$1$}}%
      \csname LTb\endcsname
      \put(1712,431){\makebox(0,0){\strut{}$187$}}%
    }%
    \gplgaddtomacro\gplfronttext{%
      \csname LTb\endcsname
      \put(159,1055){\rotatebox{-270.00}{\makebox(0,0){\strut{}\normalsize Norm. $f_{\chi^2}$}}}%
      \csname LTb\endcsname
      \put(1181,137){\makebox(0,0){\strut{}\normalsize Score rank}}%
    }%
    \gplbacktext
    \put(0,0){\includegraphics[width={99.00bp},height={85.00bp}]{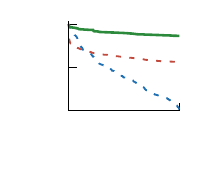}}%
    \gplfronttext
  \end{picture}%
\endgroup

%% file: curve.yeast_aj.l8d5.dedup.tex
\begingroup
  \makeatletter
  \providecommand\color[2][]{%
    \GenericError{(gnuplot) \space\space\space\@spaces}{%
      Package color not loaded in conjunction with
      terminal option `colourtext'%
    }{See the gnuplot documentation for explanation.%
    }{Either use 'blacktext' in gnuplot or load the package
      color.sty in LaTeX.}%
    \renewcommand\color[2][]{}%
  }%
  \providecommand\includegraphics[2][]{%
    \GenericError{(gnuplot) \space\space\space\@spaces}{%
      Package graphicx or graphics not loaded%
    }{See the gnuplot documentation for explanation.%
    }{The gnuplot epslatex terminal needs graphicx.sty or graphics.sty.}%
    \renewcommand\includegraphics[2][]{}%
  }%
  \providecommand\rotatebox[2]{#2}%
  \@ifundefined{ifGPcolor}{%
    \newif\ifGPcolor
    \GPcolortrue
  }{}%
  \@ifundefined{ifGPblacktext}{%
    \newif\ifGPblacktext
    \GPblacktexttrue
  }{}%
  \let\gplgaddtomacro\g@addto@macro
  \gdef\gplbacktext{}%
  \gdef\gplfronttext{}%
  \makeatother
  \ifGPblacktext
    \def\colorrgb#1{}%
    \def\colorgray#1{}%
  \else
    \ifGPcolor
      \def\colorrgb#1{\color[rgb]{#1}}%
      \def\colorgray#1{\color[gray]{#1}}%
      \expandafter\def\csname LTw\endcsname{\color{white}}%
      \expandafter\def\csname LTb\endcsname{\color{black}}%
      \expandafter\def\csname LTa\endcsname{\color{black}}%
      \expandafter\def\csname LT0\endcsname{\color[rgb]{1,0,0}}%
      \expandafter\def\csname LT1\endcsname{\color[rgb]{0,1,0}}%
      \expandafter\def\csname LT2\endcsname{\color[rgb]{0,0,1}}%
      \expandafter\def\csname LT3\endcsname{\color[rgb]{1,0,1}}%
      \expandafter\def\csname LT4\endcsname{\color[rgb]{0,1,1}}%
      \expandafter\def\csname LT5\endcsname{\color[rgb]{1,1,0}}%
      \expandafter\def\csname LT6\endcsname{\color[rgb]{0,0,0}}%
      \expandafter\def\csname LT7\endcsname{\color[rgb]{1,0.3,0}}%
      \expandafter\def\csname LT8\endcsname{\color[rgb]{0.5,0.5,0.5}}%
    \else
      \def\colorrgb#1{\color{black}}%
      \def\colorgray#1{\color[gray]{#1}}%
      \expandafter\def\csname LTw\endcsname{\color{white}}%
      \expandafter\def\csname LTb\endcsname{\color{black}}%
      \expandafter\def\csname LTa\endcsname{\color{black}}%
      \expandafter\def\csname LT0\endcsname{\color{black}}%
      \expandafter\def\csname LT1\endcsname{\color{black}}%
      \expandafter\def\csname LT2\endcsname{\color{black}}%
      \expandafter\def\csname LT3\endcsname{\color{black}}%
      \expandafter\def\csname LT4\endcsname{\color{black}}%
      \expandafter\def\csname LT5\endcsname{\color{black}}%
      \expandafter\def\csname LT6\endcsname{\color{black}}%
      \expandafter\def\csname LT7\endcsname{\color{black}}%
      \expandafter\def\csname LT8\endcsname{\color{black}}%
    \fi
  \fi
    \setlength{\unitlength}{0.0500bp}%
    \ifx\gptboxheight\undefined%
      \newlength{\gptboxheight}%
      \newlength{\gptboxwidth}%
      \newsavebox{\gptboxtext}%
    \fi%
    \setlength{\fboxrule}{0.5pt}%
    \setlength{\fboxsep}{1pt}%
    \definecolor{tbcol}{rgb}{1,1,1}%
\begin{picture}(1980.00,1700.00)%
    \gplgaddtomacro\gplbacktext{%
      \csname LTb\endcsname
      \put(566,627){\makebox(0,0)[r]{\strut{}$0$}}%
      \csname LTb\endcsname
      \put(566,1043){\makebox(0,0)[r]{\strut{}$0.5$}}%
      \csname LTb\endcsname
      \put(566,1458){\makebox(0,0)[r]{\strut{}$1$}}%
      \csname LTb\endcsname
      \put(649,431){\makebox(0,0){\strut{}$1$}}%
      \csname LTb\endcsname
      \put(1712,431){\makebox(0,0){\strut{}$9988$}}%
    }%
    \gplgaddtomacro\gplfronttext{%
      \csname LTb\endcsname
      \put(159,1055){\rotatebox{-270.00}{\makebox(0,0){\strut{}\normalsize Norm. $f_{\chi^2}$}}}%
      \csname LTb\endcsname
      \put(1181,137){\makebox(0,0){\strut{}\normalsize Score rank}}%
    }%
    \gplbacktext
    \put(0,0){\includegraphics[width={99.00bp},height={85.00bp}]{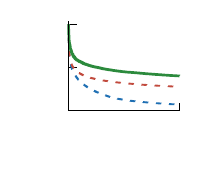}}%
    \gplfronttext
  \end{picture}%
\endgroup

%% file: curve.human_aj.l8d5.dedup.tex
\begingroup
  \makeatletter
  \providecommand\color[2][]{%
    \GenericError{(gnuplot) \space\space\space\@spaces}{%
      Package color not loaded in conjunction with
      terminal option `colourtext'%
    }{See the gnuplot documentation for explanation.%
    }{Either use 'blacktext' in gnuplot or load the package
      color.sty in LaTeX.}%
    \renewcommand\color[2][]{}%
  }%
  \providecommand\includegraphics[2][]{%
    \GenericError{(gnuplot) \space\space\space\@spaces}{%
      Package graphicx or graphics not loaded%
    }{See the gnuplot documentation for explanation.%
    }{The gnuplot epslatex terminal needs graphicx.sty or graphics.sty.}%
    \renewcommand\includegraphics[2][]{}%
  }%
  \providecommand\rotatebox[2]{#2}%
  \@ifundefined{ifGPcolor}{%
    \newif\ifGPcolor
    \GPcolortrue
  }{}%
  \@ifundefined{ifGPblacktext}{%
    \newif\ifGPblacktext
    \GPblacktexttrue
  }{}%
  \let\gplgaddtomacro\g@addto@macro
  \gdef\gplbacktext{}%
  \gdef\gplfronttext{}%
  \makeatother
  \ifGPblacktext
    \def\colorrgb#1{}%
    \def\colorgray#1{}%
  \else
    \ifGPcolor
      \def\colorrgb#1{\color[rgb]{#1}}%
      \def\colorgray#1{\color[gray]{#1}}%
      \expandafter\def\csname LTw\endcsname{\color{white}}%
      \expandafter\def\csname LTb\endcsname{\color{black}}%
      \expandafter\def\csname LTa\endcsname{\color{black}}%
      \expandafter\def\csname LT0\endcsname{\color[rgb]{1,0,0}}%
      \expandafter\def\csname LT1\endcsname{\color[rgb]{0,1,0}}%
      \expandafter\def\csname LT2\endcsname{\color[rgb]{0,0,1}}%
      \expandafter\def\csname LT3\endcsname{\color[rgb]{1,0,1}}%
      \expandafter\def\csname LT4\endcsname{\color[rgb]{0,1,1}}%
      \expandafter\def\csname LT5\endcsname{\color[rgb]{1,1,0}}%
      \expandafter\def\csname LT6\endcsname{\color[rgb]{0,0,0}}%
      \expandafter\def\csname LT7\endcsname{\color[rgb]{1,0.3,0}}%
      \expandafter\def\csname LT8\endcsname{\color[rgb]{0.5,0.5,0.5}}%
    \else
      \def\colorrgb#1{\color{black}}%
      \def\colorgray#1{\color[gray]{#1}}%
      \expandafter\def\csname LTw\endcsname{\color{white}}%
      \expandafter\def\csname LTb\endcsname{\color{black}}%
      \expandafter\def\csname LTa\endcsname{\color{black}}%
      \expandafter\def\csname LT0\endcsname{\color{black}}%
      \expandafter\def\csname LT1\endcsname{\color{black}}%
      \expandafter\def\csname LT2\endcsname{\color{black}}%
      \expandafter\def\csname LT3\endcsname{\color{black}}%
      \expandafter\def\csname LT4\endcsname{\color{black}}%
      \expandafter\def\csname LT5\endcsname{\color{black}}%
      \expandafter\def\csname LT6\endcsname{\color{black}}%
      \expandafter\def\csname LT7\endcsname{\color{black}}%
      \expandafter\def\csname LT8\endcsname{\color{black}}%
    \fi
  \fi
    \setlength{\unitlength}{0.0500bp}%
    \ifx\gptboxheight\undefined%
      \newlength{\gptboxheight}%
      \newlength{\gptboxwidth}%
      \newsavebox{\gptboxtext}%
    \fi%
    \setlength{\fboxrule}{0.5pt}%
    \setlength{\fboxsep}{1pt}%
    \definecolor{tbcol}{rgb}{1,1,1}%
\begin{picture}(1980.00,1700.00)%
    \gplgaddtomacro\gplbacktext{%
      \csname LTb\endcsname
      \put(566,627){\makebox(0,0)[r]{\strut{}$0$}}%
      \csname LTb\endcsname
      \put(566,1043){\makebox(0,0)[r]{\strut{}$0.5$}}%
      \csname LTb\endcsname
      \put(566,1458){\makebox(0,0)[r]{\strut{}$1$}}%
      \csname LTb\endcsname
      \put(649,431){\makebox(0,0){\strut{}$1$}}%
      \csname LTb\endcsname
      \put(1712,431){\makebox(0,0){\strut{}$472$}}%
    }%
    \gplgaddtomacro\gplfronttext{%
      \csname LTb\endcsname
      \put(159,1055){\rotatebox{-270.00}{\makebox(0,0){\strut{}\normalsize Norm. $f_{\chi^2}$}}}%
      \csname LTb\endcsname
      \put(1181,137){\makebox(0,0){\strut{}\normalsize Score rank}}%
    }%
    \gplbacktext
    \put(0,0){\includegraphics[width={99.00bp},height={85.00bp}]{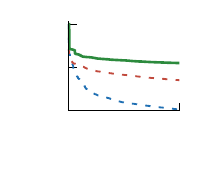}}%
    \gplfronttext
  \end{picture}%
\endgroup

%% file: curve.ecoli.l8d3.dedup.tex
\begingroup
  \makeatletter
  \providecommand\color[2][]{%
    \GenericError{(gnuplot) \space\space\space\@spaces}{%
      Package color not loaded in conjunction with
      terminal option `colourtext'%
    }{See the gnuplot documentation for explanation.%
    }{Either use 'blacktext' in gnuplot or load the package
      color.sty in LaTeX.}%
    \renewcommand\color[2][]{}%
  }%
  \providecommand\includegraphics[2][]{%
    \GenericError{(gnuplot) \space\space\space\@spaces}{%
      Package graphicx or graphics not loaded%
    }{See the gnuplot documentation for explanation.%
    }{The gnuplot epslatex terminal needs graphicx.sty or graphics.sty.}%
    \renewcommand\includegraphics[2][]{}%
  }%
  \providecommand\rotatebox[2]{#2}%
  \@ifundefined{ifGPcolor}{%
    \newif\ifGPcolor
    \GPcolortrue
  }{}%
  \@ifundefined{ifGPblacktext}{%
    \newif\ifGPblacktext
    \GPblacktexttrue
  }{}%
  \let\gplgaddtomacro\g@addto@macro
  \gdef\gplbacktext{}%
  \gdef\gplfronttext{}%
  \makeatother
  \ifGPblacktext
    \def\colorrgb#1{}%
    \def\colorgray#1{}%
  \else
    \ifGPcolor
      \def\colorrgb#1{\color[rgb]{#1}}%
      \def\colorgray#1{\color[gray]{#1}}%
      \expandafter\def\csname LTw\endcsname{\color{white}}%
      \expandafter\def\csname LTb\endcsname{\color{black}}%
      \expandafter\def\csname LTa\endcsname{\color{black}}%
      \expandafter\def\csname LT0\endcsname{\color[rgb]{1,0,0}}%
      \expandafter\def\csname LT1\endcsname{\color[rgb]{0,1,0}}%
      \expandafter\def\csname LT2\endcsname{\color[rgb]{0,0,1}}%
      \expandafter\def\csname LT3\endcsname{\color[rgb]{1,0,1}}%
      \expandafter\def\csname LT4\endcsname{\color[rgb]{0,1,1}}%
      \expandafter\def\csname LT5\endcsname{\color[rgb]{1,1,0}}%
      \expandafter\def\csname LT6\endcsname{\color[rgb]{0,0,0}}%
      \expandafter\def\csname LT7\endcsname{\color[rgb]{1,0.3,0}}%
      \expandafter\def\csname LT8\endcsname{\color[rgb]{0.5,0.5,0.5}}%
    \else
      \def\colorrgb#1{\color{black}}%
      \def\colorgray#1{\color[gray]{#1}}%
      \expandafter\def\csname LTw\endcsname{\color{white}}%
      \expandafter\def\csname LTb\endcsname{\color{black}}%
      \expandafter\def\csname LTa\endcsname{\color{black}}%
      \expandafter\def\csname LT0\endcsname{\color{black}}%
      \expandafter\def\csname LT1\endcsname{\color{black}}%
      \expandafter\def\csname LT2\endcsname{\color{black}}%
      \expandafter\def\csname LT3\endcsname{\color{black}}%
      \expandafter\def\csname LT4\endcsname{\color{black}}%
      \expandafter\def\csname LT5\endcsname{\color{black}}%
      \expandafter\def\csname LT6\endcsname{\color{black}}%
      \expandafter\def\csname LT7\endcsname{\color{black}}%
      \expandafter\def\csname LT8\endcsname{\color{black}}%
    \fi
  \fi
    \setlength{\unitlength}{0.0500bp}%
    \ifx\gptboxheight\undefined%
      \newlength{\gptboxheight}%
      \newlength{\gptboxwidth}%
      \newsavebox{\gptboxtext}%
    \fi%
    \setlength{\fboxrule}{0.5pt}%
    \setlength{\fboxsep}{1pt}%
    \definecolor{tbcol}{rgb}{1,1,1}%
\begin{picture}(1980.00,1700.00)%
    \gplgaddtomacro\gplbacktext{%
      \csname LTb\endcsname
      \put(566,627){\makebox(0,0)[r]{\strut{}$0$}}%
      \csname LTb\endcsname
      \put(566,1043){\makebox(0,0)[r]{\strut{}$0.5$}}%
      \csname LTb\endcsname
      \put(566,1458){\makebox(0,0)[r]{\strut{}$1$}}%
      \csname LTb\endcsname
      \put(649,431){\makebox(0,0){\strut{}$1$}}%
      \csname LTb\endcsname
      \put(1712,431){\makebox(0,0){\strut{}$694$}}%
    }%
    \gplgaddtomacro\gplfronttext{%
      \csname LTb\endcsname
      \put(159,1055){\rotatebox{-270.00}{\makebox(0,0){\strut{}\normalsize Norm. $f_{\chi^2}$}}}%
      \csname LTb\endcsname
      \put(1181,137){\makebox(0,0){\strut{}\normalsize Score rank}}%
    }%
    \gplbacktext
    \put(0,0){\includegraphics[width={99.00bp},height={85.00bp}]{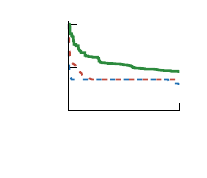}}%
    \gplfronttext
  \end{picture}%
\endgroup

%% file: curve.yeast.l8d3.dedup.tex
\begingroup
  \makeatletter
  \providecommand\color[2][]{%
    \GenericError{(gnuplot) \space\space\space\@spaces}{%
      Package color not loaded in conjunction with
      terminal option `colourtext'%
    }{See the gnuplot documentation for explanation.%
    }{Either use 'blacktext' in gnuplot or load the package
      color.sty in LaTeX.}%
    \renewcommand\color[2][]{}%
  }%
  \providecommand\includegraphics[2][]{%
    \GenericError{(gnuplot) \space\space\space\@spaces}{%
      Package graphicx or graphics not loaded%
    }{See the gnuplot documentation for explanation.%
    }{The gnuplot epslatex terminal needs graphicx.sty or graphics.sty.}%
    \renewcommand\includegraphics[2][]{}%
  }%
  \providecommand\rotatebox[2]{#2}%
  \@ifundefined{ifGPcolor}{%
    \newif\ifGPcolor
    \GPcolortrue
  }{}%
  \@ifundefined{ifGPblacktext}{%
    \newif\ifGPblacktext
    \GPblacktexttrue
  }{}%
  \let\gplgaddtomacro\g@addto@macro
  \gdef\gplbacktext{}%
  \gdef\gplfronttext{}%
  \makeatother
  \ifGPblacktext
    \def\colorrgb#1{}%
    \def\colorgray#1{}%
  \else
    \ifGPcolor
      \def\colorrgb#1{\color[rgb]{#1}}%
      \def\colorgray#1{\color[gray]{#1}}%
      \expandafter\def\csname LTw\endcsname{\color{white}}%
      \expandafter\def\csname LTb\endcsname{\color{black}}%
      \expandafter\def\csname LTa\endcsname{\color{black}}%
      \expandafter\def\csname LT0\endcsname{\color[rgb]{1,0,0}}%
      \expandafter\def\csname LT1\endcsname{\color[rgb]{0,1,0}}%
      \expandafter\def\csname LT2\endcsname{\color[rgb]{0,0,1}}%
      \expandafter\def\csname LT3\endcsname{\color[rgb]{1,0,1}}%
      \expandafter\def\csname LT4\endcsname{\color[rgb]{0,1,1}}%
      \expandafter\def\csname LT5\endcsname{\color[rgb]{1,1,0}}%
      \expandafter\def\csname LT6\endcsname{\color[rgb]{0,0,0}}%
      \expandafter\def\csname LT7\endcsname{\color[rgb]{1,0.3,0}}%
      \expandafter\def\csname LT8\endcsname{\color[rgb]{0.5,0.5,0.5}}%
    \else
      \def\colorrgb#1{\color{black}}%
      \def\colorgray#1{\color[gray]{#1}}%
      \expandafter\def\csname LTw\endcsname{\color{white}}%
      \expandafter\def\csname LTb\endcsname{\color{black}}%
      \expandafter\def\csname LTa\endcsname{\color{black}}%
      \expandafter\def\csname LT0\endcsname{\color{black}}%
      \expandafter\def\csname LT1\endcsname{\color{black}}%
      \expandafter\def\csname LT2\endcsname{\color{black}}%
      \expandafter\def\csname LT3\endcsname{\color{black}}%
      \expandafter\def\csname LT4\endcsname{\color{black}}%
      \expandafter\def\csname LT5\endcsname{\color{black}}%
      \expandafter\def\csname LT6\endcsname{\color{black}}%
      \expandafter\def\csname LT7\endcsname{\color{black}}%
      \expandafter\def\csname LT8\endcsname{\color{black}}%
    \fi
  \fi
    \setlength{\unitlength}{0.0500bp}%
    \ifx\gptboxheight\undefined%
      \newlength{\gptboxheight}%
      \newlength{\gptboxwidth}%
      \newsavebox{\gptboxtext}%
    \fi%
    \setlength{\fboxrule}{0.5pt}%
    \setlength{\fboxsep}{1pt}%
    \definecolor{tbcol}{rgb}{1,1,1}%
\begin{picture}(1980.00,1700.00)%
    \gplgaddtomacro\gplbacktext{%
      \csname LTb\endcsname
      \put(566,627){\makebox(0,0)[r]{\strut{}$0$}}%
      \csname LTb\endcsname
      \put(566,1043){\makebox(0,0)[r]{\strut{}$0.5$}}%
      \csname LTb\endcsname
      \put(566,1458){\makebox(0,0)[r]{\strut{}$1$}}%
      \csname LTb\endcsname
      \put(649,431){\makebox(0,0){\strut{}$1$}}%
      \csname LTb\endcsname
      \put(1712,431){\makebox(0,0){\strut{}$231$}}%
    }%
    \gplgaddtomacro\gplfronttext{%
      \csname LTb\endcsname
      \put(159,1055){\rotatebox{-270.00}{\makebox(0,0){\strut{}\normalsize Norm. $f_{\chi^2}$}}}%
      \csname LTb\endcsname
      \put(1181,137){\makebox(0,0){\strut{}\normalsize Score rank}}%
    }%
    \gplbacktext
    \put(0,0){\includegraphics[width={99.00bp},height={85.00bp}]{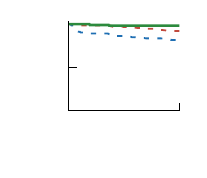}}%
    \gplfronttext
  \end{picture}%
\endgroup

%% file: curve.fly.l8d3.dedup.tex
\begingroup
  \makeatletter
  \providecommand\color[2][]{%
    \GenericError{(gnuplot) \space\space\space\@spaces}{%
      Package color not loaded in conjunction with
      terminal option `colourtext'%
    }{See the gnuplot documentation for explanation.%
    }{Either use 'blacktext' in gnuplot or load the package
      color.sty in LaTeX.}%
    \renewcommand\color[2][]{}%
  }%
  \providecommand\includegraphics[2][]{%
    \GenericError{(gnuplot) \space\space\space\@spaces}{%
      Package graphicx or graphics not loaded%
    }{See the gnuplot documentation for explanation.%
    }{The gnuplot epslatex terminal needs graphicx.sty or graphics.sty.}%
    \renewcommand\includegraphics[2][]{}%
  }%
  \providecommand\rotatebox[2]{#2}%
  \@ifundefined{ifGPcolor}{%
    \newif\ifGPcolor
    \GPcolortrue
  }{}%
  \@ifundefined{ifGPblacktext}{%
    \newif\ifGPblacktext
    \GPblacktexttrue
  }{}%
  \let\gplgaddtomacro\g@addto@macro
  \gdef\gplbacktext{}%
  \gdef\gplfronttext{}%
  \makeatother
  \ifGPblacktext
    \def\colorrgb#1{}%
    \def\colorgray#1{}%
  \else
    \ifGPcolor
      \def\colorrgb#1{\color[rgb]{#1}}%
      \def\colorgray#1{\color[gray]{#1}}%
      \expandafter\def\csname LTw\endcsname{\color{white}}%
      \expandafter\def\csname LTb\endcsname{\color{black}}%
      \expandafter\def\csname LTa\endcsname{\color{black}}%
      \expandafter\def\csname LT0\endcsname{\color[rgb]{1,0,0}}%
      \expandafter\def\csname LT1\endcsname{\color[rgb]{0,1,0}}%
      \expandafter\def\csname LT2\endcsname{\color[rgb]{0,0,1}}%
      \expandafter\def\csname LT3\endcsname{\color[rgb]{1,0,1}}%
      \expandafter\def\csname LT4\endcsname{\color[rgb]{0,1,1}}%
      \expandafter\def\csname LT5\endcsname{\color[rgb]{1,1,0}}%
      \expandafter\def\csname LT6\endcsname{\color[rgb]{0,0,0}}%
      \expandafter\def\csname LT7\endcsname{\color[rgb]{1,0.3,0}}%
      \expandafter\def\csname LT8\endcsname{\color[rgb]{0.5,0.5,0.5}}%
    \else
      \def\colorrgb#1{\color{black}}%
      \def\colorgray#1{\color[gray]{#1}}%
      \expandafter\def\csname LTw\endcsname{\color{white}}%
      \expandafter\def\csname LTb\endcsname{\color{black}}%
      \expandafter\def\csname LTa\endcsname{\color{black}}%
      \expandafter\def\csname LT0\endcsname{\color{black}}%
      \expandafter\def\csname LT1\endcsname{\color{black}}%
      \expandafter\def\csname LT2\endcsname{\color{black}}%
      \expandafter\def\csname LT3\endcsname{\color{black}}%
      \expandafter\def\csname LT4\endcsname{\color{black}}%
      \expandafter\def\csname LT5\endcsname{\color{black}}%
      \expandafter\def\csname LT6\endcsname{\color{black}}%
      \expandafter\def\csname LT7\endcsname{\color{black}}%
      \expandafter\def\csname LT8\endcsname{\color{black}}%
    \fi
  \fi
    \setlength{\unitlength}{0.0500bp}%
    \ifx\gptboxheight\undefined%
      \newlength{\gptboxheight}%
      \newlength{\gptboxwidth}%
      \newsavebox{\gptboxtext}%
    \fi%
    \setlength{\fboxrule}{0.5pt}%
    \setlength{\fboxsep}{1pt}%
    \definecolor{tbcol}{rgb}{1,1,1}%
\begin{picture}(1980.00,1700.00)%
    \gplgaddtomacro\gplbacktext{%
      \csname LTb\endcsname
      \put(566,627){\makebox(0,0)[r]{\strut{}$0$}}%
      \csname LTb\endcsname
      \put(566,1043){\makebox(0,0)[r]{\strut{}$0.5$}}%
      \csname LTb\endcsname
      \put(566,1458){\makebox(0,0)[r]{\strut{}$1$}}%
      \csname LTb\endcsname
      \put(649,431){\makebox(0,0){\strut{}$1$}}%
      \csname LTb\endcsname
      \put(1712,431){\makebox(0,0){\strut{}$893$}}%
    }%
    \gplgaddtomacro\gplfronttext{%
      \csname LTb\endcsname
      \put(159,1055){\rotatebox{-270.00}{\makebox(0,0){\strut{}\normalsize Norm. $f_{\chi^2}$}}}%
      \csname LTb\endcsname
      \put(1181,137){\makebox(0,0){\strut{}\normalsize Score rank}}%
    }%
    \gplbacktext
    \put(0,0){\includegraphics[width={99.00bp},height={85.00bp}]{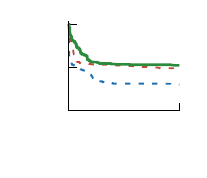}}%
    \gplfronttext
  \end{picture}%
\endgroup

%% file: curve.plant.l8d3.dedup.tex
\begingroup
  \makeatletter
  \providecommand\color[2][]{%
    \GenericError{(gnuplot) \space\space\space\@spaces}{%
      Package color not loaded in conjunction with
      terminal option `colourtext'%
    }{See the gnuplot documentation for explanation.%
    }{Either use 'blacktext' in gnuplot or load the package
      color.sty in LaTeX.}%
    \renewcommand\color[2][]{}%
  }%
  \providecommand\includegraphics[2][]{%
    \GenericError{(gnuplot) \space\space\space\@spaces}{%
      Package graphicx or graphics not loaded%
    }{See the gnuplot documentation for explanation.%
    }{The gnuplot epslatex terminal needs graphicx.sty or graphics.sty.}%
    \renewcommand\includegraphics[2][]{}%
  }%
  \providecommand\rotatebox[2]{#2}%
  \@ifundefined{ifGPcolor}{%
    \newif\ifGPcolor
    \GPcolortrue
  }{}%
  \@ifundefined{ifGPblacktext}{%
    \newif\ifGPblacktext
    \GPblacktexttrue
  }{}%
  \let\gplgaddtomacro\g@addto@macro
  \gdef\gplbacktext{}%
  \gdef\gplfronttext{}%
  \makeatother
  \ifGPblacktext
    \def\colorrgb#1{}%
    \def\colorgray#1{}%
  \else
    \ifGPcolor
      \def\colorrgb#1{\color[rgb]{#1}}%
      \def\colorgray#1{\color[gray]{#1}}%
      \expandafter\def\csname LTw\endcsname{\color{white}}%
      \expandafter\def\csname LTb\endcsname{\color{black}}%
      \expandafter\def\csname LTa\endcsname{\color{black}}%
      \expandafter\def\csname LT0\endcsname{\color[rgb]{1,0,0}}%
      \expandafter\def\csname LT1\endcsname{\color[rgb]{0,1,0}}%
      \expandafter\def\csname LT2\endcsname{\color[rgb]{0,0,1}}%
      \expandafter\def\csname LT3\endcsname{\color[rgb]{1,0,1}}%
      \expandafter\def\csname LT4\endcsname{\color[rgb]{0,1,1}}%
      \expandafter\def\csname LT5\endcsname{\color[rgb]{1,1,0}}%
      \expandafter\def\csname LT6\endcsname{\color[rgb]{0,0,0}}%
      \expandafter\def\csname LT7\endcsname{\color[rgb]{1,0.3,0}}%
      \expandafter\def\csname LT8\endcsname{\color[rgb]{0.5,0.5,0.5}}%
    \else
      \def\colorrgb#1{\color{black}}%
      \def\colorgray#1{\color[gray]{#1}}%
      \expandafter\def\csname LTw\endcsname{\color{white}}%
      \expandafter\def\csname LTb\endcsname{\color{black}}%
      \expandafter\def\csname LTa\endcsname{\color{black}}%
      \expandafter\def\csname LT0\endcsname{\color{black}}%
      \expandafter\def\csname LT1\endcsname{\color{black}}%
      \expandafter\def\csname LT2\endcsname{\color{black}}%
      \expandafter\def\csname LT3\endcsname{\color{black}}%
      \expandafter\def\csname LT4\endcsname{\color{black}}%
      \expandafter\def\csname LT5\endcsname{\color{black}}%
      \expandafter\def\csname LT6\endcsname{\color{black}}%
      \expandafter\def\csname LT7\endcsname{\color{black}}%
      \expandafter\def\csname LT8\endcsname{\color{black}}%
    \fi
  \fi
    \setlength{\unitlength}{0.0500bp}%
    \ifx\gptboxheight\undefined%
      \newlength{\gptboxheight}%
      \newlength{\gptboxwidth}%
      \newsavebox{\gptboxtext}%
    \fi%
    \setlength{\fboxrule}{0.5pt}%
    \setlength{\fboxsep}{1pt}%
    \definecolor{tbcol}{rgb}{1,1,1}%
\begin{picture}(1980.00,1700.00)%
    \gplgaddtomacro\gplbacktext{%
      \csname LTb\endcsname
      \put(566,627){\makebox(0,0)[r]{\strut{}$0$}}%
      \csname LTb\endcsname
      \put(566,1043){\makebox(0,0)[r]{\strut{}$0.5$}}%
      \csname LTb\endcsname
      \put(566,1458){\makebox(0,0)[r]{\strut{}$1$}}%
      \csname LTb\endcsname
      \put(649,431){\makebox(0,0){\strut{}$1$}}%
      \csname LTb\endcsname
      \put(1712,431){\makebox(0,0){\strut{}$2384$}}%
    }%
    \gplgaddtomacro\gplfronttext{%
      \csname LTb\endcsname
      \put(159,1055){\rotatebox{-270.00}{\makebox(0,0){\strut{}\normalsize Norm. $f_{\chi^2}$}}}%
      \csname LTb\endcsname
      \put(1181,137){\makebox(0,0){\strut{}\normalsize Score rank}}%
    }%
    \gplbacktext
    \put(0,0){\includegraphics[width={99.00bp},height={85.00bp}]{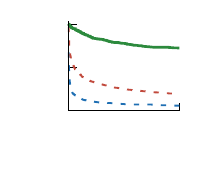}}%
    \gplfronttext
  \end{picture}%
\endgroup

%% file: curve.human.l8d3.dedup.tex
\begingroup
  \makeatletter
  \providecommand\color[2][]{%
    \GenericError{(gnuplot) \space\space\space\@spaces}{%
      Package color not loaded in conjunction with
      terminal option `colourtext'%
    }{See the gnuplot documentation for explanation.%
    }{Either use 'blacktext' in gnuplot or load the package
      color.sty in LaTeX.}%
    \renewcommand\color[2][]{}%
  }%
  \providecommand\includegraphics[2][]{%
    \GenericError{(gnuplot) \space\space\space\@spaces}{%
      Package graphicx or graphics not loaded%
    }{See the gnuplot documentation for explanation.%
    }{The gnuplot epslatex terminal needs graphicx.sty or graphics.sty.}%
    \renewcommand\includegraphics[2][]{}%
  }%
  \providecommand\rotatebox[2]{#2}%
  \@ifundefined{ifGPcolor}{%
    \newif\ifGPcolor
    \GPcolortrue
  }{}%
  \@ifundefined{ifGPblacktext}{%
    \newif\ifGPblacktext
    \GPblacktexttrue
  }{}%
  \let\gplgaddtomacro\g@addto@macro
  \gdef\gplbacktext{}%
  \gdef\gplfronttext{}%
  \makeatother
  \ifGPblacktext
    \def\colorrgb#1{}%
    \def\colorgray#1{}%
  \else
    \ifGPcolor
      \def\colorrgb#1{\color[rgb]{#1}}%
      \def\colorgray#1{\color[gray]{#1}}%
      \expandafter\def\csname LTw\endcsname{\color{white}}%
      \expandafter\def\csname LTb\endcsname{\color{black}}%
      \expandafter\def\csname LTa\endcsname{\color{black}}%
      \expandafter\def\csname LT0\endcsname{\color[rgb]{1,0,0}}%
      \expandafter\def\csname LT1\endcsname{\color[rgb]{0,1,0}}%
      \expandafter\def\csname LT2\endcsname{\color[rgb]{0,0,1}}%
      \expandafter\def\csname LT3\endcsname{\color[rgb]{1,0,1}}%
      \expandafter\def\csname LT4\endcsname{\color[rgb]{0,1,1}}%
      \expandafter\def\csname LT5\endcsname{\color[rgb]{1,1,0}}%
      \expandafter\def\csname LT6\endcsname{\color[rgb]{0,0,0}}%
      \expandafter\def\csname LT7\endcsname{\color[rgb]{1,0.3,0}}%
      \expandafter\def\csname LT8\endcsname{\color[rgb]{0.5,0.5,0.5}}%
    \else
      \def\colorrgb#1{\color{black}}%
      \def\colorgray#1{\color[gray]{#1}}%
      \expandafter\def\csname LTw\endcsname{\color{white}}%
      \expandafter\def\csname LTb\endcsname{\color{black}}%
      \expandafter\def\csname LTa\endcsname{\color{black}}%
      \expandafter\def\csname LT0\endcsname{\color{black}}%
      \expandafter\def\csname LT1\endcsname{\color{black}}%
      \expandafter\def\csname LT2\endcsname{\color{black}}%
      \expandafter\def\csname LT3\endcsname{\color{black}}%
      \expandafter\def\csname LT4\endcsname{\color{black}}%
      \expandafter\def\csname LT5\endcsname{\color{black}}%
      \expandafter\def\csname LT6\endcsname{\color{black}}%
      \expandafter\def\csname LT7\endcsname{\color{black}}%
      \expandafter\def\csname LT8\endcsname{\color{black}}%
    \fi
  \fi
    \setlength{\unitlength}{0.0500bp}%
    \ifx\gptboxheight\undefined%
      \newlength{\gptboxheight}%
      \newlength{\gptboxwidth}%
      \newsavebox{\gptboxtext}%
    \fi%
    \setlength{\fboxrule}{0.5pt}%
    \setlength{\fboxsep}{1pt}%
    \definecolor{tbcol}{rgb}{1,1,1}%
\begin{picture}(1980.00,1700.00)%
    \gplgaddtomacro\gplbacktext{%
      \csname LTb\endcsname
      \put(566,627){\makebox(0,0)[r]{\strut{}$0$}}%
      \csname LTb\endcsname
      \put(566,1043){\makebox(0,0)[r]{\strut{}$0.5$}}%
      \csname LTb\endcsname
      \put(566,1458){\makebox(0,0)[r]{\strut{}$1$}}%
      \csname LTb\endcsname
      \put(649,431){\makebox(0,0){\strut{}$1$}}%
      \csname LTb\endcsname
      \put(1712,431){\makebox(0,0){\strut{}$1836$}}%
    }%
    \gplgaddtomacro\gplfronttext{%
      \csname LTb\endcsname
      \put(159,1055){\rotatebox{-270.00}{\makebox(0,0){\strut{}\normalsize Norm. $f_{\chi^2}$}}}%
      \csname LTb\endcsname
      \put(1181,137){\makebox(0,0){\strut{}\normalsize Score rank}}%
    }%
    \gplbacktext
    \put(0,0){\includegraphics[width={99.00bp},height={85.00bp}]{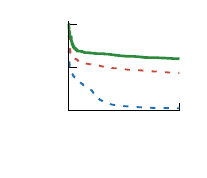}}%
    \gplfronttext
  \end{picture}%
\endgroup

%% file: curve.yeast_aj.l8d3.dedup.tex
\begingroup
  \makeatletter
  \providecommand\color[2][]{%
    \GenericError{(gnuplot) \space\space\space\@spaces}{%
      Package color not loaded in conjunction with
      terminal option `colourtext'%
    }{See the gnuplot documentation for explanation.%
    }{Either use 'blacktext' in gnuplot or load the package
      color.sty in LaTeX.}%
    \renewcommand\color[2][]{}%
  }%
  \providecommand\includegraphics[2][]{%
    \GenericError{(gnuplot) \space\space\space\@spaces}{%
      Package graphicx or graphics not loaded%
    }{See the gnuplot documentation for explanation.%
    }{The gnuplot epslatex terminal needs graphicx.sty or graphics.sty.}%
    \renewcommand\includegraphics[2][]{}%
  }%
  \providecommand\rotatebox[2]{#2}%
  \@ifundefined{ifGPcolor}{%
    \newif\ifGPcolor
    \GPcolortrue
  }{}%
  \@ifundefined{ifGPblacktext}{%
    \newif\ifGPblacktext
    \GPblacktexttrue
  }{}%
  \let\gplgaddtomacro\g@addto@macro
  \gdef\gplbacktext{}%
  \gdef\gplfronttext{}%
  \makeatother
  \ifGPblacktext
    \def\colorrgb#1{}%
    \def\colorgray#1{}%
  \else
    \ifGPcolor
      \def\colorrgb#1{\color[rgb]{#1}}%
      \def\colorgray#1{\color[gray]{#1}}%
      \expandafter\def\csname LTw\endcsname{\color{white}}%
      \expandafter\def\csname LTb\endcsname{\color{black}}%
      \expandafter\def\csname LTa\endcsname{\color{black}}%
      \expandafter\def\csname LT0\endcsname{\color[rgb]{1,0,0}}%
      \expandafter\def\csname LT1\endcsname{\color[rgb]{0,1,0}}%
      \expandafter\def\csname LT2\endcsname{\color[rgb]{0,0,1}}%
      \expandafter\def\csname LT3\endcsname{\color[rgb]{1,0,1}}%
      \expandafter\def\csname LT4\endcsname{\color[rgb]{0,1,1}}%
      \expandafter\def\csname LT5\endcsname{\color[rgb]{1,1,0}}%
      \expandafter\def\csname LT6\endcsname{\color[rgb]{0,0,0}}%
      \expandafter\def\csname LT7\endcsname{\color[rgb]{1,0.3,0}}%
      \expandafter\def\csname LT8\endcsname{\color[rgb]{0.5,0.5,0.5}}%
    \else
      \def\colorrgb#1{\color{black}}%
      \def\colorgray#1{\color[gray]{#1}}%
      \expandafter\def\csname LTw\endcsname{\color{white}}%
      \expandafter\def\csname LTb\endcsname{\color{black}}%
      \expandafter\def\csname LTa\endcsname{\color{black}}%
      \expandafter\def\csname LT0\endcsname{\color{black}}%
      \expandafter\def\csname LT1\endcsname{\color{black}}%
      \expandafter\def\csname LT2\endcsname{\color{black}}%
      \expandafter\def\csname LT3\endcsname{\color{black}}%
      \expandafter\def\csname LT4\endcsname{\color{black}}%
      \expandafter\def\csname LT5\endcsname{\color{black}}%
      \expandafter\def\csname LT6\endcsname{\color{black}}%
      \expandafter\def\csname LT7\endcsname{\color{black}}%
      \expandafter\def\csname LT8\endcsname{\color{black}}%
    \fi
  \fi
    \setlength{\unitlength}{0.0500bp}%
    \ifx\gptboxheight\undefined%
      \newlength{\gptboxheight}%
      \newlength{\gptboxwidth}%
      \newsavebox{\gptboxtext}%
    \fi%
    \setlength{\fboxrule}{0.5pt}%
    \setlength{\fboxsep}{1pt}%
    \definecolor{tbcol}{rgb}{1,1,1}%
\begin{picture}(1980.00,1700.00)%
    \gplgaddtomacro\gplbacktext{%
      \csname LTb\endcsname
      \put(566,627){\makebox(0,0)[r]{\strut{}$0$}}%
      \csname LTb\endcsname
      \put(566,1043){\makebox(0,0)[r]{\strut{}$0.5$}}%
      \csname LTb\endcsname
      \put(566,1458){\makebox(0,0)[r]{\strut{}$1$}}%
      \csname LTb\endcsname
      \put(649,431){\makebox(0,0){\strut{}$1$}}%
      \csname LTb\endcsname
      \put(1712,431){\makebox(0,0){\strut{}$5792$}}%
    }%
    \gplgaddtomacro\gplfronttext{%
      \csname LTb\endcsname
      \put(159,1055){\rotatebox{-270.00}{\makebox(0,0){\strut{}\normalsize Norm. $f_{\chi^2}$}}}%
      \csname LTb\endcsname
      \put(1181,137){\makebox(0,0){\strut{}\normalsize Score rank}}%
    }%
    \gplbacktext
    \put(0,0){\includegraphics[width={99.00bp},height={85.00bp}]{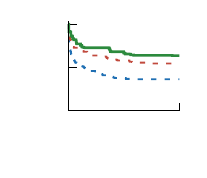}}%
    \gplfronttext
  \end{picture}%
\endgroup

%% file: curve.human_aj.l8d3.dedup.tex
\begingroup
  \makeatletter
  \providecommand\color[2][]{%
    \GenericError{(gnuplot) \space\space\space\@spaces}{%
      Package color not loaded in conjunction with
      terminal option `colourtext'%
    }{See the gnuplot documentation for explanation.%
    }{Either use 'blacktext' in gnuplot or load the package
      color.sty in LaTeX.}%
    \renewcommand\color[2][]{}%
  }%
  \providecommand\includegraphics[2][]{%
    \GenericError{(gnuplot) \space\space\space\@spaces}{%
      Package graphicx or graphics not loaded%
    }{See the gnuplot documentation for explanation.%
    }{The gnuplot epslatex terminal needs graphicx.sty or graphics.sty.}%
    \renewcommand\includegraphics[2][]{}%
  }%
  \providecommand\rotatebox[2]{#2}%
  \@ifundefined{ifGPcolor}{%
    \newif\ifGPcolor
    \GPcolortrue
  }{}%
  \@ifundefined{ifGPblacktext}{%
    \newif\ifGPblacktext
    \GPblacktexttrue
  }{}%
  \let\gplgaddtomacro\g@addto@macro
  \gdef\gplbacktext{}%
  \gdef\gplfronttext{}%
  \makeatother
  \ifGPblacktext
    \def\colorrgb#1{}%
    \def\colorgray#1{}%
  \else
    \ifGPcolor
      \def\colorrgb#1{\color[rgb]{#1}}%
      \def\colorgray#1{\color[gray]{#1}}%
      \expandafter\def\csname LTw\endcsname{\color{white}}%
      \expandafter\def\csname LTb\endcsname{\color{black}}%
      \expandafter\def\csname LTa\endcsname{\color{black}}%
      \expandafter\def\csname LT0\endcsname{\color[rgb]{1,0,0}}%
      \expandafter\def\csname LT1\endcsname{\color[rgb]{0,1,0}}%
      \expandafter\def\csname LT2\endcsname{\color[rgb]{0,0,1}}%
      \expandafter\def\csname LT3\endcsname{\color[rgb]{1,0,1}}%
      \expandafter\def\csname LT4\endcsname{\color[rgb]{0,1,1}}%
      \expandafter\def\csname LT5\endcsname{\color[rgb]{1,1,0}}%
      \expandafter\def\csname LT6\endcsname{\color[rgb]{0,0,0}}%
      \expandafter\def\csname LT7\endcsname{\color[rgb]{1,0.3,0}}%
      \expandafter\def\csname LT8\endcsname{\color[rgb]{0.5,0.5,0.5}}%
    \else
      \def\colorrgb#1{\color{black}}%
      \def\colorgray#1{\color[gray]{#1}}%
      \expandafter\def\csname LTw\endcsname{\color{white}}%
      \expandafter\def\csname LTb\endcsname{\color{black}}%
      \expandafter\def\csname LTa\endcsname{\color{black}}%
      \expandafter\def\csname LT0\endcsname{\color{black}}%
      \expandafter\def\csname LT1\endcsname{\color{black}}%
      \expandafter\def\csname LT2\endcsname{\color{black}}%
      \expandafter\def\csname LT3\endcsname{\color{black}}%
      \expandafter\def\csname LT4\endcsname{\color{black}}%
      \expandafter\def\csname LT5\endcsname{\color{black}}%
      \expandafter\def\csname LT6\endcsname{\color{black}}%
      \expandafter\def\csname LT7\endcsname{\color{black}}%
      \expandafter\def\csname LT8\endcsname{\color{black}}%
    \fi
  \fi
    \setlength{\unitlength}{0.0500bp}%
    \ifx\gptboxheight\undefined%
      \newlength{\gptboxheight}%
      \newlength{\gptboxwidth}%
      \newsavebox{\gptboxtext}%
    \fi%
    \setlength{\fboxrule}{0.5pt}%
    \setlength{\fboxsep}{1pt}%
    \definecolor{tbcol}{rgb}{1,1,1}%
\begin{picture}(1980.00,1700.00)%
    \gplgaddtomacro\gplbacktext{%
      \csname LTb\endcsname
      \put(566,627){\makebox(0,0)[r]{\strut{}$0$}}%
      \csname LTb\endcsname
      \put(566,1043){\makebox(0,0)[r]{\strut{}$0.5$}}%
      \csname LTb\endcsname
      \put(566,1458){\makebox(0,0)[r]{\strut{}$1$}}%
      \csname LTb\endcsname
      \put(649,431){\makebox(0,0){\strut{}$1$}}%
      \csname LTb\endcsname
      \put(1712,431){\makebox(0,0){\strut{}$4558$}}%
    }%
    \gplgaddtomacro\gplfronttext{%
      \csname LTb\endcsname
      \put(159,1055){\rotatebox{-270.00}{\makebox(0,0){\strut{}\normalsize Norm. $f_{\chi^2}$}}}%
      \csname LTb\endcsname
      \put(1181,137){\makebox(0,0){\strut{}\normalsize Score rank}}%
    }%
    \gplbacktext
    \put(0,0){\includegraphics[width={99.00bp},height={85.00bp}]{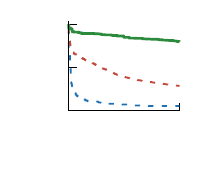}}%
    \gplfronttext
  \end{picture}%
\endgroup

%% file: density_legend.tex
\begingroup
  \makeatletter
  \providecommand\color[2][]{%
    \GenericError{(gnuplot) \space\space\space\@spaces}{%
      Package color not loaded in conjunction with
      terminal option `colourtext'%
    }{See the gnuplot documentation for explanation.%
    }{Either use 'blacktext' in gnuplot or load the package
      color.sty in LaTeX.}%
    \renewcommand\color[2][]{}%
  }%
  \providecommand\includegraphics[2][]{%
    \GenericError{(gnuplot) \space\space\space\@spaces}{%
      Package graphicx or graphics not loaded%
    }{See the gnuplot documentation for explanation.%
    }{The gnuplot epslatex terminal needs graphicx.sty or graphics.sty.}%
    \renewcommand\includegraphics[2][]{}%
  }%
  \providecommand\rotatebox[2]{#2}%
  \@ifundefined{ifGPcolor}{%
    \newif\ifGPcolor
    \GPcolortrue
  }{}%
  \@ifundefined{ifGPblacktext}{%
    \newif\ifGPblacktext
    \GPblacktexttrue
  }{}%
  \let\gplgaddtomacro\g@addto@macro
  \gdef\gplbacktext{}%
  \gdef\gplfronttext{}%
  \makeatother
  \ifGPblacktext
    \def\colorrgb#1{}%
    \def\colorgray#1{}%
  \else
    \ifGPcolor
      \def\colorrgb#1{\color[rgb]{#1}}%
      \def\colorgray#1{\color[gray]{#1}}%
      \expandafter\def\csname LTw\endcsname{\color{white}}%
      \expandafter\def\csname LTb\endcsname{\color{black}}%
      \expandafter\def\csname LTa\endcsname{\color{black}}%
      \expandafter\def\csname LT0\endcsname{\color[rgb]{1,0,0}}%
      \expandafter\def\csname LT1\endcsname{\color[rgb]{0,1,0}}%
      \expandafter\def\csname LT2\endcsname{\color[rgb]{0,0,1}}%
      \expandafter\def\csname LT3\endcsname{\color[rgb]{1,0,1}}%
      \expandafter\def\csname LT4\endcsname{\color[rgb]{0,1,1}}%
      \expandafter\def\csname LT5\endcsname{\color[rgb]{1,1,0}}%
      \expandafter\def\csname LT6\endcsname{\color[rgb]{0,0,0}}%
      \expandafter\def\csname LT7\endcsname{\color[rgb]{1,0.3,0}}%
      \expandafter\def\csname LT8\endcsname{\color[rgb]{0.5,0.5,0.5}}%
    \else
      \def\colorrgb#1{\color{black}}%
      \def\colorgray#1{\color[gray]{#1}}%
      \expandafter\def\csname LTw\endcsname{\color{white}}%
      \expandafter\def\csname LTb\endcsname{\color{black}}%
      \expandafter\def\csname LTa\endcsname{\color{black}}%
      \expandafter\def\csname LT0\endcsname{\color{black}}%
      \expandafter\def\csname LT1\endcsname{\color{black}}%
      \expandafter\def\csname LT2\endcsname{\color{black}}%
      \expandafter\def\csname LT3\endcsname{\color{black}}%
      \expandafter\def\csname LT4\endcsname{\color{black}}%
      \expandafter\def\csname LT5\endcsname{\color{black}}%
      \expandafter\def\csname LT6\endcsname{\color{black}}%
      \expandafter\def\csname LT7\endcsname{\color{black}}%
      \expandafter\def\csname LT8\endcsname{\color{black}}%
    \fi
  \fi
    \setlength{\unitlength}{0.0500bp}%
    \ifx\gptboxheight\undefined%
      \newlength{\gptboxheight}%
      \newlength{\gptboxwidth}%
      \newsavebox{\gptboxtext}%
    \fi%
    \setlength{\fboxrule}{0.5pt}%
    \setlength{\fboxsep}{1pt}%
    \definecolor{tbcol}{rgb}{1,1,1}%
\begin{picture}(3960.00,560.00)%
    \gplgaddtomacro\gplbacktext{%
    }%
    \gplgaddtomacro\gplfronttext{%
      \csname LTb\endcsname
      \put(1541,254){\makebox(0,0)[r]{\strut{}M-SLIDER}}%
      \csname LTb\endcsname
      \put(3267,254){\makebox(0,0)[r]{\strut{}SEQ-SLIDER}}%
    }%
    \gplbacktext
    \put(0,0){\includegraphics[width={198.00bp},height={28.00bp}]{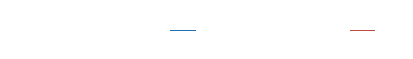}}%
    \gplfronttext
  \end{picture}%
\endgroup

%% file: density.ecoli.l8d5.h8.dedup.tex
\begingroup
  \makeatletter
  \providecommand\color[2][]{%
    \GenericError{(gnuplot) \space\space\space\@spaces}{%
      Package color not loaded in conjunction with
      terminal option `colourtext'%
    }{See the gnuplot documentation for explanation.%
    }{Either use 'blacktext' in gnuplot or load the package
      color.sty in LaTeX.}%
    \renewcommand\color[2][]{}%
  }%
  \providecommand\includegraphics[2][]{%
    \GenericError{(gnuplot) \space\space\space\@spaces}{%
      Package graphicx or graphics not loaded%
    }{See the gnuplot documentation for explanation.%
    }{The gnuplot epslatex terminal needs graphicx.sty or graphics.sty.}%
    \renewcommand\includegraphics[2][]{}%
  }%
  \providecommand\rotatebox[2]{#2}%
  \@ifundefined{ifGPcolor}{%
    \newif\ifGPcolor
    \GPcolortrue
  }{}%
  \@ifundefined{ifGPblacktext}{%
    \newif\ifGPblacktext
    \GPblacktexttrue
  }{}%
  \let\gplgaddtomacro\g@addto@macro
  \gdef\gplbacktext{}%
  \gdef\gplfronttext{}%
  \makeatother
  \ifGPblacktext
    \def\colorrgb#1{}%
    \def\colorgray#1{}%
  \else
    \ifGPcolor
      \def\colorrgb#1{\color[rgb]{#1}}%
      \def\colorgray#1{\color[gray]{#1}}%
      \expandafter\def\csname LTw\endcsname{\color{white}}%
      \expandafter\def\csname LTb\endcsname{\color{black}}%
      \expandafter\def\csname LTa\endcsname{\color{black}}%
      \expandafter\def\csname LT0\endcsname{\color[rgb]{1,0,0}}%
      \expandafter\def\csname LT1\endcsname{\color[rgb]{0,1,0}}%
      \expandafter\def\csname LT2\endcsname{\color[rgb]{0,0,1}}%
      \expandafter\def\csname LT3\endcsname{\color[rgb]{1,0,1}}%
      \expandafter\def\csname LT4\endcsname{\color[rgb]{0,1,1}}%
      \expandafter\def\csname LT5\endcsname{\color[rgb]{1,1,0}}%
      \expandafter\def\csname LT6\endcsname{\color[rgb]{0,0,0}}%
      \expandafter\def\csname LT7\endcsname{\color[rgb]{1,0.3,0}}%
      \expandafter\def\csname LT8\endcsname{\color[rgb]{0.5,0.5,0.5}}%
    \else
      \def\colorrgb#1{\color{black}}%
      \def\colorgray#1{\color[gray]{#1}}%
      \expandafter\def\csname LTw\endcsname{\color{white}}%
      \expandafter\def\csname LTb\endcsname{\color{black}}%
      \expandafter\def\csname LTa\endcsname{\color{black}}%
      \expandafter\def\csname LT0\endcsname{\color{black}}%
      \expandafter\def\csname LT1\endcsname{\color{black}}%
      \expandafter\def\csname LT2\endcsname{\color{black}}%
      \expandafter\def\csname LT3\endcsname{\color{black}}%
      \expandafter\def\csname LT4\endcsname{\color{black}}%
      \expandafter\def\csname LT5\endcsname{\color{black}}%
      \expandafter\def\csname LT6\endcsname{\color{black}}%
      \expandafter\def\csname LT7\endcsname{\color{black}}%
      \expandafter\def\csname LT8\endcsname{\color{black}}%
    \fi
  \fi
    \setlength{\unitlength}{0.0500bp}%
    \ifx\gptboxheight\undefined%
      \newlength{\gptboxheight}%
      \newlength{\gptboxwidth}%
      \newsavebox{\gptboxtext}%
    \fi%
    \setlength{\fboxrule}{0.5pt}%
    \setlength{\fboxsep}{1pt}%
    \definecolor{tbcol}{rgb}{1,1,1}%
\begin{picture}(1980.00,1700.00)%
    \gplgaddtomacro\gplbacktext{%
      \csname LTb\endcsname
      \put(566,627){\makebox(0,0)[r]{\strut{}0}}%
      \csname LTb\endcsname
      \put(566,1483){\makebox(0,0)[r]{\strut{}8.1}}%
      \csname LTb\endcsname
      \put(649,431){\makebox(0,0){\strut{}0}}%
      \csname LTb\endcsname
      \put(1181,431){\makebox(0,0){\strut{}0.5}}%
      \csname LTb\endcsname
      \put(1712,431){\makebox(0,0){\strut{}1}}%
    }%
    \gplgaddtomacro\gplfronttext{%
      \csname LTb\endcsname
      \put(159,1055){\rotatebox{-270.00}{\makebox(0,0){\strut{}\normalsize Density}}}%
      \csname LTb\endcsname
      \put(1181,137){\makebox(0,0){\strut{}\normalsize Similarity}}%
    }%
    \gplbacktext
    \put(0,0){\includegraphics[width={99.00bp},height={85.00bp}]{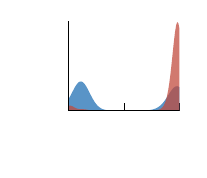}}%
    \gplfronttext
  \end{picture}%
\endgroup

%% file: density.yeast.l8d5.h8.dedup.tex
\begingroup
  \makeatletter
  \providecommand\color[2][]{%
    \GenericError{(gnuplot) \space\space\space\@spaces}{%
      Package color not loaded in conjunction with
      terminal option `colourtext'%
    }{See the gnuplot documentation for explanation.%
    }{Either use 'blacktext' in gnuplot or load the package
      color.sty in LaTeX.}%
    \renewcommand\color[2][]{}%
  }%
  \providecommand\includegraphics[2][]{%
    \GenericError{(gnuplot) \space\space\space\@spaces}{%
      Package graphicx or graphics not loaded%
    }{See the gnuplot documentation for explanation.%
    }{The gnuplot epslatex terminal needs graphicx.sty or graphics.sty.}%
    \renewcommand\includegraphics[2][]{}%
  }%
  \providecommand\rotatebox[2]{#2}%
  \@ifundefined{ifGPcolor}{%
    \newif\ifGPcolor
    \GPcolortrue
  }{}%
  \@ifundefined{ifGPblacktext}{%
    \newif\ifGPblacktext
    \GPblacktexttrue
  }{}%
  \let\gplgaddtomacro\g@addto@macro
  \gdef\gplbacktext{}%
  \gdef\gplfronttext{}%
  \makeatother
  \ifGPblacktext
    \def\colorrgb#1{}%
    \def\colorgray#1{}%
  \else
    \ifGPcolor
      \def\colorrgb#1{\color[rgb]{#1}}%
      \def\colorgray#1{\color[gray]{#1}}%
      \expandafter\def\csname LTw\endcsname{\color{white}}%
      \expandafter\def\csname LTb\endcsname{\color{black}}%
      \expandafter\def\csname LTa\endcsname{\color{black}}%
      \expandafter\def\csname LT0\endcsname{\color[rgb]{1,0,0}}%
      \expandafter\def\csname LT1\endcsname{\color[rgb]{0,1,0}}%
      \expandafter\def\csname LT2\endcsname{\color[rgb]{0,0,1}}%
      \expandafter\def\csname LT3\endcsname{\color[rgb]{1,0,1}}%
      \expandafter\def\csname LT4\endcsname{\color[rgb]{0,1,1}}%
      \expandafter\def\csname LT5\endcsname{\color[rgb]{1,1,0}}%
      \expandafter\def\csname LT6\endcsname{\color[rgb]{0,0,0}}%
      \expandafter\def\csname LT7\endcsname{\color[rgb]{1,0.3,0}}%
      \expandafter\def\csname LT8\endcsname{\color[rgb]{0.5,0.5,0.5}}%
    \else
      \def\colorrgb#1{\color{black}}%
      \def\colorgray#1{\color[gray]{#1}}%
      \expandafter\def\csname LTw\endcsname{\color{white}}%
      \expandafter\def\csname LTb\endcsname{\color{black}}%
      \expandafter\def\csname LTa\endcsname{\color{black}}%
      \expandafter\def\csname LT0\endcsname{\color{black}}%
      \expandafter\def\csname LT1\endcsname{\color{black}}%
      \expandafter\def\csname LT2\endcsname{\color{black}}%
      \expandafter\def\csname LT3\endcsname{\color{black}}%
      \expandafter\def\csname LT4\endcsname{\color{black}}%
      \expandafter\def\csname LT5\endcsname{\color{black}}%
      \expandafter\def\csname LT6\endcsname{\color{black}}%
      \expandafter\def\csname LT7\endcsname{\color{black}}%
      \expandafter\def\csname LT8\endcsname{\color{black}}%
    \fi
  \fi
    \setlength{\unitlength}{0.0500bp}%
    \ifx\gptboxheight\undefined%
      \newlength{\gptboxheight}%
      \newlength{\gptboxwidth}%
      \newsavebox{\gptboxtext}%
    \fi%
    \setlength{\fboxrule}{0.5pt}%
    \setlength{\fboxsep}{1pt}%
    \definecolor{tbcol}{rgb}{1,1,1}%
\begin{picture}(1980.00,1700.00)%
    \gplgaddtomacro\gplbacktext{%
      \csname LTb\endcsname
      \put(566,627){\makebox(0,0)[r]{\strut{}0}}%
      \csname LTb\endcsname
      \put(566,1483){\makebox(0,0)[r]{\strut{}5.3}}%
      \csname LTb\endcsname
      \put(649,431){\makebox(0,0){\strut{}0}}%
      \csname LTb\endcsname
      \put(1181,431){\makebox(0,0){\strut{}0.5}}%
      \csname LTb\endcsname
      \put(1712,431){\makebox(0,0){\strut{}1}}%
    }%
    \gplgaddtomacro\gplfronttext{%
      \csname LTb\endcsname
      \put(159,1055){\rotatebox{-270.00}{\makebox(0,0){\strut{}\normalsize Density}}}%
      \csname LTb\endcsname
      \put(1181,137){\makebox(0,0){\strut{}\normalsize Similarity}}%
    }%
    \gplbacktext
    \put(0,0){\includegraphics[width={99.00bp},height={85.00bp}]{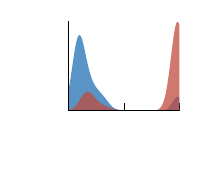}}%
    \gplfronttext
  \end{picture}%
\endgroup

%% file: density.fly.l8d5.h8.dedup.tex
\begingroup
  \makeatletter
  \providecommand\color[2][]{%
    \GenericError{(gnuplot) \space\space\space\@spaces}{%
      Package color not loaded in conjunction with
      terminal option `colourtext'%
    }{See the gnuplot documentation for explanation.%
    }{Either use 'blacktext' in gnuplot or load the package
      color.sty in LaTeX.}%
    \renewcommand\color[2][]{}%
  }%
  \providecommand\includegraphics[2][]{%
    \GenericError{(gnuplot) \space\space\space\@spaces}{%
      Package graphicx or graphics not loaded%
    }{See the gnuplot documentation for explanation.%
    }{The gnuplot epslatex terminal needs graphicx.sty or graphics.sty.}%
    \renewcommand\includegraphics[2][]{}%
  }%
  \providecommand\rotatebox[2]{#2}%
  \@ifundefined{ifGPcolor}{%
    \newif\ifGPcolor
    \GPcolortrue
  }{}%
  \@ifundefined{ifGPblacktext}{%
    \newif\ifGPblacktext
    \GPblacktexttrue
  }{}%
  \let\gplgaddtomacro\g@addto@macro
  \gdef\gplbacktext{}%
  \gdef\gplfronttext{}%
  \makeatother
  \ifGPblacktext
    \def\colorrgb#1{}%
    \def\colorgray#1{}%
  \else
    \ifGPcolor
      \def\colorrgb#1{\color[rgb]{#1}}%
      \def\colorgray#1{\color[gray]{#1}}%
      \expandafter\def\csname LTw\endcsname{\color{white}}%
      \expandafter\def\csname LTb\endcsname{\color{black}}%
      \expandafter\def\csname LTa\endcsname{\color{black}}%
      \expandafter\def\csname LT0\endcsname{\color[rgb]{1,0,0}}%
      \expandafter\def\csname LT1\endcsname{\color[rgb]{0,1,0}}%
      \expandafter\def\csname LT2\endcsname{\color[rgb]{0,0,1}}%
      \expandafter\def\csname LT3\endcsname{\color[rgb]{1,0,1}}%
      \expandafter\def\csname LT4\endcsname{\color[rgb]{0,1,1}}%
      \expandafter\def\csname LT5\endcsname{\color[rgb]{1,1,0}}%
      \expandafter\def\csname LT6\endcsname{\color[rgb]{0,0,0}}%
      \expandafter\def\csname LT7\endcsname{\color[rgb]{1,0.3,0}}%
      \expandafter\def\csname LT8\endcsname{\color[rgb]{0.5,0.5,0.5}}%
    \else
      \def\colorrgb#1{\color{black}}%
      \def\colorgray#1{\color[gray]{#1}}%
      \expandafter\def\csname LTw\endcsname{\color{white}}%
      \expandafter\def\csname LTb\endcsname{\color{black}}%
      \expandafter\def\csname LTa\endcsname{\color{black}}%
      \expandafter\def\csname LT0\endcsname{\color{black}}%
      \expandafter\def\csname LT1\endcsname{\color{black}}%
      \expandafter\def\csname LT2\endcsname{\color{black}}%
      \expandafter\def\csname LT3\endcsname{\color{black}}%
      \expandafter\def\csname LT4\endcsname{\color{black}}%
      \expandafter\def\csname LT5\endcsname{\color{black}}%
      \expandafter\def\csname LT6\endcsname{\color{black}}%
      \expandafter\def\csname LT7\endcsname{\color{black}}%
      \expandafter\def\csname LT8\endcsname{\color{black}}%
    \fi
  \fi
    \setlength{\unitlength}{0.0500bp}%
    \ifx\gptboxheight\undefined%
      \newlength{\gptboxheight}%
      \newlength{\gptboxwidth}%
      \newsavebox{\gptboxtext}%
    \fi%
    \setlength{\fboxrule}{0.5pt}%
    \setlength{\fboxsep}{1pt}%
    \definecolor{tbcol}{rgb}{1,1,1}%
\begin{picture}(1980.00,1700.00)%
    \gplgaddtomacro\gplbacktext{%
      \csname LTb\endcsname
      \put(566,627){\makebox(0,0)[r]{\strut{}0}}%
      \csname LTb\endcsname
      \put(566,1483){\makebox(0,0)[r]{\strut{}4.2}}%
      \csname LTb\endcsname
      \put(649,431){\makebox(0,0){\strut{}0}}%
      \csname LTb\endcsname
      \put(1181,431){\makebox(0,0){\strut{}0.5}}%
      \csname LTb\endcsname
      \put(1712,431){\makebox(0,0){\strut{}1}}%
    }%
    \gplgaddtomacro\gplfronttext{%
      \csname LTb\endcsname
      \put(159,1055){\rotatebox{-270.00}{\makebox(0,0){\strut{}\normalsize Density}}}%
      \csname LTb\endcsname
      \put(1181,137){\makebox(0,0){\strut{}\normalsize Similarity}}%
    }%
    \gplbacktext
    \put(0,0){\includegraphics[width={99.00bp},height={85.00bp}]{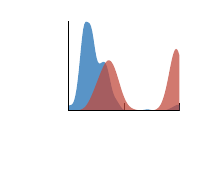}}%
    \gplfronttext
  \end{picture}%
\endgroup

%% file: density.plant.l8d5.h8.dedup.tex
\begingroup
  \makeatletter
  \providecommand\color[2][]{%
    \GenericError{(gnuplot) \space\space\space\@spaces}{%
      Package color not loaded in conjunction with
      terminal option `colourtext'%
    }{See the gnuplot documentation for explanation.%
    }{Either use 'blacktext' in gnuplot or load the package
      color.sty in LaTeX.}%
    \renewcommand\color[2][]{}%
  }%
  \providecommand\includegraphics[2][]{%
    \GenericError{(gnuplot) \space\space\space\@spaces}{%
      Package graphicx or graphics not loaded%
    }{See the gnuplot documentation for explanation.%
    }{The gnuplot epslatex terminal needs graphicx.sty or graphics.sty.}%
    \renewcommand\includegraphics[2][]{}%
  }%
  \providecommand\rotatebox[2]{#2}%
  \@ifundefined{ifGPcolor}{%
    \newif\ifGPcolor
    \GPcolortrue
  }{}%
  \@ifundefined{ifGPblacktext}{%
    \newif\ifGPblacktext
    \GPblacktexttrue
  }{}%
  \let\gplgaddtomacro\g@addto@macro
  \gdef\gplbacktext{}%
  \gdef\gplfronttext{}%
  \makeatother
  \ifGPblacktext
    \def\colorrgb#1{}%
    \def\colorgray#1{}%
  \else
    \ifGPcolor
      \def\colorrgb#1{\color[rgb]{#1}}%
      \def\colorgray#1{\color[gray]{#1}}%
      \expandafter\def\csname LTw\endcsname{\color{white}}%
      \expandafter\def\csname LTb\endcsname{\color{black}}%
      \expandafter\def\csname LTa\endcsname{\color{black}}%
      \expandafter\def\csname LT0\endcsname{\color[rgb]{1,0,0}}%
      \expandafter\def\csname LT1\endcsname{\color[rgb]{0,1,0}}%
      \expandafter\def\csname LT2\endcsname{\color[rgb]{0,0,1}}%
      \expandafter\def\csname LT3\endcsname{\color[rgb]{1,0,1}}%
      \expandafter\def\csname LT4\endcsname{\color[rgb]{0,1,1}}%
      \expandafter\def\csname LT5\endcsname{\color[rgb]{1,1,0}}%
      \expandafter\def\csname LT6\endcsname{\color[rgb]{0,0,0}}%
      \expandafter\def\csname LT7\endcsname{\color[rgb]{1,0.3,0}}%
      \expandafter\def\csname LT8\endcsname{\color[rgb]{0.5,0.5,0.5}}%
    \else
      \def\colorrgb#1{\color{black}}%
      \def\colorgray#1{\color[gray]{#1}}%
      \expandafter\def\csname LTw\endcsname{\color{white}}%
      \expandafter\def\csname LTb\endcsname{\color{black}}%
      \expandafter\def\csname LTa\endcsname{\color{black}}%
      \expandafter\def\csname LT0\endcsname{\color{black}}%
      \expandafter\def\csname LT1\endcsname{\color{black}}%
      \expandafter\def\csname LT2\endcsname{\color{black}}%
      \expandafter\def\csname LT3\endcsname{\color{black}}%
      \expandafter\def\csname LT4\endcsname{\color{black}}%
      \expandafter\def\csname LT5\endcsname{\color{black}}%
      \expandafter\def\csname LT6\endcsname{\color{black}}%
      \expandafter\def\csname LT7\endcsname{\color{black}}%
      \expandafter\def\csname LT8\endcsname{\color{black}}%
    \fi
  \fi
    \setlength{\unitlength}{0.0500bp}%
    \ifx\gptboxheight\undefined%
      \newlength{\gptboxheight}%
      \newlength{\gptboxwidth}%
      \newsavebox{\gptboxtext}%
    \fi%
    \setlength{\fboxrule}{0.5pt}%
    \setlength{\fboxsep}{1pt}%
    \definecolor{tbcol}{rgb}{1,1,1}%
\begin{picture}(1980.00,1700.00)%
    \gplgaddtomacro\gplbacktext{%
      \csname LTb\endcsname
      \put(402,627){\makebox(0,0)[r]{\strut{}0}}%
      \csname LTb\endcsname
      \put(402,1483){\makebox(0,0)[r]{\strut{}4}}%
      \csname LTb\endcsname
      \put(484,431){\makebox(0,0){\strut{}0}}%
      \csname LTb\endcsname
      \put(1098,431){\makebox(0,0){\strut{}0.5}}%
      \csname LTb\endcsname
      \put(1712,431){\makebox(0,0){\strut{}1}}%
    }%
    \gplgaddtomacro\gplfronttext{%
      \csname LTb\endcsname
      \put(159,1055){\rotatebox{-270.00}{\makebox(0,0){\strut{}\normalsize Density}}}%
      \csname LTb\endcsname
      \put(1098,137){\makebox(0,0){\strut{}\normalsize Similarity}}%
    }%
    \gplbacktext
    \put(0,0){\includegraphics[width={99.00bp},height={85.00bp}]{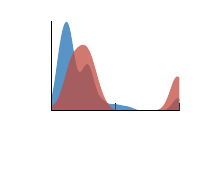}}%
    \gplfronttext
  \end{picture}%
\endgroup

%% file: density.human.l8d5.h8.dedup.tex
\begingroup
  \makeatletter
  \providecommand\color[2][]{%
    \GenericError{(gnuplot) \space\space\space\@spaces}{%
      Package color not loaded in conjunction with
      terminal option `colourtext'%
    }{See the gnuplot documentation for explanation.%
    }{Either use 'blacktext' in gnuplot or load the package
      color.sty in LaTeX.}%
    \renewcommand\color[2][]{}%
  }%
  \providecommand\includegraphics[2][]{%
    \GenericError{(gnuplot) \space\space\space\@spaces}{%
      Package graphicx or graphics not loaded%
    }{See the gnuplot documentation for explanation.%
    }{The gnuplot epslatex terminal needs graphicx.sty or graphics.sty.}%
    \renewcommand\includegraphics[2][]{}%
  }%
  \providecommand\rotatebox[2]{#2}%
  \@ifundefined{ifGPcolor}{%
    \newif\ifGPcolor
    \GPcolortrue
  }{}%
  \@ifundefined{ifGPblacktext}{%
    \newif\ifGPblacktext
    \GPblacktexttrue
  }{}%
  \let\gplgaddtomacro\g@addto@macro
  \gdef\gplbacktext{}%
  \gdef\gplfronttext{}%
  \makeatother
  \ifGPblacktext
    \def\colorrgb#1{}%
    \def\colorgray#1{}%
  \else
    \ifGPcolor
      \def\colorrgb#1{\color[rgb]{#1}}%
      \def\colorgray#1{\color[gray]{#1}}%
      \expandafter\def\csname LTw\endcsname{\color{white}}%
      \expandafter\def\csname LTb\endcsname{\color{black}}%
      \expandafter\def\csname LTa\endcsname{\color{black}}%
      \expandafter\def\csname LT0\endcsname{\color[rgb]{1,0,0}}%
      \expandafter\def\csname LT1\endcsname{\color[rgb]{0,1,0}}%
      \expandafter\def\csname LT2\endcsname{\color[rgb]{0,0,1}}%
      \expandafter\def\csname LT3\endcsname{\color[rgb]{1,0,1}}%
      \expandafter\def\csname LT4\endcsname{\color[rgb]{0,1,1}}%
      \expandafter\def\csname LT5\endcsname{\color[rgb]{1,1,0}}%
      \expandafter\def\csname LT6\endcsname{\color[rgb]{0,0,0}}%
      \expandafter\def\csname LT7\endcsname{\color[rgb]{1,0.3,0}}%
      \expandafter\def\csname LT8\endcsname{\color[rgb]{0.5,0.5,0.5}}%
    \else
      \def\colorrgb#1{\color{black}}%
      \def\colorgray#1{\color[gray]{#1}}%
      \expandafter\def\csname LTw\endcsname{\color{white}}%
      \expandafter\def\csname LTb\endcsname{\color{black}}%
      \expandafter\def\csname LTa\endcsname{\color{black}}%
      \expandafter\def\csname LT0\endcsname{\color{black}}%
      \expandafter\def\csname LT1\endcsname{\color{black}}%
      \expandafter\def\csname LT2\endcsname{\color{black}}%
      \expandafter\def\csname LT3\endcsname{\color{black}}%
      \expandafter\def\csname LT4\endcsname{\color{black}}%
      \expandafter\def\csname LT5\endcsname{\color{black}}%
      \expandafter\def\csname LT6\endcsname{\color{black}}%
      \expandafter\def\csname LT7\endcsname{\color{black}}%
      \expandafter\def\csname LT8\endcsname{\color{black}}%
    \fi
  \fi
    \setlength{\unitlength}{0.0500bp}%
    \ifx\gptboxheight\undefined%
      \newlength{\gptboxheight}%
      \newlength{\gptboxwidth}%
      \newsavebox{\gptboxtext}%
    \fi%
    \setlength{\fboxrule}{0.5pt}%
    \setlength{\fboxsep}{1pt}%
    \definecolor{tbcol}{rgb}{1,1,1}%
\begin{picture}(1980.00,1700.00)%
    \gplgaddtomacro\gplbacktext{%
      \csname LTb\endcsname
      \put(566,627){\makebox(0,0)[r]{\strut{}0}}%
      \csname LTb\endcsname
      \put(566,1483){\makebox(0,0)[r]{\strut{}5.7}}%
      \csname LTb\endcsname
      \put(649,431){\makebox(0,0){\strut{}0}}%
      \csname LTb\endcsname
      \put(1181,431){\makebox(0,0){\strut{}0.5}}%
      \csname LTb\endcsname
      \put(1712,431){\makebox(0,0){\strut{}1}}%
    }%
    \gplgaddtomacro\gplfronttext{%
      \csname LTb\endcsname
      \put(159,1055){\rotatebox{-270.00}{\makebox(0,0){\strut{}\normalsize Density}}}%
      \csname LTb\endcsname
      \put(1181,137){\makebox(0,0){\strut{}\normalsize Similarity}}%
    }%
    \gplbacktext
    \put(0,0){\includegraphics[width={99.00bp},height={85.00bp}]{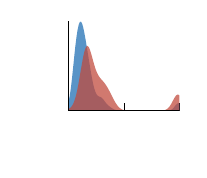}}%
    \gplfronttext
  \end{picture}%
\endgroup

%% file: density.yeast_aj.l8d5.h8.dedup.tex
\begingroup
  \makeatletter
  \providecommand\color[2][]{%
    \GenericError{(gnuplot) \space\space\space\@spaces}{%
      Package color not loaded in conjunction with
      terminal option `colourtext'%
    }{See the gnuplot documentation for explanation.%
    }{Either use 'blacktext' in gnuplot or load the package
      color.sty in LaTeX.}%
    \renewcommand\color[2][]{}%
  }%
  \providecommand\includegraphics[2][]{%
    \GenericError{(gnuplot) \space\space\space\@spaces}{%
      Package graphicx or graphics not loaded%
    }{See the gnuplot documentation for explanation.%
    }{The gnuplot epslatex terminal needs graphicx.sty or graphics.sty.}%
    \renewcommand\includegraphics[2][]{}%
  }%
  \providecommand\rotatebox[2]{#2}%
  \@ifundefined{ifGPcolor}{%
    \newif\ifGPcolor
    \GPcolortrue
  }{}%
  \@ifundefined{ifGPblacktext}{%
    \newif\ifGPblacktext
    \GPblacktexttrue
  }{}%
  \let\gplgaddtomacro\g@addto@macro
  \gdef\gplbacktext{}%
  \gdef\gplfronttext{}%
  \makeatother
  \ifGPblacktext
    \def\colorrgb#1{}%
    \def\colorgray#1{}%
  \else
    \ifGPcolor
      \def\colorrgb#1{\color[rgb]{#1}}%
      \def\colorgray#1{\color[gray]{#1}}%
      \expandafter\def\csname LTw\endcsname{\color{white}}%
      \expandafter\def\csname LTb\endcsname{\color{black}}%
      \expandafter\def\csname LTa\endcsname{\color{black}}%
      \expandafter\def\csname LT0\endcsname{\color[rgb]{1,0,0}}%
      \expandafter\def\csname LT1\endcsname{\color[rgb]{0,1,0}}%
      \expandafter\def\csname LT2\endcsname{\color[rgb]{0,0,1}}%
      \expandafter\def\csname LT3\endcsname{\color[rgb]{1,0,1}}%
      \expandafter\def\csname LT4\endcsname{\color[rgb]{0,1,1}}%
      \expandafter\def\csname LT5\endcsname{\color[rgb]{1,1,0}}%
      \expandafter\def\csname LT6\endcsname{\color[rgb]{0,0,0}}%
      \expandafter\def\csname LT7\endcsname{\color[rgb]{1,0.3,0}}%
      \expandafter\def\csname LT8\endcsname{\color[rgb]{0.5,0.5,0.5}}%
    \else
      \def\colorrgb#1{\color{black}}%
      \def\colorgray#1{\color[gray]{#1}}%
      \expandafter\def\csname LTw\endcsname{\color{white}}%
      \expandafter\def\csname LTb\endcsname{\color{black}}%
      \expandafter\def\csname LTa\endcsname{\color{black}}%
      \expandafter\def\csname LT0\endcsname{\color{black}}%
      \expandafter\def\csname LT1\endcsname{\color{black}}%
      \expandafter\def\csname LT2\endcsname{\color{black}}%
      \expandafter\def\csname LT3\endcsname{\color{black}}%
      \expandafter\def\csname LT4\endcsname{\color{black}}%
      \expandafter\def\csname LT5\endcsname{\color{black}}%
      \expandafter\def\csname LT6\endcsname{\color{black}}%
      \expandafter\def\csname LT7\endcsname{\color{black}}%
      \expandafter\def\csname LT8\endcsname{\color{black}}%
    \fi
  \fi
    \setlength{\unitlength}{0.0500bp}%
    \ifx\gptboxheight\undefined%
      \newlength{\gptboxheight}%
      \newlength{\gptboxwidth}%
      \newsavebox{\gptboxtext}%
    \fi%
    \setlength{\fboxrule}{0.5pt}%
    \setlength{\fboxsep}{1pt}%
    \definecolor{tbcol}{rgb}{1,1,1}%
\begin{picture}(1980.00,1700.00)%
    \gplgaddtomacro\gplbacktext{%
      \csname LTb\endcsname
      \put(566,627){\makebox(0,0)[r]{\strut{}0}}%
      \csname LTb\endcsname
      \put(566,1483){\makebox(0,0)[r]{\strut{}9.1}}%
      \csname LTb\endcsname
      \put(649,431){\makebox(0,0){\strut{}0}}%
      \csname LTb\endcsname
      \put(1181,431){\makebox(0,0){\strut{}0.5}}%
      \csname LTb\endcsname
      \put(1712,431){\makebox(0,0){\strut{}1}}%
    }%
    \gplgaddtomacro\gplfronttext{%
      \csname LTb\endcsname
      \put(159,1055){\rotatebox{-270.00}{\makebox(0,0){\strut{}\normalsize Density}}}%
      \csname LTb\endcsname
      \put(1181,137){\makebox(0,0){\strut{}\normalsize Similarity}}%
    }%
    \gplbacktext
    \put(0,0){\includegraphics[width={99.00bp},height={85.00bp}]{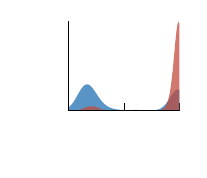}}%
    \gplfronttext
  \end{picture}%
\endgroup

%% file: density.human_aj.l8d5.h8.dedup.tex
\begingroup
  \makeatletter
  \providecommand\color[2][]{%
    \GenericError{(gnuplot) \space\space\space\@spaces}{%
      Package color not loaded in conjunction with
      terminal option `colourtext'%
    }{See the gnuplot documentation for explanation.%
    }{Either use 'blacktext' in gnuplot or load the package
      color.sty in LaTeX.}%
    \renewcommand\color[2][]{}%
  }%
  \providecommand\includegraphics[2][]{%
    \GenericError{(gnuplot) \space\space\space\@spaces}{%
      Package graphicx or graphics not loaded%
    }{See the gnuplot documentation for explanation.%
    }{The gnuplot epslatex terminal needs graphicx.sty or graphics.sty.}%
    \renewcommand\includegraphics[2][]{}%
  }%
  \providecommand\rotatebox[2]{#2}%
  \@ifundefined{ifGPcolor}{%
    \newif\ifGPcolor
    \GPcolortrue
  }{}%
  \@ifundefined{ifGPblacktext}{%
    \newif\ifGPblacktext
    \GPblacktexttrue
  }{}%
  \let\gplgaddtomacro\g@addto@macro
  \gdef\gplbacktext{}%
  \gdef\gplfronttext{}%
  \makeatother
  \ifGPblacktext
    \def\colorrgb#1{}%
    \def\colorgray#1{}%
  \else
    \ifGPcolor
      \def\colorrgb#1{\color[rgb]{#1}}%
      \def\colorgray#1{\color[gray]{#1}}%
      \expandafter\def\csname LTw\endcsname{\color{white}}%
      \expandafter\def\csname LTb\endcsname{\color{black}}%
      \expandafter\def\csname LTa\endcsname{\color{black}}%
      \expandafter\def\csname LT0\endcsname{\color[rgb]{1,0,0}}%
      \expandafter\def\csname LT1\endcsname{\color[rgb]{0,1,0}}%
      \expandafter\def\csname LT2\endcsname{\color[rgb]{0,0,1}}%
      \expandafter\def\csname LT3\endcsname{\color[rgb]{1,0,1}}%
      \expandafter\def\csname LT4\endcsname{\color[rgb]{0,1,1}}%
      \expandafter\def\csname LT5\endcsname{\color[rgb]{1,1,0}}%
      \expandafter\def\csname LT6\endcsname{\color[rgb]{0,0,0}}%
      \expandafter\def\csname LT7\endcsname{\color[rgb]{1,0.3,0}}%
      \expandafter\def\csname LT8\endcsname{\color[rgb]{0.5,0.5,0.5}}%
    \else
      \def\colorrgb#1{\color{black}}%
      \def\colorgray#1{\color[gray]{#1}}%
      \expandafter\def\csname LTw\endcsname{\color{white}}%
      \expandafter\def\csname LTb\endcsname{\color{black}}%
      \expandafter\def\csname LTa\endcsname{\color{black}}%
      \expandafter\def\csname LT0\endcsname{\color{black}}%
      \expandafter\def\csname LT1\endcsname{\color{black}}%
      \expandafter\def\csname LT2\endcsname{\color{black}}%
      \expandafter\def\csname LT3\endcsname{\color{black}}%
      \expandafter\def\csname LT4\endcsname{\color{black}}%
      \expandafter\def\csname LT5\endcsname{\color{black}}%
      \expandafter\def\csname LT6\endcsname{\color{black}}%
      \expandafter\def\csname LT7\endcsname{\color{black}}%
      \expandafter\def\csname LT8\endcsname{\color{black}}%
    \fi
  \fi
    \setlength{\unitlength}{0.0500bp}%
    \ifx\gptboxheight\undefined%
      \newlength{\gptboxheight}%
      \newlength{\gptboxwidth}%
      \newsavebox{\gptboxtext}%
    \fi%
    \setlength{\fboxrule}{0.5pt}%
    \setlength{\fboxsep}{1pt}%
    \definecolor{tbcol}{rgb}{1,1,1}%
\begin{picture}(1980.00,1700.00)%
    \gplgaddtomacro\gplbacktext{%
      \csname LTb\endcsname
      \put(566,627){\makebox(0,0)[r]{\strut{}0}}%
      \csname LTb\endcsname
      \put(566,1483){\makebox(0,0)[r]{\strut{}5.1}}%
      \csname LTb\endcsname
      \put(649,431){\makebox(0,0){\strut{}0}}%
      \csname LTb\endcsname
      \put(1181,431){\makebox(0,0){\strut{}0.5}}%
      \csname LTb\endcsname
      \put(1712,431){\makebox(0,0){\strut{}1}}%
    }%
    \gplgaddtomacro\gplfronttext{%
      \csname LTb\endcsname
      \put(159,1055){\rotatebox{-270.00}{\makebox(0,0){\strut{}\normalsize Density}}}%
      \csname LTb\endcsname
      \put(1181,137){\makebox(0,0){\strut{}\normalsize Similarity}}%
    }%
    \gplbacktext
    \put(0,0){\includegraphics[width={99.00bp},height={85.00bp}]{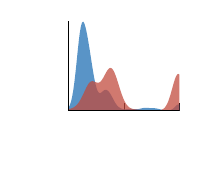}}%
    \gplfronttext
  \end{picture}%
\endgroup

%% file: density.ecoli.l8d3.h8.dedup.tex
\begingroup
  \makeatletter
  \providecommand\color[2][]{%
    \GenericError{(gnuplot) \space\space\space\@spaces}{%
      Package color not loaded in conjunction with
      terminal option `colourtext'%
    }{See the gnuplot documentation for explanation.%
    }{Either use 'blacktext' in gnuplot or load the package
      color.sty in LaTeX.}%
    \renewcommand\color[2][]{}%
  }%
  \providecommand\includegraphics[2][]{%
    \GenericError{(gnuplot) \space\space\space\@spaces}{%
      Package graphicx or graphics not loaded%
    }{See the gnuplot documentation for explanation.%
    }{The gnuplot epslatex terminal needs graphicx.sty or graphics.sty.}%
    \renewcommand\includegraphics[2][]{}%
  }%
  \providecommand\rotatebox[2]{#2}%
  \@ifundefined{ifGPcolor}{%
    \newif\ifGPcolor
    \GPcolortrue
  }{}%
  \@ifundefined{ifGPblacktext}{%
    \newif\ifGPblacktext
    \GPblacktexttrue
  }{}%
  \let\gplgaddtomacro\g@addto@macro
  \gdef\gplbacktext{}%
  \gdef\gplfronttext{}%
  \makeatother
  \ifGPblacktext
    \def\colorrgb#1{}%
    \def\colorgray#1{}%
  \else
    \ifGPcolor
      \def\colorrgb#1{\color[rgb]{#1}}%
      \def\colorgray#1{\color[gray]{#1}}%
      \expandafter\def\csname LTw\endcsname{\color{white}}%
      \expandafter\def\csname LTb\endcsname{\color{black}}%
      \expandafter\def\csname LTa\endcsname{\color{black}}%
      \expandafter\def\csname LT0\endcsname{\color[rgb]{1,0,0}}%
      \expandafter\def\csname LT1\endcsname{\color[rgb]{0,1,0}}%
      \expandafter\def\csname LT2\endcsname{\color[rgb]{0,0,1}}%
      \expandafter\def\csname LT3\endcsname{\color[rgb]{1,0,1}}%
      \expandafter\def\csname LT4\endcsname{\color[rgb]{0,1,1}}%
      \expandafter\def\csname LT5\endcsname{\color[rgb]{1,1,0}}%
      \expandafter\def\csname LT6\endcsname{\color[rgb]{0,0,0}}%
      \expandafter\def\csname LT7\endcsname{\color[rgb]{1,0.3,0}}%
      \expandafter\def\csname LT8\endcsname{\color[rgb]{0.5,0.5,0.5}}%
    \else
      \def\colorrgb#1{\color{black}}%
      \def\colorgray#1{\color[gray]{#1}}%
      \expandafter\def\csname LTw\endcsname{\color{white}}%
      \expandafter\def\csname LTb\endcsname{\color{black}}%
      \expandafter\def\csname LTa\endcsname{\color{black}}%
      \expandafter\def\csname LT0\endcsname{\color{black}}%
      \expandafter\def\csname LT1\endcsname{\color{black}}%
      \expandafter\def\csname LT2\endcsname{\color{black}}%
      \expandafter\def\csname LT3\endcsname{\color{black}}%
      \expandafter\def\csname LT4\endcsname{\color{black}}%
      \expandafter\def\csname LT5\endcsname{\color{black}}%
      \expandafter\def\csname LT6\endcsname{\color{black}}%
      \expandafter\def\csname LT7\endcsname{\color{black}}%
      \expandafter\def\csname LT8\endcsname{\color{black}}%
    \fi
  \fi
    \setlength{\unitlength}{0.0500bp}%
    \ifx\gptboxheight\undefined%
      \newlength{\gptboxheight}%
      \newlength{\gptboxwidth}%
      \newsavebox{\gptboxtext}%
    \fi%
    \setlength{\fboxrule}{0.5pt}%
    \setlength{\fboxsep}{1pt}%
    \definecolor{tbcol}{rgb}{1,1,1}%
\begin{picture}(1980.00,1700.00)%
    \gplgaddtomacro\gplbacktext{%
      \csname LTb\endcsname
      \put(402,627){\makebox(0,0)[r]{\strut{}0}}%
      \csname LTb\endcsname
      \put(402,1483){\makebox(0,0)[r]{\strut{}7}}%
      \csname LTb\endcsname
      \put(484,431){\makebox(0,0){\strut{}0}}%
      \csname LTb\endcsname
      \put(1098,431){\makebox(0,0){\strut{}0.5}}%
      \csname LTb\endcsname
      \put(1712,431){\makebox(0,0){\strut{}1}}%
    }%
    \gplgaddtomacro\gplfronttext{%
      \csname LTb\endcsname
      \put(159,1055){\rotatebox{-270.00}{\makebox(0,0){\strut{}\normalsize Density}}}%
      \csname LTb\endcsname
      \put(1098,137){\makebox(0,0){\strut{}\normalsize Similarity}}%
    }%
    \gplbacktext
    \put(0,0){\includegraphics[width={99.00bp},height={85.00bp}]{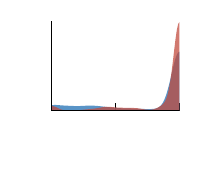}}%
    \gplfronttext
  \end{picture}%
\endgroup

%% file: density.plant.l8d3.h8.dedup.tex
\begingroup
  \makeatletter
  \providecommand\color[2][]{%
    \GenericError{(gnuplot) \space\space\space\@spaces}{%
      Package color not loaded in conjunction with
      terminal option `colourtext'%
    }{See the gnuplot documentation for explanation.%
    }{Either use 'blacktext' in gnuplot or load the package
      color.sty in LaTeX.}%
    \renewcommand\color[2][]{}%
  }%
  \providecommand\includegraphics[2][]{%
    \GenericError{(gnuplot) \space\space\space\@spaces}{%
      Package graphicx or graphics not loaded%
    }{See the gnuplot documentation for explanation.%
    }{The gnuplot epslatex terminal needs graphicx.sty or graphics.sty.}%
    \renewcommand\includegraphics[2][]{}%
  }%
  \providecommand\rotatebox[2]{#2}%
  \@ifundefined{ifGPcolor}{%
    \newif\ifGPcolor
    \GPcolortrue
  }{}%
  \@ifundefined{ifGPblacktext}{%
    \newif\ifGPblacktext
    \GPblacktexttrue
  }{}%
  \let\gplgaddtomacro\g@addto@macro
  \gdef\gplbacktext{}%
  \gdef\gplfronttext{}%
  \makeatother
  \ifGPblacktext
    \def\colorrgb#1{}%
    \def\colorgray#1{}%
  \else
    \ifGPcolor
      \def\colorrgb#1{\color[rgb]{#1}}%
      \def\colorgray#1{\color[gray]{#1}}%
      \expandafter\def\csname LTw\endcsname{\color{white}}%
      \expandafter\def\csname LTb\endcsname{\color{black}}%
      \expandafter\def\csname LTa\endcsname{\color{black}}%
      \expandafter\def\csname LT0\endcsname{\color[rgb]{1,0,0}}%
      \expandafter\def\csname LT1\endcsname{\color[rgb]{0,1,0}}%
      \expandafter\def\csname LT2\endcsname{\color[rgb]{0,0,1}}%
      \expandafter\def\csname LT3\endcsname{\color[rgb]{1,0,1}}%
      \expandafter\def\csname LT4\endcsname{\color[rgb]{0,1,1}}%
      \expandafter\def\csname LT5\endcsname{\color[rgb]{1,1,0}}%
      \expandafter\def\csname LT6\endcsname{\color[rgb]{0,0,0}}%
      \expandafter\def\csname LT7\endcsname{\color[rgb]{1,0.3,0}}%
      \expandafter\def\csname LT8\endcsname{\color[rgb]{0.5,0.5,0.5}}%
    \else
      \def\colorrgb#1{\color{black}}%
      \def\colorgray#1{\color[gray]{#1}}%
      \expandafter\def\csname LTw\endcsname{\color{white}}%
      \expandafter\def\csname LTb\endcsname{\color{black}}%
      \expandafter\def\csname LTa\endcsname{\color{black}}%
      \expandafter\def\csname LT0\endcsname{\color{black}}%
      \expandafter\def\csname LT1\endcsname{\color{black}}%
      \expandafter\def\csname LT2\endcsname{\color{black}}%
      \expandafter\def\csname LT3\endcsname{\color{black}}%
      \expandafter\def\csname LT4\endcsname{\color{black}}%
      \expandafter\def\csname LT5\endcsname{\color{black}}%
      \expandafter\def\csname LT6\endcsname{\color{black}}%
      \expandafter\def\csname LT7\endcsname{\color{black}}%
      \expandafter\def\csname LT8\endcsname{\color{black}}%
    \fi
  \fi
    \setlength{\unitlength}{0.0500bp}%
    \ifx\gptboxheight\undefined%
      \newlength{\gptboxheight}%
      \newlength{\gptboxwidth}%
      \newsavebox{\gptboxtext}%
    \fi%
    \setlength{\fboxrule}{0.5pt}%
    \setlength{\fboxsep}{1pt}%
    \definecolor{tbcol}{rgb}{1,1,1}%
\begin{picture}(1980.00,1700.00)%
    \gplgaddtomacro\gplbacktext{%
      \csname LTb\endcsname
      \put(566,627){\makebox(0,0)[r]{\strut{}0}}%
      \csname LTb\endcsname
      \put(566,1483){\makebox(0,0)[r]{\strut{}6.8}}%
      \csname LTb\endcsname
      \put(649,431){\makebox(0,0){\strut{}0}}%
      \csname LTb\endcsname
      \put(1181,431){\makebox(0,0){\strut{}0.5}}%
      \csname LTb\endcsname
      \put(1712,431){\makebox(0,0){\strut{}1}}%
    }%
    \gplgaddtomacro\gplfronttext{%
      \csname LTb\endcsname
      \put(159,1055){\rotatebox{-270.00}{\makebox(0,0){\strut{}\normalsize Density}}}%
      \csname LTb\endcsname
      \put(1181,137){\makebox(0,0){\strut{}\normalsize Similarity}}%
    }%
    \gplbacktext
    \put(0,0){\includegraphics[width={99.00bp},height={85.00bp}]{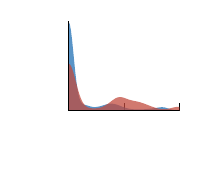}}%
    \gplfronttext
  \end{picture}%
\endgroup

%% file: density.human.l8d3.h8.dedup.tex
\begingroup
  \makeatletter
  \providecommand\color[2][]{%
    \GenericError{(gnuplot) \space\space\space\@spaces}{%
      Package color not loaded in conjunction with
      terminal option `colourtext'%
    }{See the gnuplot documentation for explanation.%
    }{Either use 'blacktext' in gnuplot or load the package
      color.sty in LaTeX.}%
    \renewcommand\color[2][]{}%
  }%
  \providecommand\includegraphics[2][]{%
    \GenericError{(gnuplot) \space\space\space\@spaces}{%
      Package graphicx or graphics not loaded%
    }{See the gnuplot documentation for explanation.%
    }{The gnuplot epslatex terminal needs graphicx.sty or graphics.sty.}%
    \renewcommand\includegraphics[2][]{}%
  }%
  \providecommand\rotatebox[2]{#2}%
  \@ifundefined{ifGPcolor}{%
    \newif\ifGPcolor
    \GPcolortrue
  }{}%
  \@ifundefined{ifGPblacktext}{%
    \newif\ifGPblacktext
    \GPblacktexttrue
  }{}%
  \let\gplgaddtomacro\g@addto@macro
  \gdef\gplbacktext{}%
  \gdef\gplfronttext{}%
  \makeatother
  \ifGPblacktext
    \def\colorrgb#1{}%
    \def\colorgray#1{}%
  \else
    \ifGPcolor
      \def\colorrgb#1{\color[rgb]{#1}}%
      \def\colorgray#1{\color[gray]{#1}}%
      \expandafter\def\csname LTw\endcsname{\color{white}}%
      \expandafter\def\csname LTb\endcsname{\color{black}}%
      \expandafter\def\csname LTa\endcsname{\color{black}}%
      \expandafter\def\csname LT0\endcsname{\color[rgb]{1,0,0}}%
      \expandafter\def\csname LT1\endcsname{\color[rgb]{0,1,0}}%
      \expandafter\def\csname LT2\endcsname{\color[rgb]{0,0,1}}%
      \expandafter\def\csname LT3\endcsname{\color[rgb]{1,0,1}}%
      \expandafter\def\csname LT4\endcsname{\color[rgb]{0,1,1}}%
      \expandafter\def\csname LT5\endcsname{\color[rgb]{1,1,0}}%
      \expandafter\def\csname LT6\endcsname{\color[rgb]{0,0,0}}%
      \expandafter\def\csname LT7\endcsname{\color[rgb]{1,0.3,0}}%
      \expandafter\def\csname LT8\endcsname{\color[rgb]{0.5,0.5,0.5}}%
    \else
      \def\colorrgb#1{\color{black}}%
      \def\colorgray#1{\color[gray]{#1}}%
      \expandafter\def\csname LTw\endcsname{\color{white}}%
      \expandafter\def\csname LTb\endcsname{\color{black}}%
      \expandafter\def\csname LTa\endcsname{\color{black}}%
      \expandafter\def\csname LT0\endcsname{\color{black}}%
      \expandafter\def\csname LT1\endcsname{\color{black}}%
      \expandafter\def\csname LT2\endcsname{\color{black}}%
      \expandafter\def\csname LT3\endcsname{\color{black}}%
      \expandafter\def\csname LT4\endcsname{\color{black}}%
      \expandafter\def\csname LT5\endcsname{\color{black}}%
      \expandafter\def\csname LT6\endcsname{\color{black}}%
      \expandafter\def\csname LT7\endcsname{\color{black}}%
      \expandafter\def\csname LT8\endcsname{\color{black}}%
    \fi
  \fi
    \setlength{\unitlength}{0.0500bp}%
    \ifx\gptboxheight\undefined%
      \newlength{\gptboxheight}%
      \newlength{\gptboxwidth}%
      \newsavebox{\gptboxtext}%
    \fi%
    \setlength{\fboxrule}{0.5pt}%
    \setlength{\fboxsep}{1pt}%
    \definecolor{tbcol}{rgb}{1,1,1}%
\begin{picture}(1980.00,1700.00)%
    \gplgaddtomacro\gplbacktext{%
      \csname LTb\endcsname
      \put(566,627){\makebox(0,0)[r]{\strut{}0}}%
      \csname LTb\endcsname
      \put(566,1483){\makebox(0,0)[r]{\strut{}2.8}}%
      \csname LTb\endcsname
      \put(649,431){\makebox(0,0){\strut{}0}}%
      \csname LTb\endcsname
      \put(1181,431){\makebox(0,0){\strut{}0.5}}%
      \csname LTb\endcsname
      \put(1712,431){\makebox(0,0){\strut{}1}}%
    }%
    \gplgaddtomacro\gplfronttext{%
      \csname LTb\endcsname
      \put(159,1055){\rotatebox{-270.00}{\makebox(0,0){\strut{}\normalsize Density}}}%
      \csname LTb\endcsname
      \put(1181,137){\makebox(0,0){\strut{}\normalsize Similarity}}%
    }%
    \gplbacktext
    \put(0,0){\includegraphics[width={99.00bp},height={85.00bp}]{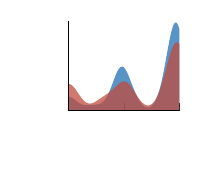}}%
    \gplfronttext
  \end{picture}%
\endgroup

%% file: density.yeast_aj.l8d3.h8.dedup.tex
\begingroup
  \makeatletter
  \providecommand\color[2][]{%
    \GenericError{(gnuplot) \space\space\space\@spaces}{%
      Package color not loaded in conjunction with
      terminal option `colourtext'%
    }{See the gnuplot documentation for explanation.%
    }{Either use 'blacktext' in gnuplot or load the package
      color.sty in LaTeX.}%
    \renewcommand\color[2][]{}%
  }%
  \providecommand\includegraphics[2][]{%
    \GenericError{(gnuplot) \space\space\space\@spaces}{%
      Package graphicx or graphics not loaded%
    }{See the gnuplot documentation for explanation.%
    }{The gnuplot epslatex terminal needs graphicx.sty or graphics.sty.}%
    \renewcommand\includegraphics[2][]{}%
  }%
  \providecommand\rotatebox[2]{#2}%
  \@ifundefined{ifGPcolor}{%
    \newif\ifGPcolor
    \GPcolortrue
  }{}%
  \@ifundefined{ifGPblacktext}{%
    \newif\ifGPblacktext
    \GPblacktexttrue
  }{}%
  \let\gplgaddtomacro\g@addto@macro
  \gdef\gplbacktext{}%
  \gdef\gplfronttext{}%
  \makeatother
  \ifGPblacktext
    \def\colorrgb#1{}%
    \def\colorgray#1{}%
  \else
    \ifGPcolor
      \def\colorrgb#1{\color[rgb]{#1}}%
      \def\colorgray#1{\color[gray]{#1}}%
      \expandafter\def\csname LTw\endcsname{\color{white}}%
      \expandafter\def\csname LTb\endcsname{\color{black}}%
      \expandafter\def\csname LTa\endcsname{\color{black}}%
      \expandafter\def\csname LT0\endcsname{\color[rgb]{1,0,0}}%
      \expandafter\def\csname LT1\endcsname{\color[rgb]{0,1,0}}%
      \expandafter\def\csname LT2\endcsname{\color[rgb]{0,0,1}}%
      \expandafter\def\csname LT3\endcsname{\color[rgb]{1,0,1}}%
      \expandafter\def\csname LT4\endcsname{\color[rgb]{0,1,1}}%
      \expandafter\def\csname LT5\endcsname{\color[rgb]{1,1,0}}%
      \expandafter\def\csname LT6\endcsname{\color[rgb]{0,0,0}}%
      \expandafter\def\csname LT7\endcsname{\color[rgb]{1,0.3,0}}%
      \expandafter\def\csname LT8\endcsname{\color[rgb]{0.5,0.5,0.5}}%
    \else
      \def\colorrgb#1{\color{black}}%
      \def\colorgray#1{\color[gray]{#1}}%
      \expandafter\def\csname LTw\endcsname{\color{white}}%
      \expandafter\def\csname LTb\endcsname{\color{black}}%
      \expandafter\def\csname LTa\endcsname{\color{black}}%
      \expandafter\def\csname LT0\endcsname{\color{black}}%
      \expandafter\def\csname LT1\endcsname{\color{black}}%
      \expandafter\def\csname LT2\endcsname{\color{black}}%
      \expandafter\def\csname LT3\endcsname{\color{black}}%
      \expandafter\def\csname LT4\endcsname{\color{black}}%
      \expandafter\def\csname LT5\endcsname{\color{black}}%
      \expandafter\def\csname LT6\endcsname{\color{black}}%
      \expandafter\def\csname LT7\endcsname{\color{black}}%
      \expandafter\def\csname LT8\endcsname{\color{black}}%
    \fi
  \fi
    \setlength{\unitlength}{0.0500bp}%
    \ifx\gptboxheight\undefined%
      \newlength{\gptboxheight}%
      \newlength{\gptboxwidth}%
      \newsavebox{\gptboxtext}%
    \fi%
    \setlength{\fboxrule}{0.5pt}%
    \setlength{\fboxsep}{1pt}%
    \definecolor{tbcol}{rgb}{1,1,1}%
\begin{picture}(1980.00,1700.00)%
    \gplgaddtomacro\gplbacktext{%
      \csname LTb\endcsname
      \put(566,627){\makebox(0,0)[r]{\strut{}0}}%
      \csname LTb\endcsname
      \put(566,1483){\makebox(0,0)[r]{\strut{}5.9}}%
      \csname LTb\endcsname
      \put(649,431){\makebox(0,0){\strut{}0}}%
      \csname LTb\endcsname
      \put(1181,431){\makebox(0,0){\strut{}0.5}}%
      \csname LTb\endcsname
      \put(1712,431){\makebox(0,0){\strut{}1}}%
    }%
    \gplgaddtomacro\gplfronttext{%
      \csname LTb\endcsname
      \put(159,1055){\rotatebox{-270.00}{\makebox(0,0){\strut{}\normalsize Density}}}%
      \csname LTb\endcsname
      \put(1181,137){\makebox(0,0){\strut{}\normalsize Similarity}}%
    }%
    \gplbacktext
    \put(0,0){\includegraphics[width={99.00bp},height={85.00bp}]{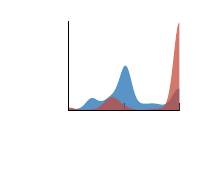}}%
    \gplfronttext
  \end{picture}%
\endgroup

%% file: density.human_aj.l8d3.h8.dedup.tex
\begingroup
  \makeatletter
  \providecommand\color[2][]{%
    \GenericError{(gnuplot) \space\space\space\@spaces}{%
      Package color not loaded in conjunction with
      terminal option `colourtext'%
    }{See the gnuplot documentation for explanation.%
    }{Either use 'blacktext' in gnuplot or load the package
      color.sty in LaTeX.}%
    \renewcommand\color[2][]{}%
  }%
  \providecommand\includegraphics[2][]{%
    \GenericError{(gnuplot) \space\space\space\@spaces}{%
      Package graphicx or graphics not loaded%
    }{See the gnuplot documentation for explanation.%
    }{The gnuplot epslatex terminal needs graphicx.sty or graphics.sty.}%
    \renewcommand\includegraphics[2][]{}%
  }%
  \providecommand\rotatebox[2]{#2}%
  \@ifundefined{ifGPcolor}{%
    \newif\ifGPcolor
    \GPcolortrue
  }{}%
  \@ifundefined{ifGPblacktext}{%
    \newif\ifGPblacktext
    \GPblacktexttrue
  }{}%
  \let\gplgaddtomacro\g@addto@macro
  \gdef\gplbacktext{}%
  \gdef\gplfronttext{}%
  \makeatother
  \ifGPblacktext
    \def\colorrgb#1{}%
    \def\colorgray#1{}%
  \else
    \ifGPcolor
      \def\colorrgb#1{\color[rgb]{#1}}%
      \def\colorgray#1{\color[gray]{#1}}%
      \expandafter\def\csname LTw\endcsname{\color{white}}%
      \expandafter\def\csname LTb\endcsname{\color{black}}%
      \expandafter\def\csname LTa\endcsname{\color{black}}%
      \expandafter\def\csname LT0\endcsname{\color[rgb]{1,0,0}}%
      \expandafter\def\csname LT1\endcsname{\color[rgb]{0,1,0}}%
      \expandafter\def\csname LT2\endcsname{\color[rgb]{0,0,1}}%
      \expandafter\def\csname LT3\endcsname{\color[rgb]{1,0,1}}%
      \expandafter\def\csname LT4\endcsname{\color[rgb]{0,1,1}}%
      \expandafter\def\csname LT5\endcsname{\color[rgb]{1,1,0}}%
      \expandafter\def\csname LT6\endcsname{\color[rgb]{0,0,0}}%
      \expandafter\def\csname LT7\endcsname{\color[rgb]{1,0.3,0}}%
      \expandafter\def\csname LT8\endcsname{\color[rgb]{0.5,0.5,0.5}}%
    \else
      \def\colorrgb#1{\color{black}}%
      \def\colorgray#1{\color[gray]{#1}}%
      \expandafter\def\csname LTw\endcsname{\color{white}}%
      \expandafter\def\csname LTb\endcsname{\color{black}}%
      \expandafter\def\csname LTa\endcsname{\color{black}}%
      \expandafter\def\csname LT0\endcsname{\color{black}}%
      \expandafter\def\csname LT1\endcsname{\color{black}}%
      \expandafter\def\csname LT2\endcsname{\color{black}}%
      \expandafter\def\csname LT3\endcsname{\color{black}}%
      \expandafter\def\csname LT4\endcsname{\color{black}}%
      \expandafter\def\csname LT5\endcsname{\color{black}}%
      \expandafter\def\csname LT6\endcsname{\color{black}}%
      \expandafter\def\csname LT7\endcsname{\color{black}}%
      \expandafter\def\csname LT8\endcsname{\color{black}}%
    \fi
  \fi
    \setlength{\unitlength}{0.0500bp}%
    \ifx\gptboxheight\undefined%
      \newlength{\gptboxheight}%
      \newlength{\gptboxwidth}%
      \newsavebox{\gptboxtext}%
    \fi%
    \setlength{\fboxrule}{0.5pt}%
    \setlength{\fboxsep}{1pt}%
    \definecolor{tbcol}{rgb}{1,1,1}%
\begin{picture}(1980.00,1700.00)%
    \gplgaddtomacro\gplbacktext{%
      \csname LTb\endcsname
      \put(566,627){\makebox(0,0)[r]{\strut{}0}}%
      \csname LTb\endcsname
      \put(566,1483){\makebox(0,0)[r]{\strut{}2.9}}%
      \csname LTb\endcsname
      \put(649,431){\makebox(0,0){\strut{}0}}%
      \csname LTb\endcsname
      \put(1181,431){\makebox(0,0){\strut{}0.5}}%
      \csname LTb\endcsname
      \put(1712,431){\makebox(0,0){\strut{}1}}%
    }%
    \gplgaddtomacro\gplfronttext{%
      \csname LTb\endcsname
      \put(159,1055){\rotatebox{-270.00}{\makebox(0,0){\strut{}\normalsize Density}}}%
      \csname LTb\endcsname
      \put(1181,137){\makebox(0,0){\strut{}\normalsize Similarity}}%
    }%
    \gplbacktext
    \put(0,0){\includegraphics[width={99.00bp},height={85.00bp}]{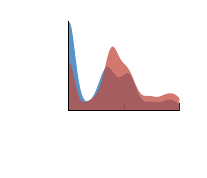}}%
    \gplfronttext
  \end{picture}%
\endgroup